\documentclass[%
reprint,
amsmath,amssymb,
aps,
]{revtex4-1}

\usepackage{graphicx}
\usepackage{dcolumn}
\usepackage{bm}

\usepackage{subfigure}  

\usepackage{stmaryrd} 
\usepackage{amsthm}  
\usepackage{IEEEtrantools} 
\usepackage{dsfont} 

\usepackage{mathtools} 

\newtheorem{lem}{Lemma}
\newtheorem{theorem}{Theorem}

\newcommand{\Id}[1]{\boldsymbol{\mathbb{I #1}}}

\begin{document}

\preprint{APS/123-QED}

\title{Shadow poles in the alternative parametrization of $R$-matrix theory}


\author{Pablo Ducru}
\email{p\_ducru@mit.edu ; pablo.ducru@polytechnique.org}
\altaffiliation[]{Also from \'Ecole Polytechnique, France. \& Schwarzman Scholars, Tsinghua University, China.}
\affiliation{%
Massachusetts Institute of Technology\\
Department of Nuclear Science \& Engineering\\
77 Massachusetts Avenue, Cambridge, MA, 02139 U.S.A.\\
}%

\author{Vladimir Sobes}%
\email{sobesv@utk.edu}
\affiliation{%
University of Tennessee\\
Department of Nuclear Engineering\\
1412 Circle Drive, Knoxville, TN, 37996, U.S.A.
}%

\author{Gerald Hale}
\email{ghale@lanl.gov}
\author{Mark Paris}%
\email{mparis@lanl.gov}
\affiliation{%
Los Alamos National Laboratory\\
Theoretical Division (T-2)\\
MS B283, Los Alamos, NM, 87454, U.S.A.\\
}%

\author{Benoit Forget}
\email{bforget@mit.edu}
\affiliation{%
Massachusetts Institute of Technology\\
Department of Nuclear Science \& Engineering\\
77 Massachusetts Avenue, Cambridge, MA, 02139 U.S.A.\\
}%


%


\date{\today}

\begin{abstract}

We discover new, hitherto unknown, shadow poles in Brune's alternative parametrization of R-matrix theory [C. R. Brune, Phys. Rev. C 66, 044611 (2002)].
Where these poles are, and how many, depends on how one continues R-matrix operators to complex wavenumbers (specially the shift $S$ and penetration $P$ functions). 
This has little consequence for the exact R-matrix formalism (past the last energy threshold), as we show one can still always fully reconstruct the scattering matrix with only the previously known alternative parameters (poles and corresponding resonance widths), for which there were as many poles as the number of levels $N_\lambda$. 
However, we generalize the alternative parametrization to the Reich-Moore formalism, and show that the choice of continuation is now critical as it changes the alternative parameters values (poles and residue widths are now complex). 
In order to establish nuclear libraries with alternative parameters, the nuclear community will thus have to decide what convention to adopt. We argue in favor of analytical continuation (against the legacy Lane and Thomas approach) in a follow-up article [P. Ducru, Phys. Rev. C, submitted (2020)]. 
We observe the first evidence of shadow poles in the alternative parametrization of R-matrix theory in isotope xenon $^{\mathrm{134}}\mathrm{Xe}$ spin-parity group $J^\pi = 1/2^{(-)}$, and show how they indeed depend on the choice of continuation to complex wavenumbers.

\end{abstract}

\maketitle


\section{\label{sec:Introduction}Introduction}

When two nuclear bodies collide at a given energy -- say a neutron and an uranium-235 nucleus ($\mathrm{n} + \mathrm{^{235} _{\; \; 92} U}$), a $\gamma$ particle (photon) and a beryllium atom ($\mathrm{\gamma} + \mathrm{^{9} _4 Be}$), or an alpha particle ($\mathrm{^{4} _2 He}$) and a gold atom ($\mathrm{\alpha} + {^{197} _{\; \,  79} \mathrm{Au}} $) -- the outcomes of this interaction are expressed as nuclear cross sections.
These cross sections are a fundamental component of our nuclear physics knowledge, documented in standard nuclear data libraries (ENDF \cite{ENDFBVIII8th2018brown}, JEFF\cite{JEFF_2020_plompenJointEvaluatedFission2020}, JENDL\cite{JENDL_shibataJENDL4NewLibrary2011}, BROND\cite{BROND_2016}, CENDL\cite{CENDLProjectChinese2017}, TENDL\cite{TENDLkoningModernNuclearData2012, koningTENDLCompleteNuclear2019}). 
To constitute nuclear data libraries, an evaluation process fits experimental measurements of reaction rates with a parametric model of nuclear interaction cross sections called R-matrix theory, using evaluation codes such as EDA \cite{EDA_2008, EDA_2015}, SAMMY \cite{SAMMY_2008}, or AZURE \cite{azumaAZUREMatrixCode2010}.
R-matrix theory models nuclear interactions as two incoming bodies yielding two outgoing bodies through the action of a total Hamiltonian. The latter is assumed to be the addition of a short-range, interior Hamiltonian that is null beyond channel radius $a_c$, and a long-range, exterior Hamiltonian that we know, say Coulomb potential or free moving.
This partitioning, along with an orthogonality assumption of channels at the channel boundary, is what we could call the \textit{R-matrix scattering model}, described by Kapur and Peierls in their seminal article \cite{Kapur_and_Peierls_1938}, unified by Bloch in \cite{Bloch_1957}, and reviewed by Lane and Thomas in \cite{Lane_and_Thomas_1958}.
The outcomes of the interaction depend on the energy at which the interaction occurs, and R-matrix theory parametrizes, for calculability reasons, this energy dependence. It can do so in several ways: the one that has come to prevail in the nuclear physics community is the Wigner-Eisenbud parametrization \cite{Wigner_and_Eisenbud_1947, Bloch_1957, Lane_and_Thomas_1958}.

There are good reasons for this: the Wigner-Eisenbud parameters are unconstrained real parameters --- i.e. though physically and statistically correlated, any set of real parameters is mathematically acceptable (though not necessarily present in nature) ---
that parametrize the interior interaction Hamiltonian (usually an intractable many-body nuclear problem) and separate it from the exterior one (usually a well-known free-body or Coulomb Hamiltonian with analytic Harmonic expansions).
Thus, Wigner and Eisenbud constructed a parametrization of the scattering matrix for calculability purposes: introducing simple real parameters that help de-correlate what happens in the inner interaction region from the asymptotic outer region.
Despite all their advantages, the Wigner-Eisenbud parameters present a drawback for nuclear data evaluators: they require the introduction, for every channel $c$, of an arbitrary real ``boundary condition'' parameter, $B_c$. If this arbitrary parameter is set to different values, the same experimental nuclear data will yield different Wigner-Eisenbud resonance parameters. This poses both a physics interpretability problem, and a standardization problem when edifying the standard nuclear data libraries.

In order to circumvent the need for arbitrary boundary parameters $B_c$, Brune introduced an alternative parametrization of R-matrix theory in \cite{Brune_2002}. The alternative parameters are real (like the Wigner-Eisenbud ones) and are independent of the arbitrary boundary condition parameters $B_c$. However, they do entangle the interior region (function of the total energy $E$) with the outer region (function of the incoming wavenumber $k_c$ and outgoing wavenumber $k_{c'}$), so that the alternative parameters depend on how the wavenumbers are related to the energy of the system, $k_c(E)$, and this mapping has branch-points and different sheets corresponding to all the choices of sign in the square roots $\pm \sqrt{E - E_{T_c}}$ of mapping (\ref{eq:rho_c massive}). 
Using monotonic properties of the shift function $S(E)$ on the physical sheet $\big\{ E, + \big\}$ of energy-wavenumber (\ref{eq:rho_c massive}) mapping (recently proved in \cite{Brune_Mark_monotonic_properties_of_shift_2018}), Brune showed a one-to-one correspondence between the number $N_\lambda$ of resonances (or levels) and the number of alternative resonance energies (or poles) \cite{Brune_2002}. This would make the conversion of nuclear data libraries from Wigner-Eisenbud to alternative parameters very convenient.

Section \ref{sec:R-matrix equations} summarizes the Wigner-Eisenbud R-matrix parametrization, reports on the branch-point nature of the energy-wavenumber mapping (\ref{eq:rho_c massive}) -- and its simple relativistic generalisation (\ref{eq:rho_c(E) mapping}) -- and, for the first time, establishes in theorem \ref{theo::Mittag-Leffler of L_c theorem} the Mittag-Leffler expansion of the reduced logarithmic derivative of the outgoing wave operator $L_c(\rho_c)$.
These results are used in section \ref{sec:R_S Brune transform} to show there exists more alternative poles than previously thought: these \textit{shadow poles} reside below the reaction threshold energies $E_{T_c}$.
We also show that these alternative shadow poles depend on the definition that is chosen to continue the R-matrix operators to complex wavenumbers. If the legacy Lane \& Thomas definition (\ref{eq:: Def S = Re[L], P = Im[L]}) is chosen, then we call them \textit{alternative branch poles} and establish their properties in theorem \ref{theo::branch_brune_poles}, amongst which that the shadow poles reside on the nonphysical sheet $\big\{ E, - \big\}$ sub-threshold.
If, instead, the analytic continuation definition (\ref{eq:: Def S analytic}) is chosen, then we call them \textit{alternative analytic poles}, and we establish their properties in theorem \ref{theo::analytic_Brune_poles}, in particular we show alternative analytic poles are in general complex, of which there exists at least $ N_\lambda$ real ones. Moreover, and similarly to the Wigner-Eisenbud parameters, alternative analytic poles only depend on the total energy $E$ and thus no longer present the branches of mapping (\ref{eq:rho_c(E) mapping}).
In theorem \ref{theo::Choice of Brune poles}, we also show that, under a proper generalization of the alternative level matrix, one can choose any subset of $N_S$ alternative poles (for both definitions and real or complex) and still fully reconstruct the scattering matrix (and thus the cross section), as long as $N_S \geq N_\lambda$. 

In nuclear libraries, many isotopes are evaluated with the Reich-Moore formalism instead of the full R-matrix one. 
In section \ref{sec:Generalized Brune parameters for Reich-Moore approximation}, we generalize the alternative parametrization of R-matrix theory to the Reich-Moore formalism, including the newly discovered alternative shadow poles. 
The first evidence of alternative shadow poles is observed in isotope xenon $^{\mathrm{134}}\mathrm{Xe}$ spin-parity group $J^{\pi}=1/2^{(-)}$, and reported in section \ref{sec:Evidence of shadow Brune poles in xenon 134}. 
We also demonstrate how in practice (for Reich-Moore isotopes or when thresholds are present) all alternative parameters depend on the choice of continuation of R-matrix operators to complex wavenumbers. This means that in order to convert nuclear data libraries to alternative parameters, the nuclear physics community must first agree on how to continue the R-matrix operators to complex wavenumbers. We argue in favor of analytic continuation in a follow-up article \cite{Ducru_Scattering_Matrix_of_Complex_Wavenumbers_2019}.

\section{\label{sec:R-matrix equations}R-matrix Wigner-Eisenbud parametrization}

We here recall some fundamental definitions and equations of the Wigner-Eisenbud R-matrix parameters \cite{Wigner_and_Eisenbud_1947, Bloch_1957, Lane_and_Thomas_1958}.
As described by Bloch and Lane \& Thomas, for each channel $c$, R-matrix theory treats the two-body-in/two-body-out many-body system as a reduced one-body system.
All the study is then performed in the reduced system and we consider the wave-number of each channel $k_c$, which we can render dimensionless using the channel radius $a_c$ and defining $\boldsymbol{\rho}=\boldsymbol{\mathrm{diag}}\left(\rho_c\right)$ with $\rho_c = k_c a_c$.

\subsection{\label{subsec:Energy dependence and wavenumber mapping} Energy dependence and wavenumber mapping}

All of the channel wavenumbers link back to one unique total system energy $E$, eigenvalue of the total Hamiltonian.
Conservation of energy entails that this energy $E$ must be the total energy of any given channel $c$ (c.f. equation (5.12), p.557 of \cite{Theory_of_Nuclear_Reactions_I_resonances_Humblet_and_Rosenfeld_1961}):
\begin{equation}
\begin{IEEEeqnarraybox}[][c]{rcl}
E  & \ = \ & E_c = E_{c'} = \hdots \; \; , \; \forall \; c
\IEEEstrut\end{IEEEeqnarraybox}
\label{eq:conservation of energy E = E_c = E_c'}
\end{equation}
Each channel's total energy $E_c$ is then linked to the wavenumber $k_c$ of the channel by its corresponding relation (\ref{eq:rho_c(E) mapping}), say (\ref{eq:rho_c EDA}) and (\ref{eq:E_c as a function of s_c}).

In the semi-classical model described in Lane \& Thomas \cite{Lane_and_Thomas_1958}, we can separate on the one hand massive particles, for which the wavenumber $k_c$ is related to the center-of-mass energy $E_c$ of relative motion of channel $c$ particle pair with masses $m_{c,1}$ and $ m_{c,2}$ as
\begin{equation}
\begin{IEEEeqnarraybox}[][c]{rcl}
k_c  & \ = \ & \sqrt{\frac{2m_{c,1} m_{c,2}}{\left(m_{c,1}+m_{c,2}\right) \mathrm{\hbar}^2} \left(E_c - E_{T_c}\right)}
\IEEEstrut\end{IEEEeqnarraybox}
\label{eq:rho_c massive}
\end{equation}
where $E_{T_c}$ denotes a threshold energy beyond which the channel $c$ is closed, as energy conservation cannot be respected ($E_{T_c} = 0 $ for reactions without threshold).
On the other hand, for a photon particle interacting with a massive body of mass $m_{c,1}$ the center-of-mass wavenumber $k_c$ is linked to the total center-of-mass energy $E_c$ of channel $c$ according to:
\begin{equation}
k_c = \frac{\left( E_c - E_{T_c} \right)}{2 \mathrm{\hbar}\mathrm{c}}\left[ 1 + \frac{m_{c,1} \mathrm{c}^2}{\left( E - E_{T_c} \right) + m_{c,1}\mathrm{c}^2}\right]
\label{eq:rho_c photon}
\end{equation}

Alternatively, in a more unified approach, one can perform a relativistic correction and smooth these differences away by means of the special relativity Mandelstam variable $s_c = (p_{c,1} + p_{c,2})$, also known as the square of the center-of-mass energy, where $p_{c,1}$ and $p_{c,2}$ are the Minkowsky metric four-momenta of the two bodies composing channel $c$, with respective masses $m_{c,1}$ and $ m_{c,2}$ (null for photons). The channel wavenumber $k_c$ can then be expressed as:
\begin{equation}
k_c =  \sqrt{\frac{\left[ s_c - (m_{c,1} + m_{c,2})^2\mathrm{c}^2\right]\left[ s_c - (m_{c,1} - m_{c,2})^2\mathrm{c}^2\right]}{4 \mathrm{\hbar}^2s_c}}
\label{eq:rho_c EDA}
\end{equation}
and the Mandelstam variable $s_c$ can be linked to the center-of-mass energy of the channel $E_c$ through
\begin{equation}
E_c = \frac{ s_c - (m_{c,1} + m_{c,2})^2\mathrm{c}^2}{2 (m_{c,1} + m_{c,2})}
\label{eq:E_c as a function of s_c}
\end{equation}
Interestingly, this is identical to the non-relativistic expression for the center-of-mass energy in terms of the lab energy in whichever channel the total mass $(m_{c,1} + m_{c,2})$ is chosen to be the reference for $E$ (but not in any other).
This special relativistic correction to the non-relativistic R-matrix theory is the approach taken by the EDA code in use at the Los Alamos National Laboratory \cite{EDA_2008, EDA_2015}.

Regardless of the approach taken to link the channel energy $E_c$ to the channel wavenumber $k_c$, conservation of energy (\ref{eq:conservation of energy E = E_c = E_c'}) entails there exists a complex mapping linking the total center-of-mass energy $E$ to the wavenumbers $k_c$, or their associated dimensionless variable $\rho_c= k_c r_c$:
\begin{equation}
\rho_c(E)  \quad \longleftrightarrow \quad E
\label{eq:rho_c(E) mapping}
\end{equation}

\begin{figure}[ht!!] 
  \centering
  \includegraphics[width=0.5\textwidth]{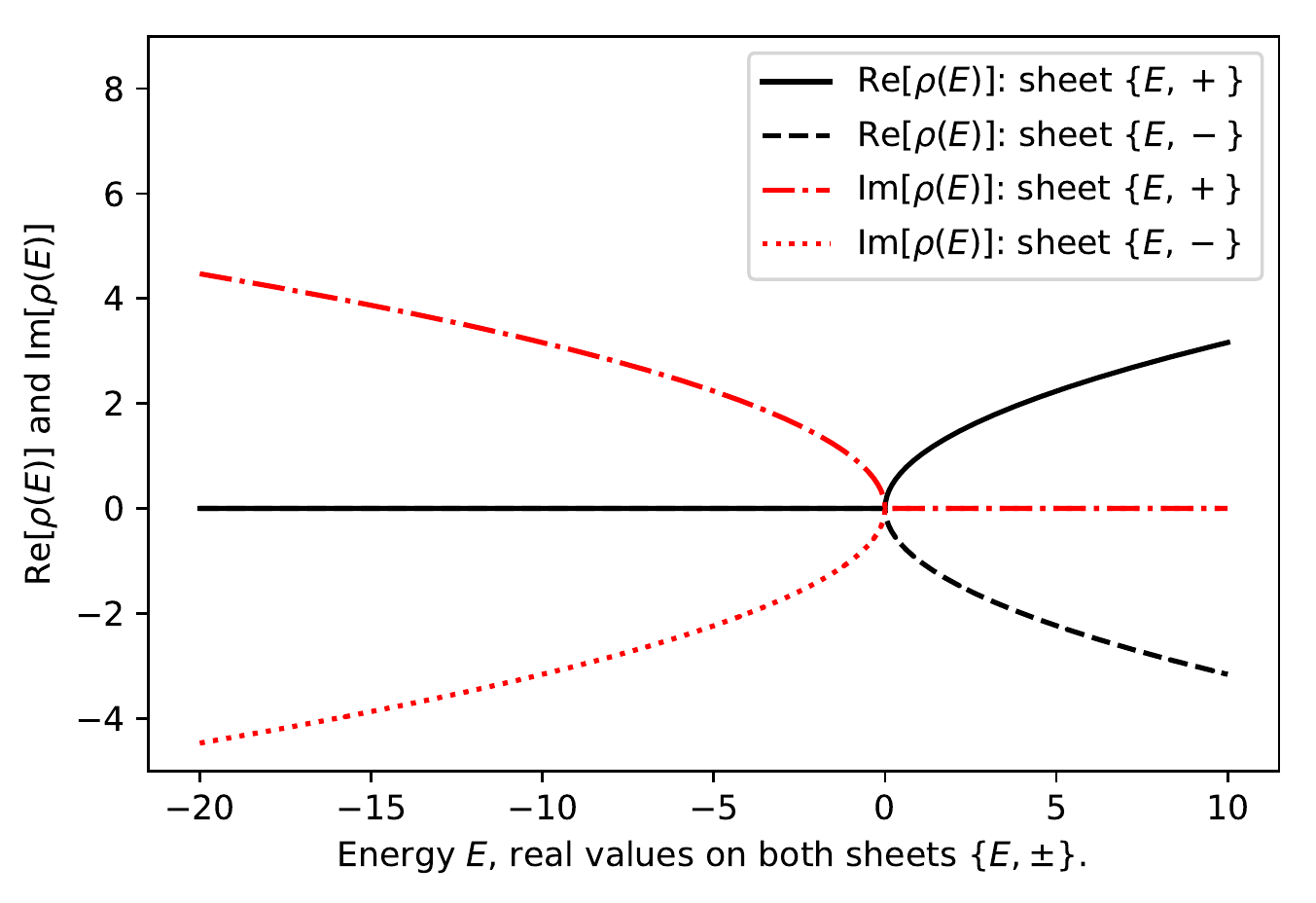}
  \caption{\small{Riemann surface of energy-wavenumber mapping $\rho(E)$ for massive particles in the semi-classical limit (\ref{eq:rho_c massive}). The square root $\rho_c(E) = \pm \rho_0\sqrt{E - E_{T_c}} $ gives rise to two sheets: $\left\{ E,+\right\}$ and $\left\{ E,-\right\}$. Units such that $\rho_0 = 1$. Threshold set at zero: $E_{T_c} = 0$.}}
  \label{fig:mapping rho - E}
\end{figure}

Critical properties throughout this article will stem from the analytic continuation of R-matrix operators.
As the outgoing $O_c$ and incoming $I_c$ wave functions are defined according to $\rho_c$ (c.f. section \ref{subsec:External_region_waves} below), the natural variable to perform analytic continuation is thus $\rho_c$, which is equivalent to extending the wavenumbers into the complex plane $k_c \in \mathbb{C}$.
We can see that the mapping (\ref{eq:rho_c(E) mapping}) from complex $k_c$ to complex energies is non-trivial, specially since the wavenumbers are themselves all interconnected.
This creates a multi-sheeted Riemann surface, with branch-points at each threshold $E_{T_c}$, well documented by Eden \& Taylor \cite{Eden_and_Taylor} (also c.f. section 8 of \cite{Theory_of_Nuclear_Reactions_I_resonances_Humblet_and_Rosenfeld_1961}). 
More precisely, when calculating $\rho_c$ from $E$ one has to chose which sign to assign to $\pm \sqrt{E - E_{T_c}}$ in the semi-classical mapping (\ref{eq:rho_c massive}) of massive particles (i.e. not photons), or to the more general mapping (\ref{eq:rho_c EDA}). Figure \ref{fig:mapping rho - E} shows mapping (\ref{eq:rho_c massive}) with zero threshold $E_{T_c} = 0$: one can see that $\rho_c(E)$ has two branches $\rho_c(E) = \pm \rho_0 \sqrt{E - E_{T_c}}$ and is purely real above threshold (zero imaginary part), and purely imaginary below threshold (zero real part).
Each channel $c$ thus introduces two choices, and hence there are $2^{N_c}$ sheets to the Riemann surface mapping (\ref{eq:conservation of energy E = E_c = E_c'}) onto (\ref{eq:rho_c(E) mapping}), with the branch points close or equal to the threshold energies $E_{T_c}$. As we will see, the choice of the sheet will have an impact when finding different R-matrix and alternative parameters.

\subsection{\label{subsec:External_region_waves} External region wave functions}

\begin{table*}
\caption{\label{tab::L_values_neutral} Reduced logarithmic derivative $L_\ell(\rho) \triangleq \frac{\rho}{O_\ell} \frac{\partial O_\ell}{\partial r}(\rho)$ of outgoing wavefunction $O_\ell(\rho)$, and $L_\ell^0(\rho) \triangleq L_\ell(\rho) - B_\ell $ using $B_\ell = - \ell$, irreducible forms and Mittag-Leffler pole expansions for neutral particles, for angular momenta $0 \leq \ell \leq 4$. }
\begin{ruledtabular}
\begin{tabular}{c|c|c|c|c}
\ \ & $L_\ell(\rho)$ from recurrence (\ref{eq::L_ell recurrence formula}) & $\begin{array}{c}
     L_\ell^0(\rho) \triangleq L_\ell(\rho) - B_\ell  \\
     \text{using } B_\ell = - \ell \text{ in (\ref{eq::L_ell recurrence formula})} 
\end{array}$ & $\begin{array}{c}
    L_\ell(\rho) \text{ from theorem \ref{theo::Mittag-Leffler of L_c theorem},}  \\
    \text{poles } \big\{ \omega_n\big\} \text{ from table \ref{tab::roots of the outgoing wave functions}} 
\end{array} $ & $\begin{array}{c}
     \text{Outgoing wavefunction } \\
     O_{\ell}(\rho) \text{ from  (\ref{eq::Outgoing wavefuntion O_ell expression for neutral particles})}
\end{array} $ \tabularnewline
\hline
$\ell$  &  $ L_\ell(\rho) = \frac{\rho^2 }{\ell - L_{\ell-1}(\rho)} - \ell $ & $ L_\ell^0(\rho) =  \frac{\rho^2}{2\ell -1 - L_{\ell-1}^0(\rho) }$   & $L_\ell(\rho) = -\ell + \mathrm{i} \rho + \sum_{n\geq 1} \frac{\rho}{\rho - \omega_n}$ &$ O_\ell(\rho)  =  \mathrm{e}^{\mathrm{i}\left(\rho + \frac{1}{2}\ell \pi \right)}\frac{\prod_{n \geq 1}\left(\rho-\omega_n\right)}{\rho^\ell}$  \tabularnewline
\hline \hline
0  &  $\mathrm{i}\rho$  & $\mathrm{i}\rho$ & $\left\{ \emptyset \right\}$ &  $\mathrm{e}^{\mathrm{i}\rho}  $\tabularnewline
1  &  $ \frac{-1 + \mathrm{i}\rho + \rho^2}{1-\mathrm{i}\rho}$ & $ \frac{\rho^2}{1-\mathrm{i}\rho}$  &  $\omega_{1}^{\ell = 2} = -\mathrm{i} $ & $\mathrm{e}^{\mathrm{i}\rho}\left(\frac{1}{\rho} - \mathrm{i}\right)  $ \tabularnewline
2   &  $ \frac{- 6 + 6\mathrm{i}\rho + 3 \rho^2 - \mathrm{i}\rho^3 }{3 - 3\mathrm{i}\rho - \rho^2} $ & $ \frac{ \rho^2 - \mathrm{i}\rho^3 }{3 - 3\mathrm{i}\rho - \rho^2} $ & $\omega_{1,2}^{\ell = 2} \approx \pm 0.86602 - 1.5\mathrm{i} $ &  $\mathrm{e}^{\mathrm{i}\rho}\left(\frac{3}{\rho^2} -\frac{3\mathrm{i}}{\rho} - 1\right)  $ \tabularnewline
3  &  $\frac{- 45 + 45 \mathrm{i} \rho + 21\rho^2 - 6\mathrm{i}\rho^3 -\rho^4}{15 - 15\mathrm{i}\rho - 6\rho^2 + \mathrm{i}\rho^3}$ & $\frac{3\rho^2 - 3\mathrm{i}\rho^3 -\rho^4}{15 - 15\mathrm{i}\rho - 6\rho^2 + \mathrm{i}\rho^3}$ & $\begin{array}{rl}
     \omega_1^{\ell = 3} & \approx - 2.32219 \mathrm{i}  \\
     \omega_{2,3}^{\ell = 3}  & \approx \pm 1.75438 - 1.83891 \mathrm{i}  
\end{array} $ &  $\mathrm{e}^{\mathrm{i}\rho}\left(\frac{15}{\rho^3} -\frac{15\mathrm{i}}{\rho^2} -\frac{6}{\rho} + \mathrm{i}\right)  $  \tabularnewline
4  &  $\frac{ - 420 + 420\mathrm{i}\rho + 195\rho^2 -55\mathrm{i}\rho^3 - 10\rho^4 + \mathrm{i}\rho^5}{105 - 105\mathrm{i}\rho -45\rho^2 + 10\mathrm{i}\rho^3 + \rho^4}$ & $\frac{ 15\rho^2 -15\mathrm{i}\rho^3 - 6\rho^4 + \mathrm{i}\rho^5}{105 - 105\mathrm{i}\rho -45\rho^2 + 10\mathrm{i}\rho^3 + \rho^4}$  &  $ \begin{array}{rl}
   \omega_{1,2}^{\ell = 4} & \approx \pm 2.65742 - 2.10379\mathrm{i}   \\
    \omega_{3,4}^{\ell = 4} & \approx \pm 0.867234 - 2.89621 \mathrm{i} 
       \end{array}$ & $\mathrm{e}^{\mathrm{i}\rho}\left(\frac{105}{\rho^4} - \frac{105\mathrm{i}}{\rho^3} -\frac{45}{\rho^2} + \frac{10 \mathrm{i}}{\rho} + 1 \right)  $ \tabularnewline
\end{tabular}
\end{ruledtabular}
\end{table*}

In the R-matrix model, the external region is subject to either a Coulomb interaction or a free particle movement. In either case, the solutions form a two-dimensional vector space, a basis of which is composed of the incoming and outgoing wave functions: $\boldsymbol{O}(\boldsymbol{k} ) \triangleq \boldsymbol{\mathrm{diag}}\left(O_c(k_c)\right)$, $\boldsymbol{I}(\boldsymbol{k} ) \triangleq \boldsymbol{\mathrm{diag}}\left(I_c(k_c)\right)$.
These are Whittaker or confluent hypergeometric function whose analytic continuation is discussed in section II.2.b and the appendix of \cite{Lane_and_Thomas_1958}, and for whose elemental properties and calculation we refer to chapter 14 of \cite{Abramowitz_and_Stegun} and chapter 33 of \cite{NIST_DLMF}, as well as Powell \cite{Powell_1949}, Thompson \cite{Thompson_1986}, and Michel \cite{Michel_2006}.

Note that the incoming and outgoing wave functions are only dependent on the wavenumber of the given channel $k_c$, this is a fundamental hypothesis of the R-matrix model.
For clarity of writing, we will not explicitly write the $k_c$ dependence of these operators unless it is of importance for the argument.

Importantly, the Wronksian of the system is constant: $\forall c, \; w_c \triangleq O_c^{(1)}I_c - I_c^{(1)}O_c = 2\mathrm{i} $, or with identity matrix $\Id{}$
\begin{equation}
\begin{IEEEeqnarraybox}[][c]{rcl}
\boldsymbol{w} & \ \triangleq \ & \boldsymbol{O}^{(1)}\boldsymbol{I} - \boldsymbol{I}^{(1)}\boldsymbol{O}\\
& \ = \ & 2\mathrm{i} \Id{}
\IEEEstrut\end{IEEEeqnarraybox}
\label{eq:wronksian expression}
\end{equation}

Of central importance to R-matrix theory is the \textit{Bloch operator}, $\mathcal{L}$, which Claude Bloch introduced as the \textit{op\'erateur de conditions aux limites} in equation (35) of \cite{Bloch_1957}, and that projects the system radially onto the channel boundaries for each channel, at the channel radius $r_c = a_c$. The Bloch operator $\mathcal{L}$ is then added to the Hamiltonian to form an invertible and diagonalizable (but not Hermitian) operator in the internal region (c.f. equation (34) of \cite{Bloch_1957}). One can then diagonalize or invert this operator using different complete discrete generative eigenbases of the Hilbert space to construct different parameterizations of the solutions of the Schrodinger equation for the R-matrix scattering model (Kapur-Peierls, Wigner-Eisenbud, etc.). This is the essence of R-matrix theory, as best described by Claude Bloch in \cite{Bloch_1957}.

When using the Wigner-Eisenbud basis, this projection on the channel boundaries at $r_c = a_c$, gives rise to the as yet unnamed quantity $\boldsymbol{L}^0$, introduced in equation (1.6a), section VII.1. p.289 of \cite{Lane_and_Thomas_1958}, and which can be recognized in equation (57) of \cite{Bloch_1957}, defined for each channel as:
\begin{equation}
\begin{IEEEeqnarraybox}[][c]{rcl}
L_c^0(\rho_c) & \ \triangleq \ & L_c(\rho_c) - B_c
\IEEEstrut\end{IEEEeqnarraybox}
\label{eq:Bloch L expression}
\end{equation}
where $\rho_c = k_c a_c$ has been projected on the channel surface, $B_c$ is the arbitrary outgoing-wave boundary condition parameter, and $L_c(\rho_c)$ is the dimensionless reduced logarithmic derivative of the outgoing-wave function at the channel surface:
\begin{equation}
\begin{IEEEeqnarraybox}[][c]{rcl}
L_c(\rho_c) & \ \triangleq \ & \frac{\rho_c}{O_c} \frac{\partial O_c}{\partial \rho_c}
\IEEEstrut\end{IEEEeqnarraybox}
\label{eq:L expression}
\end{equation}
or, equivalently, in matrix notation, and where $\left[ \; \cdot \; \right]^{(1)}$ designates the derivative with respect to $\rho_c$:
\begin{equation}
\begin{IEEEeqnarraybox}[][c]{rcl}
\boldsymbol{L} = \boldsymbol{\mathrm{diag}}\left(L_c\right) = \boldsymbol{\rho}\boldsymbol{O}^{-1}\boldsymbol{O}^{(1)}  \IEEEstrut\end{IEEEeqnarraybox}
\label{eq:L expression matrix}
\end{equation}
so that the $\boldsymbol{L^0}$ matrix function is written: $\boldsymbol{L^0} \triangleq \boldsymbol{L} - \boldsymbol{B}$.

Using the Powell recurrence formulae \cite{Powell_1949}, R.G. Thomas established the following scheme to calculate the outgoing-wave reduced logarithmic derivatives $L_c$ for different angular momenta $\ell$ values in the Coulomb case (c.f. p.350, appendix of \cite{Lane_and_Thomas_1958}, eqs.(A.12) and (A.13))
\begin{equation}
\begin{IEEEeqnarraybox}[][c]{rcl}
\IEEEstrut
L_\ell & \ = \ & \frac{a_\ell}{b_\ell - L_{\ell-1}} - b_\ell
\end{IEEEeqnarraybox}
\label{eq::L_ell recurrence formula}
\end{equation}
with
\begin{equation}
a_\ell \triangleq \rho^2 + \left(\frac{\rho \eta}{\ell}\right)^2 \quad , \quad b_\ell \triangleq \ell + \left(\frac{\rho \eta}{\ell}\right)
\end{equation}

In general, both $O_c(\rho)$ and $L_{\ell}(\rho)$ are meromorphic functions of $\rho$ with \textit{a priori} an infinity of poles, and for whose computation we refer to \cite{Powell_1949, Thompson_1986, Michel_2006}.
In theorem \ref{theo::Mittag-Leffler of L_c theorem}, we here establish the Mittag-Leffler expansion of $L_c(\rho)$.

\begin{theorem}
\label{theo::Mittag-Leffler of L_c theorem}
\textsc{Outgoing-wave reduced logarithmic derivative $L_c(\rho)$ Mittag-Leffler Expansion}. \\
The outgoing-wave reduced logarithmic derivative $L_c(\rho)$, defined in (\ref{eq:L expression}), admits the following Mittag-Leffler pole expansion:
\begin{equation}
\begin{IEEEeqnarraybox}[][c]{rcl}
    \frac{L_c(\rho)}{\rho} & = & \frac{- \ell}{\rho} + \mathrm{i} + \sum_{n \geq 1}  \frac{1}{\rho-\omega_n}
\IEEEstrut\end{IEEEeqnarraybox}
\label{eq::Explicit Mittag-Leffler expansion of L_c}
\end{equation}
where $\left\{\omega_n \right\}$ are the roots of the $O_c(\rho)$ outgoing wavefunctions: $\forall n , \; O_c(\omega_n) = 0$.
For neutral particles, there are a finite number of such roots, reported in table \ref{tab::roots of the outgoing wave functions}.
\end{theorem}

\begin{proof}
From definition (\ref{eq:L expression}), $L_c$ is the reduced logarithmic derivative of the outgoing wavefunction $L_c(\rho) \triangleq  \rho \frac{O_c^{(1)}(\rho)}{O_c(\rho)}$. In both the Coulomb and the neutral particle case, the outgoing wavefunction $O_c(\rho)$ is a confluent hypergeometric function with simple roots $\left\{\omega_n \right\}$. Moreover, their logarithmic derivative $ \frac{O_c^{(1)}(\rho)}{O_c(\rho) }$ is bound at infinity. Thus, the following hypotheses stand:
\begin{itemize}
    \item $L_{\ell}(\rho)$ has simple poles $\left\{ \omega_n \right\}$, zeros of the $O_c(\rho)$,
    \item $\frac{L_{\ell}(\rho)}{\rho}$ has residues of $1$ at the $\left\{ \omega_n \right\}$ poles,
    \item $\exists M\in \mathbb{R}$ such as $\left| L_{\ell}(\rho) \right| < M |z|  $ on circles $\mathcal{C}_D$ as $D \longrightarrow \infty$ 
\end{itemize}
By removing the pole of $ \frac{O_c^{(1)}(\rho)}{O_c(\rho) }$ at zero, these hypotheses ensure Mittag-Leffler expansion (\ref{eq::Mittag-Leffler expansion of L_c}) is verified:
\begin{equation}
\begin{IEEEeqnarraybox}[][c]{rcl}
\frac{L_c(\rho)}{\rho} & = & \frac{L_c(0)}{\rho} + L_c^{(1)}(0) + \sum_{n \geq 1} \left[ \frac{1}{\rho-\omega_n} +  \frac{1}{\omega_n} \right]
\IEEEstrut\end{IEEEeqnarraybox}
\label{eq::Mittag-Leffler expansion of L_c}
\end{equation}
R.G. Thomas' recurrence formula (\ref{eq::L_ell recurrence formula}) implies that $L_c(\rho_c)$ satisfies $L_\ell(0) = - \ell$, for both neutral and charged particles.
Moreover, evaluating $ \frac{O_c^{(1)}(\rho)}{O_c(\rho) }$ at the limit of infinity yields:
\begin{equation}
     L_c^{(1)}(0) + \sum_{k\geq 1}\frac{1}{\omega_k} = \underset{\rho \to \infty }{\mathrm{Lim}}\left( \frac{L_c(\rho)}{\rho }\right) = \underset{\rho \to \infty }{\mathrm{Lim}}\left( \frac{O_c^{(1)}(\rho)}{O_c(\rho) }\right) = \mathrm{i}
\end{equation}
so that the Mittag-Leffler expansion (\ref{eq::Mittag-Leffler expansion of L_c}) takes the desired form of (\ref{eq::Explicit Mittag-Leffler expansion of L_c}).

\end{proof}

\begin{table*}
\caption{\label{tab::roots of the outgoing wave functions} Roots $\big\{\omega_n\big\}$ of the outgoing wave function $O_{\ell}(\rho)$, algebraic solutions for neutral particles up to $\ell \leq 4$.}
\begin{ruledtabular}
\begin{tabular}{ll}
$\ell = 0$ : s-wave  &  \quad \quad $\big\{\omega_0^{\ell = 0}\big\} = \left\{ \emptyset \right\}$  \tabularnewline
\hline 
$\ell = 1$ : p-wave  &  \quad\quad $ \big\{\omega_1^{\ell = 1}\big\} = \big\{  -\mathrm{i}\big\}$  \tabularnewline
\hline 
$\ell = 2$ : d-wave  &   \quad \quad $ \big\{\omega_1^{\ell = 2} , \omega_2^{\ell = 2} \big\} = \Big\{ \frac{1}{2}\left(-\sqrt{3} - 3 \mathrm{i}\right)  \quad , \quad \frac{1}{2}\left(\sqrt{3} - 3 \mathrm{i}\right) \Big\}$
\tabularnewline
\hline 
$\ell = 3$ : f-wave  &  \quad \quad $ \left\{  \omega_1^{\ell = 3} ,  \omega_2^{\ell = 3} ,
 \omega_3^{\ell = 3} \right\}$
\tabularnewline
\multicolumn{2}{c}{$ \begin{array}{l}
\omega_1^{\ell = 3} \triangleq -2\mathrm{i}-\frac{1}{2}\left(\sqrt{3} - \mathrm{i}\right)\sqrt[3]{\frac{1}{2}\left(1+\sqrt{5}\right)} - \frac{\sqrt{3} + \mathrm{i}}{2^{2/3}\sqrt[3]{1 + \sqrt{5}}}  \\
\omega_2^{\ell = 3} \triangleq \mathrm{i}\left(-2 + \sqrt[3]{\frac{2}{1+\sqrt{5}}} - \sqrt[3]{\frac{1}{2}\left(1+\sqrt{5} \right)} \right) \\
\omega_3^{\ell = 3} \triangleq -2\mathrm{i}+\frac{1}{2}\left(\sqrt{3} + \mathrm{i}\right)\sqrt[3]{\frac{1}{2}\left(1+\sqrt{5}\right)} + \frac{\sqrt{3} - \mathrm{i}}{2^{2/3}\sqrt[3]{1 + \sqrt{5}}}
\end{array} $}
\tabularnewline
\hline 
$ \ell = 4$ : g-wave  & \quad \quad $ \left\{  \omega_1^{\ell = 4} ,  \omega_2^{\ell = 4} ,
 \omega_3^{\ell = 4 } ,  \omega_4^{\ell = 4 } \right\}$
\tabularnewline
\\
\multicolumn{2}{c}{$ \begin{array}{l}
\omega_1^{\ell = 4} \triangleq  -\frac{5 \mathrm{i}}{2} - \frac{1}{2}\sqrt{5 + \frac{15^{2/3}}{\sqrt[3]{\frac{1}{2}\left(5 + \mathrm{i}\sqrt{35}\right)}} + \sqrt[3]{\frac{15}{2}\left(5 + \mathrm{i}\sqrt{35} \right)}}  - \frac{1}{2}\sqrt{10-\frac{15^{2/3}}{\sqrt[3]{\frac{1}{2}\left(5 + \mathrm{i}\sqrt{35} \right)}} - \sqrt[3]{\frac{15}{2}\left(5 + \mathrm{i}\sqrt{35}\right)} - \frac{10\mathrm{i}}{\sqrt{5 + \frac{15^{2/3}}{\sqrt[3]{\frac{1}{2}\left(5 + \mathrm{i}\sqrt{35}\right)}} + \sqrt[3]{\frac{15}{2}\left(5 + \mathrm{i}\sqrt{35}\right)} }}}   \\
\omega_2^{\ell = 4} \triangleq  -\frac{5 \mathrm{i}}{2} - \frac{1}{2}\sqrt{5 + \frac{15^{2/3}}{\sqrt[3]{\frac{1}{2}\left(5 + \mathrm{i}\sqrt{35}\right)}} + \sqrt[3]{\frac{15}{2}\left(5 + \mathrm{i}\sqrt{35} \right)}} + \frac{1}{2}\sqrt{10-\frac{15^{2/3}}{\sqrt[3]{\frac{1}{2}\left(5 + \mathrm{i}\sqrt{35} \right)}} - \sqrt[3]{\frac{15}{2}\left(5 + \mathrm{i}\sqrt{35}\right)} - \frac{10\mathrm{i}}{\sqrt{5 + \frac{15^{2/3}}{\sqrt[3]{\frac{1}{2}\left(5 + \mathrm{i}\sqrt{35}\right)}} + \sqrt[3]{\frac{15}{2}\left(5 + \mathrm{i}\sqrt{35}\right)} }}}  \\
\omega_3^{\ell = 4} \triangleq  -\frac{5 \mathrm{i}}{2} + \frac{1}{2}\sqrt{5 + \frac{15^{2/3}}{\sqrt[3]{\frac{1}{2}\left(5 + \mathrm{i}\sqrt{35}\right)}} + \sqrt[3]{\frac{15}{2}\left(5 + \mathrm{i}\sqrt{35} \right)}} - \frac{1}{2}\sqrt{10-\frac{15^{2/3}}{\sqrt[3]{\frac{1}{2}\left(5 + \mathrm{i}\sqrt{35} \right)}} - \sqrt[3]{\frac{15}{2}\left(5 + \mathrm{i}\sqrt{35}\right)} + \frac{10\mathrm{i}}{\sqrt{5 + \frac{15^{2/3}}{\sqrt[3]{\frac{1}{2}\left(5 + \mathrm{i}\sqrt{35}\right)}} + \sqrt[3]{\frac{15}{2}\left(5 + \mathrm{i}\sqrt{35}\right)} }}} \\
\omega_4^{\ell = 4} \triangleq  -\frac{5 \mathrm{i}}{2} + \frac{1}{2}\sqrt{5 + \frac{15^{2/3}}{\sqrt[3]{\frac{1}{2}\left(5 + \mathrm{i}\sqrt{35}\right)}} + \sqrt[3]{\frac{15}{2}\left(5 + \mathrm{i}\sqrt{35} \right)}} + \frac{1}{2}\sqrt{10-\frac{15^{2/3}}{\sqrt[3]{\frac{1}{2}\left(5 + \mathrm{i}\sqrt{35} \right)}} - \sqrt[3]{\frac{15}{2}\left(5 + \mathrm{i}\sqrt{35}\right)} + \frac{10\mathrm{i}}{\sqrt{5 + \frac{15^{2/3}}{\sqrt[3]{\frac{1}{2}\left(5 + \mathrm{i}\sqrt{35}\right)}} + \sqrt[3]{\frac{15}{2}\left(5 + \mathrm{i}\sqrt{35}\right)} }}}
\end{array} $}
\end{tabular}
\end{ruledtabular}
\end{table*}

Theorem \ref{theo::Mittag-Leffler of L_c theorem} establishes, for the first time, the Mittag-Leffler expansion of $L^0_c(\rho_c)$ as a function of the roots $\left\{\omega_n \right\}$ of the outgoing wavefunctions $O_c(\rho)$, which are Hankel functions in the neutral particle case, and Whittaker functions in the more general case of charged particles (c.f. equations (2.14b) and (2.17) section III.2.b. p.269 of \cite{Lane_and_Thomas_1958}).
Extensive literature covers these functions \cite{Abramowitz_and_Stegun, NIST_DLMF}. In the neutral particles case of Hankel functions \cite{Zeros_of_Hankel_Functions_and_Poles_1963, Complex_zeros_of_cylinder_functions_1966, Zeros_of_Hankel_Functions_1982, Reduced_Logarithmic_Dervative_of_Hankel_functions_1983, Calculating_zeros_of_Hankel_functions_USSR_1985, Partial_fraction_expansion_Bessel_2005} the search for their zeros established that the reduced logarithmic derivative of the outgoing wave function is a rational function of $k_c$ of degree $\ell$.
In the general case, there are indeed $\ell$ zeros to the Hankel function for $\left| \Re[\rho] \right| < \ell$, but for $\left| \Re[\rho] \right| > \ell$ there exists an infinity of zeros, on or close to the real axis (c.f. FIG.1\&2 of \cite{Complex_zeros_of_cylinder_functions_1966}).
However, in our particular case of physical (i. e. integer) angular momenta $\ell \in \mathbb{Z}$, the order of the Hankel function happens to be a half-integer: $H_{\ell+1/2}$. Crucially, Hankel functions of half integer order constitute a very special case: they have only a finite number of zeros in the finite complex plane, where all but $\ell$ of them have migrated to infinity.
This behavior is reported in \cite{Zeros_of_Hankel_Functions_1982}, where one can observe how the zeros of $H_{\nu}$ as $\nu$ varies between two consecutive integer values.
Here, we report in table \ref{tab::roots of the outgoing wave functions} all the algebraically solvable cases of up to $\ell=4$, past which $\left\{\omega_n\right\}$ are not guaranteed to be solvable by radicals (c.f. Abel-Ruffini theorem and Galois theory).

Another perspective over this property is that in the neutral particle case, $\eta = 0$ and $L_{\ell=0}(\rho) = \mathrm{i}\rho$, so that recurrence relation (\ref{eq::L_ell recurrence formula}) entails $L_c(\rho_c)$ -- and thus the $\boldsymbol{L^0}$ function -- is a rational fraction in $\rho_c$, whose irreducible expressions are reported in table \ref{tab::L_values_neutral} along with their partial fraction decomposition, established in theorem \ref{theo::Mittag-Leffler of L_c theorem}, and whose poles are documented in table \ref{tab::roots of the outgoing wave functions}.
Moreover, since definition (\ref{eq:L expression}) entails $\frac{\partial O_c}{\partial \rho}(\rho) = \frac{L_c}{\rho}(\rho) O_c(\rho)$, a direct integration of (\ref{eq::Mittag-Leffler expansion of L_c}) yields (with the correct multiplicative constant):
\begin{equation}
\begin{IEEEeqnarraybox}[][c]{rcl}
O_\ell(\rho) & = & \mathrm{e}^{\mathrm{i}\left(\rho + \frac{1}{2}\ell \pi \right)}\frac{\prod_{n \geq 1}\left(\rho-\omega_n\right)}{\rho^\ell} 
\IEEEstrut\end{IEEEeqnarraybox}
\label{eq::Outgoing wavefuntion O_ell expression for neutral particles}
\end{equation}
This expression converges for neutral particles as the number of poles is finite, so using Vieta's formulas with the denominator of $L_\ell(\rho)$ enables to construct the developed forms reported in table \ref{tab::L_values_neutral}.

Similar results do not hold for the charged particules case of Whittaker functions, where there always exists an infinity of zeros to the outgoing wavefunction \cite{Zeros_of_Whittaker_function_1999, New_Asymptotics_Whittaker_zeros_2001}, and where a Coulomb phase shift would be present for any Weierstrass expansion in infinite product of type (\ref{eq::Outgoing wavefuntion O_ell expression for neutral particles}).

\subsection{\label{subsec:Internal region parameters} Internal region parameters}

Projections upon the orthonormal basis formed by the eigenvectors of the Hamiltonian completed by the Bloch operator $\mathcal{L}$ allow for the parametrization of the interaction Hamiltonian in the internal region by means of the Wigner-Eisenbud \textit{resonance parameters} \cite{Bloch_1957}, composed of both the real \textit{resonance energies} $E_\lambda \in \mathbb{R}$, and the real \textit{resonance widths} $\gamma_{\lambda,c} \in \mathbb{R}$.
From the latter, and using Brune's notation $\boldsymbol{e} \triangleq \boldsymbol{\mathrm{diag}}\left( E_\lambda \right)$ and $\boldsymbol{\gamma}\triangleq \boldsymbol{\mathrm{mat}}\left(\gamma_{\lambda,c}\right)_{\lambda,c}$, the \textit{Channel R matrix}, $\boldsymbol{R}$, is defined as
\begin{equation}
\begin{IEEEeqnarraybox}[][c]{rcl}
R_{c,c'} & \; \triangleq \; &  \sum_{\lambda=1}^{N_\lambda}\frac{\gamma_{\lambda,c}\gamma_{\lambda,c'}}{E_\lambda - E} \quad  \mathrm{i.e.}  \quad \boldsymbol{R}\; = \; \boldsymbol{\gamma}^\mathsf{T} \left(\boldsymbol{e} - E\Id{}\right)^{-1}\boldsymbol{\gamma}
\IEEEstrut\end{IEEEeqnarraybox}
\label{eq:R expression}
\end{equation}
and the \textit{Level A matrix}, $\boldsymbol{A}$, is defined through its inverse:
\begin{equation}
\begin{IEEEeqnarraybox}[][c]{rcl}
\boldsymbol{A^{-1}} & \ \triangleq \ & \boldsymbol{e} - E\Id{} - \boldsymbol{\gamma}\left( \boldsymbol{L} - \boldsymbol{B} \right)\boldsymbol{\gamma}^\mathsf{T}
\IEEEstrut\end{IEEEeqnarraybox}
\label{eq:inv_A expression}
\end{equation}
where $\boldsymbol{B} = \boldsymbol{\mathrm{diag}}\left( B_c \right)$ is the arbitrary outgoing-wave boundary condition, which is arbitrary, constant (non-dependent on the wavenumber), and for which Bloch demonstrated that if it is real (i.e. $B_c \in \mathbb{R}$), then the Wigner-Eisenbud resonance parameters are also real \cite{Bloch_1957}.
From this, one can view the Wigner-Eisenbud parameters as the set of channel radii $a_c$, boundary conditions $B_c$, resonance widths $\gamma_{\lambda, c}$, resonance energies $E_\lambda$ and thresholds $E_{T_c}$. This set of parameters $\left\{a_c , B_c , \gamma_{\lambda,c}, E_\lambda, E_{T_c} \right\}$ fully determines the energy (or wavenumber) dependence of the scattering matrix $\boldsymbol{U}$ through equation (\ref{eq:U expression}).

\subsection{\label{subsec:Scattering matrix} Scattering matrix and R-matrix parameters}

As explained by Claude Bloch, the genius of R-matrix theory stems from it combining the internal region with the external region to simply express the resulting scattering matrix $\boldsymbol{U}$ (also called \textit{collision matrix}, and often noted $\boldsymbol{S}$, though we here stick to the Lane \& Thomas scripture $\boldsymbol{U}$ for the scattering matrix) as:
\begin{equation}
\begin{IEEEeqnarraybox}[][c]{rcl}
\boldsymbol{U} & \ = \ & \boldsymbol{O}^{-1}\boldsymbol{I} + \boldsymbol{w} \boldsymbol{\rho}^{1/2} \boldsymbol{O}^{-1}\left[ \boldsymbol{R}^{-1} + \boldsymbol{B} -  \boldsymbol{L} \right]^{-1}  \boldsymbol{O}^{-1} \boldsymbol{\rho}^{1/2} \\
& \ = \ & \boldsymbol{O}^{-1}\boldsymbol{I} + 2\mathrm{i} \boldsymbol{\rho}^{1/2} \boldsymbol{O}^{-1} \boldsymbol{\gamma}^\mathsf{T} \boldsymbol{A} \boldsymbol{\gamma} \boldsymbol{O}^{-1} \boldsymbol{\rho}^{1/2} \\
& \ = \ & \boldsymbol{O}^{-1}\boldsymbol{I} + 2\mathrm{i} \boldsymbol{\rho}^{1/2} \boldsymbol{O}^{-1} \boldsymbol{R}_{L} \boldsymbol{O}^{-1} \boldsymbol{\rho}^{1/2}
\IEEEstrut\end{IEEEeqnarraybox}
\label{eq:U expression}
\end{equation}
The equivalence between these channel and level matrix expressions stems from the identity $ \left[ \Id{} -  \boldsymbol{R}\boldsymbol{L^0} \right]^{-1} \boldsymbol{R}  = \boldsymbol{\gamma}^\mathsf{T} \boldsymbol{A} \boldsymbol{\gamma}$ which defines the \textit{Kapur-Peierls operator}, $\boldsymbol{R}_{L}$:
\begin{equation}
\begin{IEEEeqnarraybox}[][c]{rcl}
\boldsymbol{R}_{L}  & \ \triangleq  \ &  \left[ \Id{} -  \boldsymbol{R}\boldsymbol{L^0} \right]^{-1} \boldsymbol{R}  =   \boldsymbol{\gamma}^\mathsf{T} \boldsymbol{A} \boldsymbol{\gamma}
    \label{eq:Kapur-Peierls Operator and Channel-Level equivalence}
\IEEEstrut\end{IEEEeqnarraybox}
\end{equation}
Identity (\ref{eq:Kapur-Peierls Operator and Channel-Level equivalence}) can be proved by means of the \textit{Woodbury identity}:
\begin{equation}
\begin{IEEEeqnarraybox}[][c]{rcl}
\left[ \boldsymbol{A} + \boldsymbol{B}\boldsymbol{D}^{-1}\boldsymbol{C}\right]^{-1}   & \ = \ & \boldsymbol{A}^{-1} - \boldsymbol{A}^{-1}\boldsymbol{B}\left[ \boldsymbol{D} + \boldsymbol{C}\boldsymbol{A}^{-1}\boldsymbol{B} \right]^{-1}\boldsymbol{C}\boldsymbol{A}^{-1}
\IEEEstrut\end{IEEEeqnarraybox}
\label{eq:Woodbury identity}
\end{equation}
Indeed, the application of the Woodbury identity (\ref{eq:Woodbury identity}) to equality (\ref{eq:Kapur-Peierls Operator and Channel-Level equivalence}), with $\boldsymbol{A}_{\mathrm{Wood}} = \boldsymbol{R}^{-1}$, $\boldsymbol{B}_{\mathrm{Wood}} = \boldsymbol{L^0}$, and $\boldsymbol{C}_{\mathrm{Wood}} = \boldsymbol{D}_{\mathrm{Wood}} = \Id{}$ yields
\begin{equation*}
\begin{IEEEeqnarraybox}[][c]{r}
\IEEEstrut
 \left[ \Id{} -  \boldsymbol{R}\boldsymbol{L^0} \right]^{-1} \boldsymbol{R}   =  \boldsymbol{R} + \boldsymbol{R} \boldsymbol{L^0} \left[\Id{} - \boldsymbol{R}\boldsymbol{L^0} \right]^{-1}\boldsymbol{R} \\
=  \boldsymbol{\gamma}^\mathsf{T} \left[  \left(\boldsymbol{e} - E\Id{}\right)^{-1} +  \left(\boldsymbol{e} - E\Id{}\right)^{-1}\boldsymbol{\gamma} \boldsymbol{L^0} \times \right. \\
\left. \left[\Id{} - \boldsymbol{\gamma}^\mathsf{T} \left(\boldsymbol{e} - E\Id{}\right)^{-1}\boldsymbol{\gamma} \boldsymbol{L^0}  \right]^{-1}\boldsymbol{\gamma}^\mathsf{T} \left(\boldsymbol{e} - E\Id{}\right)^{-1} \right] \boldsymbol{\gamma}
\end{IEEEeqnarraybox}
\label{eq:R_L Woodbury channel-level equivalence 1}
\end{equation*}
and then reversely applying the Woodbury identity with $\boldsymbol{A}_{\mathrm{Wood}} = \left(\boldsymbol{e} - E\Id{}\right)$, $\boldsymbol{B}_{\mathrm{Wood}} = -\boldsymbol{\gamma} \boldsymbol{L^0}$, $\boldsymbol{C}_{\mathrm{Wood}} = \boldsymbol{\gamma}^\mathsf{T}$, and $ \boldsymbol{D}_{\mathrm{Wood}} = \Id{}$ one now recognizes
\begin{equation*}
\begin{IEEEeqnarraybox}[][c]{rCl}
\IEEEstrut
 \left[ \Id{} -  \boldsymbol{R}\boldsymbol{L^0} \right]^{-1} \boldsymbol{R}  
& = & \boldsymbol{\gamma}^\mathsf{T}   \left[ \left(\boldsymbol{e} - E\Id{}\right) - \boldsymbol{\gamma} \boldsymbol{L^0} \boldsymbol{\gamma}^\mathsf{T}\right]^{-1} \boldsymbol{\gamma}  \\
& = & \boldsymbol{\gamma}^\mathsf{T} \boldsymbol{A} \boldsymbol{\gamma}
\end{IEEEeqnarraybox}
\label{eq:R_L Woodbury channel-level equivalence 2}
\end{equation*}

Considering the multi-sheeted Riemann surface stemming from the analytic continuation of mapping (\ref{eq:rho_c(E) mapping}),
a truly remarkable and seldom noted property of the Wigner-Eisenbud formalism is that it completely de-entangles the branch points and the multi-sheeted structure --- entirely present in the outgoing $\boldsymbol{O}$ and incoming $\boldsymbol{I}$ wave functions in the scattering matrix expression (\ref{eq:U expression}) --- from the resonance parameters --- which are the poles and residues of the channel matrix $\boldsymbol{R}$ as of equation (\ref{eq:R expression}), and these poles and residues live on a simple complex energy $E$ sheet, with no branch points, and furthermore are all real.
This de-entanglement of the branch-point structure gives the $\boldsymbol{R}$ matrix all its uniqueness in R-matrix theory.
For instance, it does not translate to the level matrix $\boldsymbol{A}$, whose analytic continuation entails a multi-sheeted Riemann surface due to the introduction of the $\boldsymbol{L^0}(\boldsymbol{\rho}(E)))$ matrix function in its definition (\ref{eq:inv_A expression}). The same is true for the alternative parameters, as will be discussed throughout this article.

\subsection{\label{subsec:Cross section and scattering matrix} Cross section and scattering matrix}

General scattering theory expresses the incoming channel $c$ and outgoing channel $c'$ angle-integrated partial cross section $\sigma_{c,c'}(E)$ at energy $E$ as a function of the scattering matrix $U_{c,c'}(E)$ according to eq.(3.2d) VIII.3. p.293 of \cite{Lane_and_Thomas_1958}:
\begin{equation}
\begin{IEEEeqnarraybox}[][c]{rcl}
\sigma_{c,c'}(E) & \ = \ & \pi g_{J^\pi_c} \left| \frac{ \delta_{c,c'}\mathrm{e}^{2\mathrm{i}\big( \sigma_{\ell_c}(\eta_c) - \sigma_{0}(\eta_c) \big) } - U_{c,c'}(E)}{k_c(E)}\right|^2
\IEEEstrut\end{IEEEeqnarraybox}
\label{eq:partial sigma_cc'from scattering matrix}
\end{equation}
where $ g_{J^\pi_c} \ \triangleq \ \frac{2 J + 1 }{\left(2 I_1 + 1 \right)\left(2 I_2 + 1 \right) }$ is the \textit{spin statistical factor} defined eq.(3.2c) VIII.3. p.293, and where the \textit{Coulomb phase shift}, $\sigma_{\ell_c}(\eta_c) $, is defined by Ian Thompson in eq.(33.2.10) of \cite{NIST_DLMF} for angular momentum $\ell_c$ and dimensionless Coulomb field parameter $\eta_c = \frac{Z_1 Z_2 e^2 M_\alpha a_c}{\hbar^2 \rho_c}$.

\subsection{\label{sec:Invariance to B_c}Invariance to arbiraty boundary parameter $B_c$}

Having recalled essential results from R-matrix theory and the Wigner-Eisenbud parameters $\left\{a_c , B_c , \gamma_{\lambda,c}, E_\lambda, E_{T_c} \right\}$, we here focus on the fact that the fundamental physical operator describing the scattering event is the scattering matrix $\boldsymbol{U}$, and while the threshold energies $E_{T_c}$ are intrinsic physical properties of the system, all the other Wigner-Eisenbud parameters $a_c$, $B_c$, $\gamma_{\lambda,c}$, and $E_\lambda$ are interrelated and depend on arbitrary values of the channel radius $a_c$, or the boundary condition $B_c$.
Though the channel radius $a_c$ can arguably have some physical interpretation, this is not the case of the boundary condition $B_c$.

The dependence of the Wigner-Eisenbud parameters to the boundary condition $B_c$ can be made explicit by fixing the channel radius $a_c$ and performing a change of boundary condition $\boldsymbol{B} \to \boldsymbol{B'}$. This must entail a change in resonance parameters $E_\lambda \rightarrow E_\lambda'$ and $\gamma_{\lambda,c} \rightarrow \gamma_{\lambda,c}'$ which leaves the scattering matrix $\boldsymbol{U}$ unchanged.

As described by Barker in \cite{Boundary_condition_Barker_1972}, such change of variables can be performed by noticing that $ \boldsymbol{e} - \boldsymbol{\gamma}\left( \boldsymbol{B}' - \boldsymbol{B} \right)\boldsymbol{\gamma}^\mathsf{T}$ is a real symmetric matrix when both $\boldsymbol{B}$ and $\boldsymbol{B'}$ are real.
The spectral theorem thus assures there exists a real orthogonal matrix $\boldsymbol{K}$ and a real diagonal matrix $\boldsymbol{D}$ such that
\begin{equation}
\boldsymbol{e} - \boldsymbol{\gamma}\left( \boldsymbol{B}' - \boldsymbol{B} \right)\boldsymbol{\gamma}^\mathsf{T} = \boldsymbol{K}^\mathsf{T}\boldsymbol{D}\boldsymbol{K}
\end{equation}
The new parameters are then defined as
\begin{equation}
\boldsymbol{e}' \triangleq  \boldsymbol{D} \quad \quad , \quad \quad \boldsymbol{\gamma}' \triangleq \boldsymbol{K}\boldsymbol{\gamma}
\label{eq: Wigner-Eisenbud parameters transformations under change of Bc}
\end{equation}
This change of variables satisfies:
\begin{equation}
{\boldsymbol{\gamma}'}^\mathsf{T} \boldsymbol{A}_{B'} \boldsymbol{\gamma}' =  \boldsymbol{\gamma}^\mathsf{T} \boldsymbol{A}_{B} \boldsymbol{\gamma}
\label{eq:: gAg invariance for B'}
\end{equation}
and thus leaves the scattering matrix unaltered through equation (\ref{eq:U expression}). Here $\boldsymbol{A}_{B'}$ designates the level matrix from parameters $\boldsymbol{e}'$, $\boldsymbol{\gamma}'$ and $\boldsymbol{B}'$.
Equivalently, using the Woodbury identity (\ref{eq:Woodbury identity}) shows that this change of variables verifies (c.f. eq.(4) of \cite{Boundary_condition_Barker_1972} or eq. (3.27) of \cite{Descouvemont_2010}):
\begin{equation}
\boldsymbol{R}_B^{-1} + \boldsymbol{B} \ = \ \boldsymbol{R}_{B'}^{-1} + \boldsymbol{B'}
\label{eq:: R_B invariance for B'}
\end{equation}
If the change of variable is infinitesimal, this invariance property translates into the following equivalent differential equations on the Wigner-Eisenbud $\boldsymbol{R}_B$ matrix,
\begin{equation}
\frac{\partial \boldsymbol{R}_B^{-1}}{\partial \boldsymbol{B}} + \Id{} \ = \ \boldsymbol{0} \quad \text{i.e.} \quad  \frac{\partial \boldsymbol{R}_B}{\partial \boldsymbol{B}} - \boldsymbol{R}_B^2 \ = \ \boldsymbol{0}
\label{eq:: R_B invariance for infinitesimal B'}
\end{equation}
(c.f. eq (2.5b) section IV.2. p.274 of \cite{Lane_and_Thomas_1958}) where we made use of the following property to prove the equivalence:
\begin{equation}
\frac{\partial \boldsymbol{M^{-1}}}{\partial z}(z) = - \boldsymbol{M^{-1}}(z)  \left( \frac{\partial \boldsymbol{M}}{\partial z}(z) \right) \boldsymbol{M^{-1}}(z)
\label{eq::inv M derivatie property}
\end{equation}

\section{\label{sec:R_S Brune transform} The alternative parametrization of R-matrix theory}

Since the physics of the system are invariant with the choice of the arbitrary $B_c$ boundary condition, Brune built on the work of Barker \cite{Boundary_condition_Barker_1972}, Angulo and Descouvemont \cite{anguloMatrixAnalysisInterference2000}, to propose an alternative parametrization of R-matrix theory in which the alternative parameters, $\boldsymbol{\widetilde{e}}$ and $\boldsymbol{\widetilde{\gamma}}$, are boundary-condition independent \cite{Brune_2002}.

\subsection{\label{sec:R_S def}Definition of the alternative $\boldsymbol{R}_S$ parametrization}

Key to the alternative parametrization is the splitting of the outgoing-wave reduced logarithmic derivative -- and thus the $\boldsymbol{L^0}$ matrix function -- into real and imaginary parts, respectively the shift $\boldsymbol{S}$ and penetration $\boldsymbol{P}$ factors:
\begin{equation}
\boldsymbol{L} \ = \ \boldsymbol{S} + i \boldsymbol{P}
\label{eq:: L = S + iP}
\end{equation}

From there, and with slight changes from the notation in \cite{Brune_2002}, the \textit{alternative level matrix} $\boldsymbol{\widetilde{A}}$ is defined as:
\begin{equation}
\boldsymbol{\widetilde{A}^{-1}}(E) \ = \ \boldsymbol{\widetilde{G}} + \boldsymbol{\widetilde{e}} - E\left[\Id{} + \boldsymbol{\widetilde{H}} \right] - \boldsymbol{\widetilde{\gamma}}\boldsymbol{L}(E)\boldsymbol{\widetilde{\gamma}}^\mathsf{T}
\label{eq::Brune physical level matrix}
\end{equation}
with
\begin{equation}
\widetilde{G}_{\lambda \mu} \ = \ \frac{\widetilde{\gamma_\mu}\left( S_\mu \widetilde{E_\lambda} -  S_\lambda \widetilde{E_\mu}  \right)\widetilde{\gamma_\lambda}}{ \widetilde{E_\lambda} - \widetilde{E_\mu} }
\end{equation}
and
\begin{equation}
\widetilde{H}_{\lambda \mu} \ = \ \frac{\widetilde{\gamma_\mu}\left( S_\mu -  S_\lambda \right)\widetilde{\gamma_\lambda}}{ \widetilde{E_\lambda} - \widetilde{E_\mu} }
\end{equation}
such that with the new \textit{alternative resonance parameters}, $\widetilde{E_i}$ and $\widetilde{\gamma_{i,c}}$, the following equality stands,
\begin{equation}
\boldsymbol{\gamma^\mathsf{T} A \gamma} = \boldsymbol{\widetilde{\gamma}^\mathsf{T} \widetilde{A} \widetilde{\gamma}}
\label{eq:R_L unchanged by Brune}
\end{equation}
and thus the scattering matrix $\boldsymbol{U}$ is left unchanged.

These alternative parameters $\widetilde{e}$ and $\widetilde{\gamma}$ are no longer $\boldsymbol{B}$ dependent since the arbitrary boundary condition does not appear in the definition of the alternative level matrix, and from there in the parametrization of the scattering matrix.

Brune explains how to compute these alternative parameters from the Wigner-Eisenbud ones by finding the $\left\{\widetilde{E_i}\right\}$ scalars and $\left\{\boldsymbol{a_i}\right\}$ vectors that solve the Brune generalized eigenproblem \cite{Brune_2002}:
\begin{equation}
\begin{IEEEeqnarraybox}[][c]{rCl}
\left[\boldsymbol{e} - \boldsymbol{\gamma} \left( \boldsymbol{S}(\widetilde{E_i}) - \boldsymbol{B} \right)\boldsymbol{ \gamma}^\mathsf{T} \right]\boldsymbol{a_i} = \widetilde{E_i}\boldsymbol{a_i}
\IEEEstrut\end{IEEEeqnarraybox}
\label{eq:Brune eigenproblem}
\end{equation}
where each eigenvector is normalized so that: 
\begin{equation}
\begin{IEEEeqnarraybox}[][c]{rCl}
\boldsymbol{a_i}^\mathsf{T}\boldsymbol{a_i} =1
\IEEEstrut\end{IEEEeqnarraybox}
\label{eq:Brune eigenvectors normalized}
\end{equation}
and defining the alternative parameters as:
\begin{equation}
\boldsymbol{\widetilde{e}} \triangleq  \boldsymbol{\mathrm{diag}}(\widetilde{E_i}) \quad \quad , \quad \quad \boldsymbol{\widetilde{\gamma}} \triangleq \boldsymbol{a}^\mathsf{T} \boldsymbol{\gamma}
\label{eq:Brune parameters}
\end{equation}
where $\boldsymbol{a}$ is the matrix composed of the column eigenvectors: $\boldsymbol{a} \triangleq \left[\boldsymbol{a_1}, \hdots, \boldsymbol{a_i} , \hdots \right]$.
The alternative level matrix is then defined as (c.f. equation (30), \cite{Brune_2002}):
\begin{equation}
\boldsymbol{\widetilde{A}}^{-1} \ \triangleq \ \boldsymbol{a}^\mathsf{T}\boldsymbol{A}^{-1}\boldsymbol{a}
\label{eq:: Brune invA}
\end{equation}
which guarantees
\begin{equation}
\boldsymbol{A} \ = \ \boldsymbol{a}\boldsymbol{\widetilde{A}}\boldsymbol{a}^\mathsf{T}
\label{eq:: Brune A = aAa}
\end{equation}
and thus (\ref{eq:R_L unchanged by Brune}), and whose explicit expression is (\ref{eq::Brune physical level matrix}).

Note that searching for the general eigenvalues in (\ref{eq:Brune eigenproblem}) is equivalent to solving (apply the Sylvester determinant identity theorem, or c.f. eq. (49)-(50) in \cite{Brune_2002}):
\begin{equation}
\left.\mathrm{det}\left(\boldsymbol{R}_S^{-1}(E)\right)\right|_{E = \widetilde{E_i}} = 0
\label{eq:R_S by Brune det search}
\end{equation}
i.e. solving for the poles of the $\boldsymbol{R}_S$ operator defined as
\begin{equation}
\boldsymbol{R}_S^{-1} \triangleq \boldsymbol{R}^{-1} + \boldsymbol{B} - \boldsymbol{S}
\label{eq:R_S by Brune}
\end{equation}

The key insight is that in equation (22) of \cite{Brune_2002}, Brune builds a square matrix $\boldsymbol{a}\triangleq \left[\boldsymbol{a_1}, \hdots, \boldsymbol{a_i} , \hdots, \boldsymbol{a_{N_\lambda}} \right]$, from which he is able to built the inverse alternative level matrix in his equation (30) of \cite{Brune_2002}. Brune justifies that this matrix is indeed square in the paragraphs between equations (46) and (47) by a three-step monotony argument depicted in FIG. 1 of \cite{Brune_2002}: 1) he assumes $S_c(E)$ is continuous (i.e. has no real poles); 2) he assumes $\frac{ \partial S_c}{\partial E} \geq 0$, which is always true for negative energies and has recently proved to be true for positive energies in the case of repulsive Coulomb interactions \cite{Brune_Mark_monotonic_properties_of_shift_2018} (a general proof is lacking for positive energy attractive Coulomb channels but has always been verified in practice); 3) he invokes the eigenvalue repulsion behavior (no-crossing rule).
If these three assumption are true, since the left-hand-side of (\ref{eq:Brune eigenproblem}) is a real symmetric matrix for any real energy value, then the spectral theorem guarantees there exists $N_\lambda$ different real eigenvalues to it, and Brune's three assumptions above elegantly guarantee that there exists exactly $N_\lambda$ real solutions to the generalized eiganvalue problem (\ref{eq:Brune eigenproblem}).

\subsection{\label{subsubsec::Ambiguity in shift and penetration}Ambiguity in shift and penetration factors definition for complex wavenumbers}

There is a subtlety, however. A careful analysis reveals that the assumption that $S_c(E)$ is continuous or monotonously increasing is not unequivocal, and points to an open discussion in the field of R-matrix theory and nuclear cross section evaluations: how should we continue the scattering matrix $\boldsymbol{U}$ to complex wavenumbers $k_c \in \mathbb{C}$ ?
Indeed, there is an ambiguity in the definition of the shift $S_c(E)$ and penetration $P_c(E)$ functions: two approaches are possible, and the community is not clear on which one is correct.

The first approach, legacy of Lane \& Thomas, is to define the shift and penetration functions as the real and imaginary parts of the the outgoing-wave reduced logarithmic derivative:
\begin{equation}
\forall E \in \mathbb{C}\;, \left\{ \quad \begin{array}{cc}
     \boldsymbol{S}(E) \triangleq & \Re\left[\boldsymbol{L}(E)\right]  \; \; \in \; \mathbb{R} \\
     \boldsymbol{P}(E) \triangleq & \Im\left[\boldsymbol{L}(E)\right]  \; \; \in \; \mathbb{R}
\end{array}   \right.
\label{eq:: Def S = Re[L], P = Im[L]}
\end{equation}
This definition, introduced in \cite{Lane_and_Thomas_1958} III.4.a. from equations (4.4) to (4.7c), finds its justification in the discussion between equations (2.1) and (2.2) of \cite{Lane_and_Thomas_1958} VII.2, as it presents the advantage of automatically closing the sub-threshold channels since:
\begin{equation}
\forall E<E_{T_c}\;, \quad  \Im\left[L_c(E)\right]= 0
\label{eq:: Im[L] = 0 for E<0}
\end{equation}
This elegant closure of channels comes at the cost of loosing the mathematical properties of the scattering matrix $\boldsymbol{U}(\boldsymbol{k)}$: it is no longer analytic for complex wavenumbers $k_c \in \mathbb{C}$ (we will also show in a follow-up article \cite{Ducru_Scattering_Matrix_of_Complex_Wavenumbers_2019} that this introduces non-physical spurious poles to the scattering matrix and brakes the generalized unitarity of Eden \& Taylor \cite{Eden_and_Taylor}).
In this Lane \& Thomas approach (\ref{eq:: Def S = Re[L], P = Im[L]}), the function calculated for $\boldsymbol{S}$ changes from $S(E) \triangleq S_c(E)$ above threshold ($E\geq E_{T_c}$), to $S(E) \triangleq L_c(E)$ below threshold ($E <  E_{T_c}$), because of (\ref{eq:: Im[L] = 0 for E<0}). 
Moreover, definition (\ref{eq:: Def S = Re[L], P = Im[L]}) induces ramifications for both the shift and the penetration factors, as we show in lemma \ref{lem:: Lane and Thomas S_c ramification properties}.

\begin{lem}\label{lem:: Lane and Thomas S_c ramification properties}
\textsc{Branch-point definition of shift $S_c(E)$ and penetration $P_c(E)$ functions}.\\
Definition (\ref{eq:: Def S = Re[L], P = Im[L]}) of the shift $S_c(E)$ and penetration $P_c(E)$ functions, legacy of Lane \& Thomas, entails:
\begin{itemize}
    \item branch-points for both $S_c(E)$ and $P_c(E)$, induced by the multi-sheeted nature of mapping (\ref{eq:rho_c(E) mapping}),
    \item on the $\big\{ E, - \big\}$ sheet below threshold $E < E_{T_c}$, the shift function $S_c(E)$ can present discontinuities and areas where $\frac{ \partial S_c }{\partial E}(E) < 0$,
    \item in particular, for neutral particles of odd angular momenta $\ell_c \equiv 1 \; (\mathrm{mod} \; 2)$, there is exactly one real sub-threshold pole to $S_c(E)$ on the $\big\{ E, - \big\}$ sheet,
    \item everywhere other than sub-threshold $\big\{ E, - \big\}$ sheet, and in particular on all of the $\big\{ E, + \big\}$ sheet, the shift function $S_c(E)$ is continuous and monotonously increasing:  $\frac{ \partial S_c }{\partial E}(E) \geq 0$.
\end{itemize}
\end{lem}

\begin{proof}
The proof simply introduces the branch-structure of the $\rho_c(E)$ mapping (\ref{eq:rho_c(E) mapping}), observable in figure \ref{fig:mapping rho - E}, into the Lane \& Thomas definition (\ref{eq:: Def S = Re[L], P = Im[L]}).
Historically, the study of the properties emanating from this definition have neglected the $\big\{ E, - \big\}$ sheet. 
Importantly, it was recently proved that $\frac{\partial S_c}{\partial E}(E)  \geq 0$ is true for most cases \cite{Brune_Mark_monotonic_properties_of_shift_2018}. This proof did not consider the $\big\{ E, - \big\}$ sheet of mapping (\ref{eq:rho_c(E) mapping}).
However, their proof of $\frac{\partial S_c}{\partial E} (E) \geq 0$ should still stand on the $\big\{ E, + \big\}$ sheet.
Moreover, the proof of lemma \ref{lem:: analytic S_c and P_c lemma} establishes that all the discontinuity points, i.e. the real-energy poles, happen at sub-threshold energies, and in particular that neutral particles with odd angular moment introduce exactly one such sub-threshold discontinuity. 
This means that above threshold, both the shift $S_c(E)$ and penetration $P_c(E)$ functions are continuous. 
These behaviors are depicted in figure \ref{fig:Lane_and_Thomas_shift_and_penetration_factors_with_branch_points}.
Finally, one will notice that the $\big\{ E, + \big\}$ and $\big\{ E, - \big\}$ sheets coincide above threshold for the shift function $S_c(E)$, and below threshold for the penetration function $P_c(E)$.
For $P_c(E)$, this is because of property (\ref{eq:: Im[L] = 0 for E<0}).
For $S_c(E)$, this is because for real energies above threshold, both definitions (\ref{eq:: Def S = Re[L], P = Im[L]}) and (\ref{eq:: Def S and P analytic continuation from L}) coincide, and lemma \ref{lem:: analytic S_c and P_c lemma} shows the analytic continuation definition of $S_c(E)$ is function of $\rho_c^2(E)$, which unfolds the sheets of the Riemann mapping (\ref{eq:rho_c(E) mapping}). Hence, for above-threshold energies, this property still stands for the Lane \& Thomas definition of the shift factor $S_c(E)$.
\end{proof}

\begin{figure}[ht!!] 
  \centering
  \subfigure[\ Real part $\Re{\big[ L_{\ell} (E) \big]}$, on both $\big\{ E, \pm \big\}$ sheets. ]{\includegraphics[width=0.49\textwidth]{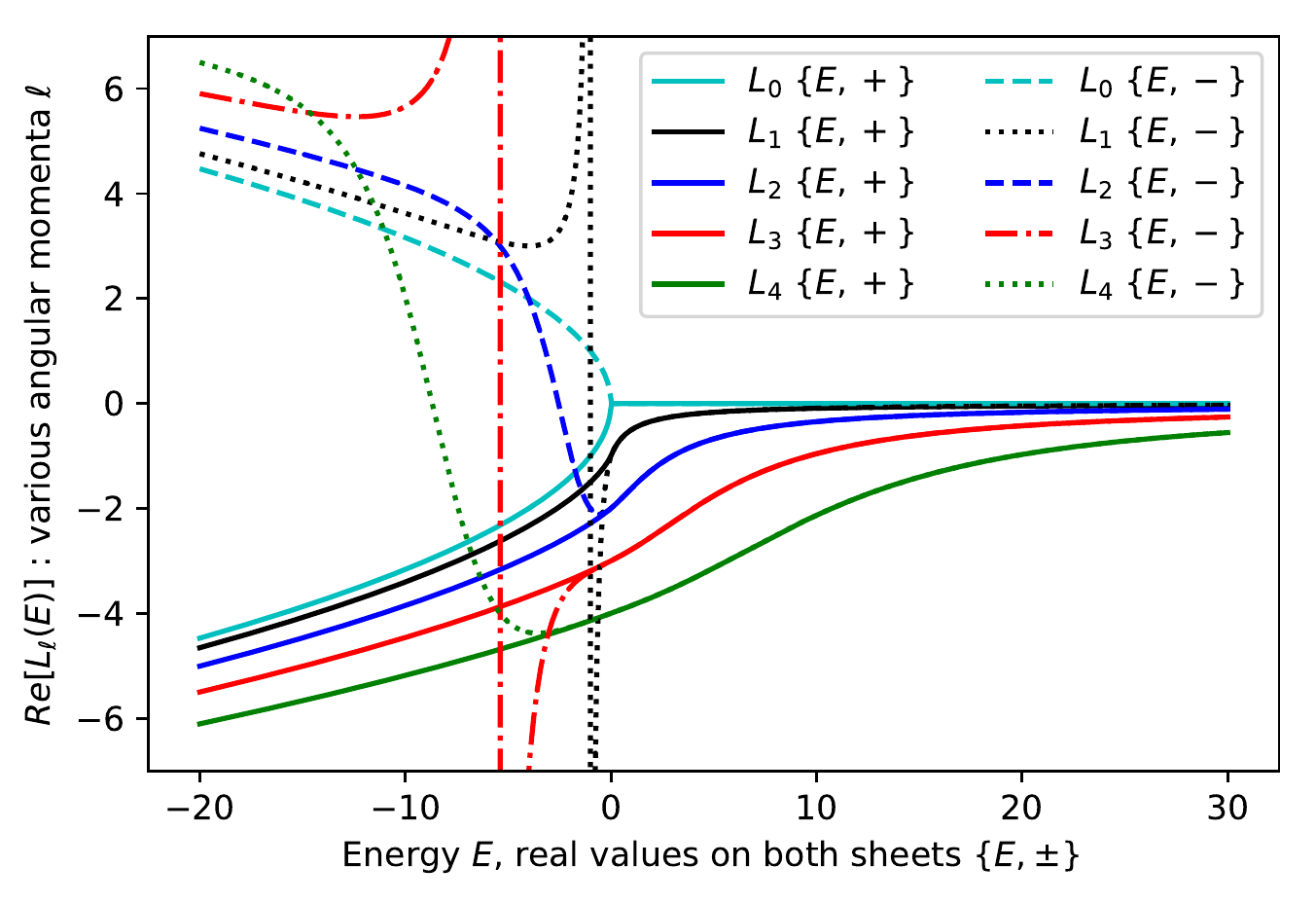}}
  \subfigure[\ Imaginary part $\Im{\big[ L_{\ell} (E) \big]}$, on both $\big\{ E, \pm \big\}$ sheets.]{\includegraphics[width=0.49\textwidth]{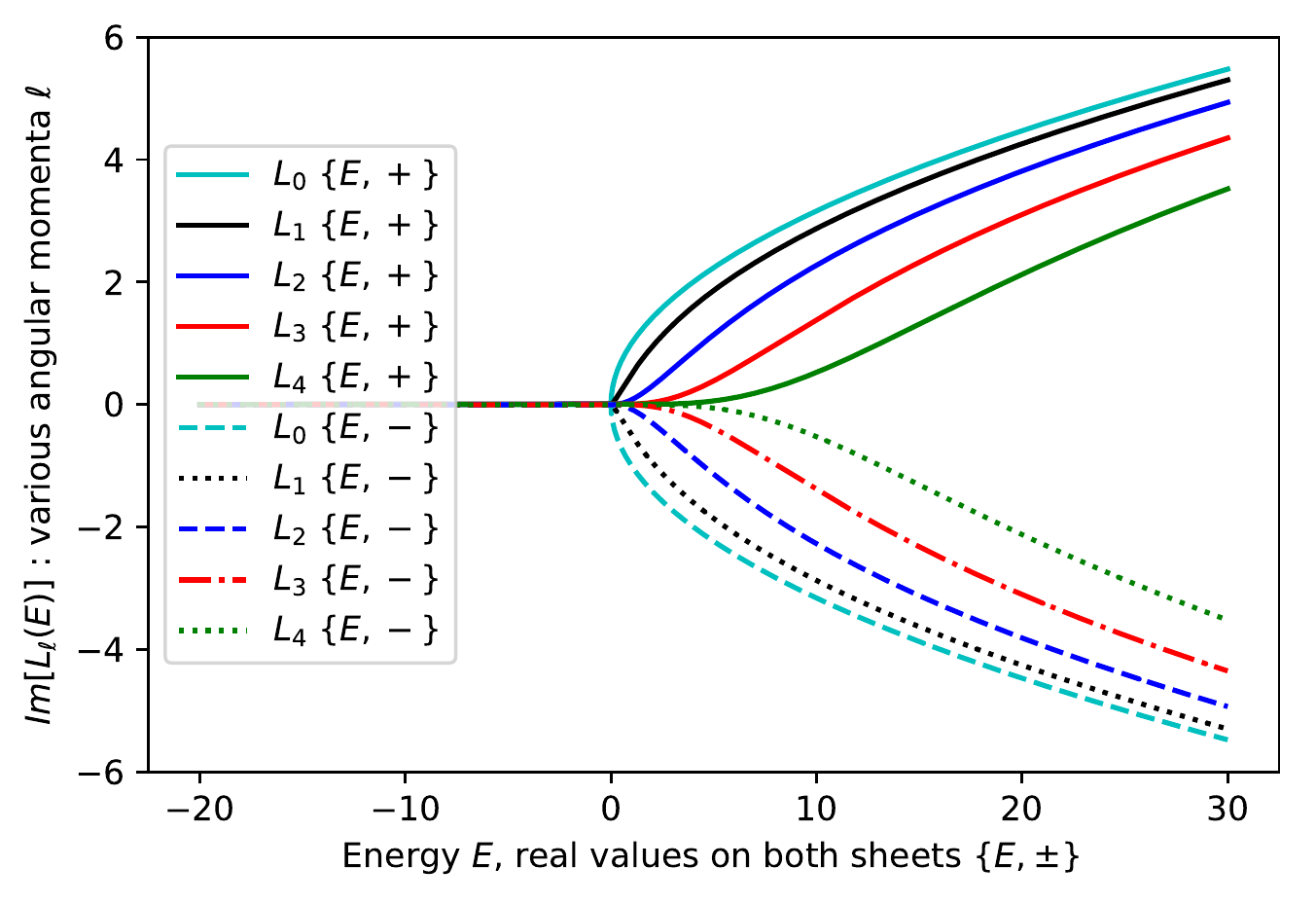}}
  \caption{\small{Real and imaginary parts of the reduced logarithmic derivative of the outgoing wavefunction $L_\ell(E)$, for semi-classical massive neutral particles (table \ref{tab::L_values_neutral}) energy-wavenumber mapping (\ref{eq:rho_c massive}), for different angular momenta $\ell \in \llbracket 1, 4 \rrbracket$. These real and imaginary parts were used by Lane \& Thomas to define the shift and penetration functions, $S_\ell(E) $ and $P_\ell(E)$, as (\ref{eq:: Def S = Re[L], P = Im[L]}). This definition commands branch points from mapping (\ref{eq:rho_c massive}) (c.f. figure \ref{fig:mapping rho - E}). $\Re\big[L_\ell(E)\big]$ presents sub-threshold discontinuities (for odd angular momenta $\ell$) and non-monotonic behavior (for even angular momenta $\ell$) below threshold on the $\big\{ E, - \big\}$ sheet. Units such that $\rho_0 = 1$. Threshold set at zero: $E_{T_c} = 0$.}}
  \label{fig:Lane_and_Thomas_shift_and_penetration_factors_with_branch_points}
\end{figure}

The second approach to defining the shift and penetration functions, $\boldsymbol{S}$ and $\boldsymbol{P}$, consists of performing analytic continuation of the scattering matrix $\boldsymbol{U}$ to complex energies $E\in\mathbb{C}$. This is implicit in the Kapur-Peierls or Siegert-Humblet expansions (c.f. \cite{Siegert, Humblet_thesis} and section  sections
IX.2.c-d-e p.297-298 of \cite{Lane_and_Thomas_1958}), and an abundant literature revolves around the analytic properties of the scattering matrix in the complex plane, including the vast Theory of Nuclear Reaction of Humblet and Rosenfeld \cite{Theory_of_Nuclear_Reactions_I_resonances_Humblet_and_Rosenfeld_1961,Theory_of_Nuclear_Reactions_II_optical_model_Rosenfeld_1961,Theory_of_Nuclear_Reactions_III_Channel_radii_Humblet_1961_channel_Radii,Theory_of_Nuclear_Reactions_IV_Coulomb_Humblet_1964,Theory_of_Nuclear_Reactions_V_low_energy_penetrations_Jeukenne_1965,Theory_of_Nuclear_Reactions_VI_unitarity_Humblet_1964,Theory_of_Nuclear_Reactions_VII_Photons_Mahaux_1965,Theory_of_Nuclear_Reactions_VIII_evolutions_Rosenfeld_1965,Theory_of_Nuclear_Reactions_IX_few_levels_approx_Mahaux_1965}, or the general unitarity condition on the multi-sheeted Riemann surface introduced by Eden and Taylor in \cite{Eden_and_Taylor}.
In this approach, energy dependence of the shift and penetration factors for positive energies are analytically continued into the complex plane, i.e.
\begin{equation}
\boldsymbol{S} : \left\{ \begin{array}{rcl}
\mathbb{C} & \mapsto & \mathbb{C} \\
E & \to & S_c(E)
\end{array}\right. \mathrm{s.t.} \; \; S(E) = S_c(E) , \; \forall (E-E_{T_c}) \in \mathbb{R}_+
\label{eq:: Def S analytic}
\end{equation}
so that they can be computed from the outgoing wavefunction reduced logarithmic derivative $\boldsymbol{L}$ by analytic continuation in wavenumber space $k_c \in \mathbb{C}$: 
\begin{equation}
\forall \rho_c \in \mathbb{C}\;, \left\{ \quad \begin{array}{cc}
     S_c(\rho_c) \triangleq & \frac{L_c(\rho_c) +\left[ L_c(\rho_c^*)\right]^* }{2}  \; \; \in \; \mathbb{C} \\
     P_c(\rho_c) \triangleq &  \frac{L_c(\rho_c) -\left[ L_c(\rho_c^*)\right]^* }{2\mathrm{i}} \; \; \in \; \mathbb{C}
\end{array}   \right.
\label{eq:: Def S and P analytic continuation from L}
\end{equation}
From this definition (\ref{eq:: Def S and P analytic continuation from L}), and using the recurrence relation (\ref{eq::L_ell recurrence formula}), one readily finds the expressions for the neutral particles shift and penetration factors documented in table \ref{tab::S_and_P_expressions_neutral}.
Critically, both definitions (\ref{eq:: Def S = Re[L], P = Im[L]}) and (\ref{eq:: Def S and P analytic continuation from L}) will yield the same shift $S_c(E)$ and penetration $P_c(E)$ functions for real energies above threshold $E \geq E_{T_c}$. 
Moreover, definition (\ref{eq:: Def S and P analytic continuation from L}) bestows interesting analytic properties onto the shift and penetration functions, here established in lemma \ref{lem:: analytic S_c and P_c lemma}.

\begin{table*}
\caption{\label{tab::S_and_P_expressions_neutral} Shift $S_\ell(\rho)$, $S_\ell^0(\rho) \triangleq S_\ell(\rho) - B_\ell$ using $B_\ell = - \ell$, and $P_\ell(\rho)$ irreducible forms for neutral particles, for angular momenta $0 \leq \ell \leq 4$, all defined from analytic continuation (\ref{eq:: Def S and P analytic continuation from L}).}
\begin{ruledtabular}
\begin{tabular}{c|c|c|c}
\ \ & $S_\ell(\rho)$ & $S_\ell^0(\rho) \triangleq S_\ell(\rho) - B_\ell$ (recurrence for $B_\ell = - \ell$ ) & $P_\ell(\rho)$  \tabularnewline
\hline
$\ell$  &  $ S_\ell(\rho) = \frac{\rho^2 \left(\ell - S_{\ell-1}(\rho)  \right)}{\left(\ell - S_{\ell-1}(\rho)  \right)^2 + P_{\ell-1}(\rho)^2} - \ell $  & $ S_\ell^0(\rho) \triangleq S_\ell(\rho) + \ell =  \frac{\rho^2 \left(2\ell - 1 - S_{\ell-1}^0(\rho)  \right)}{\left(2\ell -1 - S_{\ell-1}^0(\rho)  \right)^2 + P_{\ell-1}(\rho)^2}$   & $ P_\ell(\rho) = \frac{\rho P_{\ell-1}(\rho) }{\left(\ell - S_{\ell-1}(\rho)  \right)^2 + P_{\ell-1}(\rho)^2}$  \tabularnewline
\hline \hline
0  &  $0$  & $0$ & $\rho$ \tabularnewline
1  &  $ - \frac{1}{1+\rho^2}$ & $\frac{\rho^2}{1+\rho^2}$ & $\frac{\rho^3}{1+\rho^2}$ \tabularnewline
2   &  $ - \frac{18 + 3 \rho^2}{9 + 3\rho^2 + \rho^4}$ & $\frac{ 3 \rho^2 + 2 \rho^4}{9 + 3\rho^2 + \rho^4}$ & $\frac{\rho^5}{9 + 3\rho^2 + \rho^4}$ \tabularnewline
3  &  $ - \frac{675 + 90 \rho^2 + 6\rho^4}{225 + 45\rho^2 + 6\rho^4 + \rho^6}$ & $\frac{45 \rho^2 + 12\rho^4 + 3 \rho^6}{225 + 45\rho^2 + 6\rho^4 + \rho^6}$ &  $\frac{\rho^7}{225 + 45\rho^2 + 6\rho^4 + \rho^6}$ \tabularnewline
4  &  $ - \frac{44100 + 4725 \rho^2 + 270\rho^4 + 10 \rho^6}{11025 + 1575\rho^2 + 135\rho^4 + 10\rho^6 + \rho^8}$ & $ \frac{1575 \rho^2 + 270\rho^4 + 30 \rho^6 + 4\rho^8}{11025 + 1575\rho^2 + 135\rho^4 + 10\rho^6 + \rho^8}$ &  $\frac{\rho^9}{11025 + 1575\rho^2 + 135\rho^4 + 10\rho^6 + \rho^8}$\tabularnewline
\end{tabular}
\end{ruledtabular}
\end{table*}

\begin{lem}\label{lem:: analytic S_c and P_c lemma}
\textsc{Analytic continuation definition of shift $S_c(E)$ and penetration $P_c(E)$ functions}. \\
When defined by analytic continuation (\ref{eq:: Def S and P analytic continuation from L}), the shift function, $S_c(\rho)$, satisfies the Mittag-Leffler expansion:
\begin{equation}
\begin{IEEEeqnarraybox}[][c]{rcl}
    S_c(\rho) & \;  = \; & - \ell + \underset{\mathrm{arg}(\omega_n)\in \left[-\frac{\pi}{2},0\right]}{\sum_{n \geq 1}} \frac{\rho^2}{\rho^2 - \omega_n^2} + \frac{\rho^2}{\rho^2 - {\omega_n^*}^2}
\IEEEstrut\end{IEEEeqnarraybox}
\label{eq::S_c expansion in rho^2}
\end{equation}
where the poles $\big\{ \omega_n \big\}$ are only the lower-right-quadrant roots -- i.e. such that $\mathrm{arg}(\omega_n) \in [-\frac{\pi}{2},0]$ -- of the outgoing wave function $O_c(\rho_c)$.
In the neutral particles cases, these are reported in table \ref{tab::roots of the outgoing wave functions}.
Given $\rho_c(E)$ mapping (\ref{eq:rho_c(E) mapping}), this entails $S_c(E)$:
\begin{itemize}
    \item unfolds the sheets of $\rho_c(E)$ mapping (\ref{eq:rho_c(E) mapping}),
    \item is purely real for real energies: $\forall E\in\mathbb{R}, \; S_c(E) \in \mathbb{R}$.
\end{itemize}
The penetration function, $P_c(\rho)$, satisfies the Mittag-Leffler expansion:
\begin{equation}
\begin{IEEEeqnarraybox}[][c]{rcl}
    P_c(\rho) & \;  = \; & \rho \Bigg[ 1 - \mathrm{i} \underset{\mathrm{arg}(\omega_n)\in \left[- \frac{\pi}{2},0\right]}{\sum_{n \geq 1}} \frac{\omega_n}{\rho^2 - \omega_n^2} - \frac{{\omega_n^*}}{\rho^2 - {\omega_n^*}^2} \Bigg]
\IEEEstrut\end{IEEEeqnarraybox}
\label{eq::P_c expansion in rho^2}
\end{equation}
which in turn entails that $P_c(E)$:
\begin{itemize}
    \item is purely real for above threshold energies: $\forall E > E_{T_c}, \;  P_c(E) \in \mathbb{R}$,
    \item is purely imaginary for sub-threshold energies: $\forall E < E_{T_c}, \;  P_c(E) \in \mathrm{i}\mathbb{R}$,
\end{itemize}
In the neutral particles case, Mittag-Leffler expansions (\ref{eq::S_c expansion in rho^2}) and (\ref{eq::P_c expansion in rho^2}) are the partial fraction decompositions of the rational fractions reported in table \ref{tab::S_and_P_expressions_neutral}, and for all odd angular momenta $\ell_c \equiv 1 \; (\mathrm{mod} \; 2)$, both have one, shared, real sub-threshold pole.
\end{lem}

\begin{proof}
The proof uses theorem \ref{theo::Mittag-Leffler of L_c theorem}, where we establish the Mittag-Leffler expansion (\ref{eq::Mittag-Leffler expansion of L_c}) of the reduced logarithmic derivative $L_c(\rho_c)$.
We recall the conjugacy relations of the outgoing and incoming wavefunctions (eq. (2.12), VI.2.c. in \cite{Lane_and_Thomas_1958}), whereby, for any channel $c$:
\begin{equation}
\begin{IEEEeqnarraybox}[][c]{rclcrcl}
\big[O_c(k_c^*)\big]^* & \ = \ & I_c(k_c) \quad  &,& \quad \big[I_c(k_c^*)\big]^* & \ = \ & O_c(k_c) \\
O_c(-k_c) & \ = \ & I_c(k_c) \quad &,& \quad I_c(-k_c) & \ = \ & O_c(k_c) \\ -O_c^{(1)}(-k_c) & \ = \ & I_c^{(1)}(k_c) \quad &,& \quad -I_c^{(1)}(-k_c) & \ = \ & O_c^{(1)}(k_c) \\
\IEEEstrut\end{IEEEeqnarraybox}
\label{eq:Conjugacy for O and I}
\end{equation}
where the third line was obtained by taking the derivative of the second. 
Properties (\ref{eq:Conjugacy for O and I}) on the poles $\big\{ \omega_n \big\}$ mean each pole $\omega_n$ on the lower right quadrant of the complex plane -- i.e. such that $\mathrm{arg}(\omega_n) \in [-\frac{\pi}{2},0]$ -- induces a specular pole $-\omega_n^*$.
Dividing the poles in specular pairs, we can re-write the Mittag-Leffler expansion (\ref{eq::Mittag-Leffler expansion of L_c}) as:
\begin{equation}
\begin{IEEEeqnarraybox}[][c]{rcl}
    L_c(\rho) & \;  = \; & - \ell + \mathrm{i}\rho + \underset{\mathrm{arg}(\omega_n)\in \left[-\frac{\pi}{2},0\right]}{\sum_{n \geq 1}} \frac{\rho}{\rho - \omega_n} + \frac{\rho}{\rho + {\omega_n^*}}
\IEEEstrut\end{IEEEeqnarraybox}
\label{eq::L_c Mittag-Lellfer expansion in pairs}
\end{equation}
Plugging-in expression (\ref{eq::L_c Mittag-Lellfer expansion in pairs}) into the shift function definition (\ref{eq:: Def S and P analytic continuation from L}) readily yields (\ref{eq::S_c expansion in rho^2}) and (\ref{eq::P_c expansion in rho^2}).

Note that (\ref{eq::S_c expansion in rho^2}) unfolds the Riemann surface of mapping (\ref{eq:rho_c(E) mapping}), whereas (\ref{eq::P_c expansion in rho^2}) factors-out the branch points so that all its branches are symmetric.
In (\ref{eq::P_c expansion in rho^2}) we recognize the odd powers of $\rho$ in the neutral particles case of table \ref{tab::S_and_P_expressions_neutral}, which do not unfold the Riemann sheets of mapping (\ref{eq:rho_c(E) mapping}).
These behaviors are illustrated in figure \ref{fig:Analytic_continuation_S_and_P}.

In the neutral particles case, $L_c$ is a rational fraction in $\rho_c$, and its denominator is of degree $\ell_c$, as can be observed in table \ref{tab::L_values_neutral}, thus inducing $\ell_c $ poles, reported in table \ref{tab::roots of the outgoing wave functions}.
Since these poles $\big\{ \omega_n \big\}$ must respect the specular symmetry: $ \omega \longleftrightarrow -\omega_n^*$; it thus entails that these poles come in symmetric pairs. 
For neutral particles, odd angular momenta mean there is an odd number of poles $\big\{ \omega_n \big\}$. For them to come in pairs thus imposes one is exactly imaginary $\omega_n = - \mathrm{i} x_n$, with $x_n \in \mathbb{R}_+$.
When squared, this purely imaginary pole will introduce a real energy sub-threshold pole in both (\ref{eq::S_c expansion in rho^2}) and (\ref{eq::P_c expansion in rho^2}), through: $ \frac{1}{\rho^2 + x_n^2}$.
\end{proof}

\begin{figure}[ht!!] 
\center
\subfigure[\ Analytic shift function $S_\ell(E)$, for both $\big\{ E, \pm \big\}$ sheets.]{\includegraphics[width=0.478\textwidth]{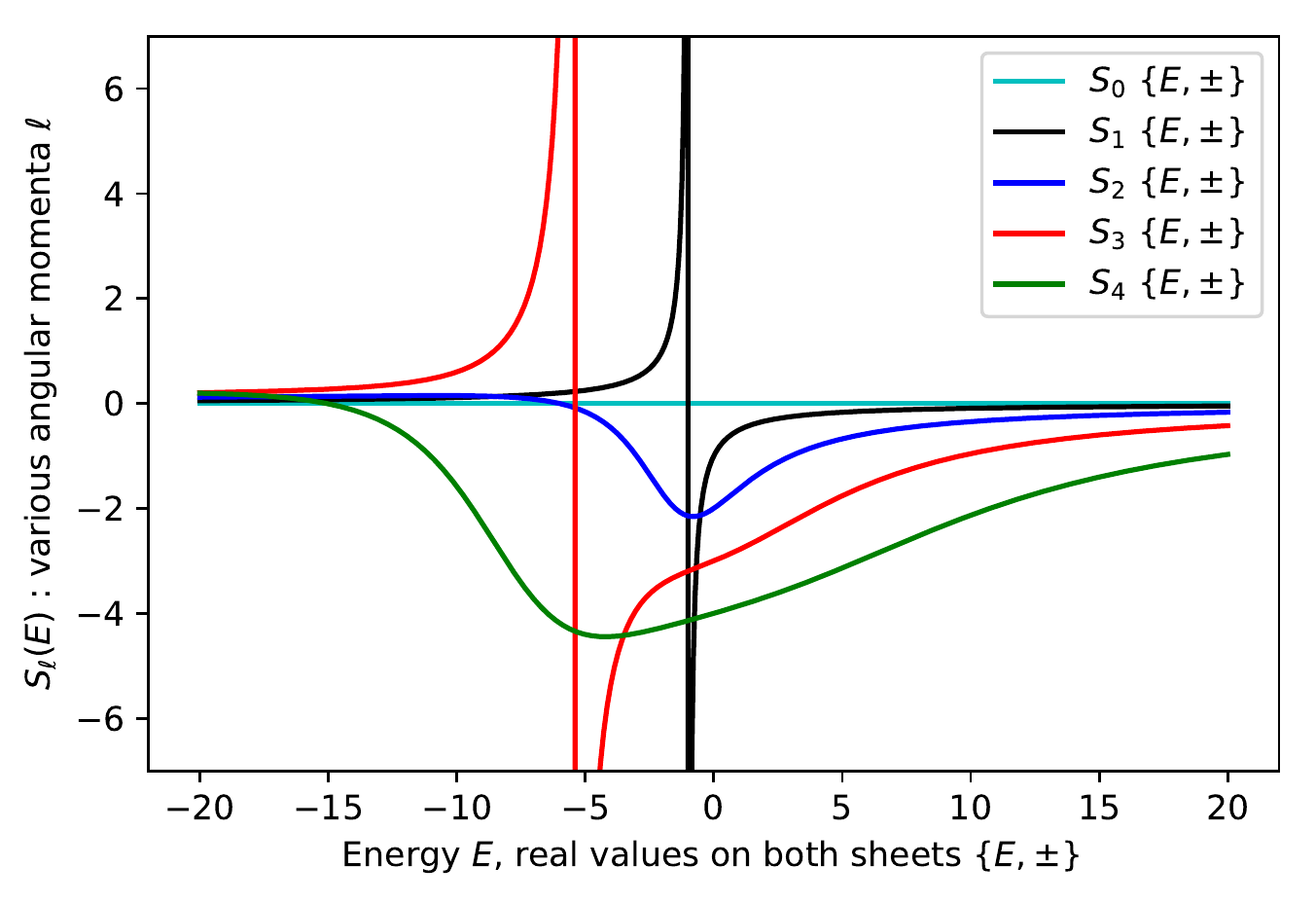}} 
\subfigure[\ Analytic penetration function real part $\Re{\left[ P_\ell(E) \right]}$.]{\includegraphics[width=0.478\textwidth]{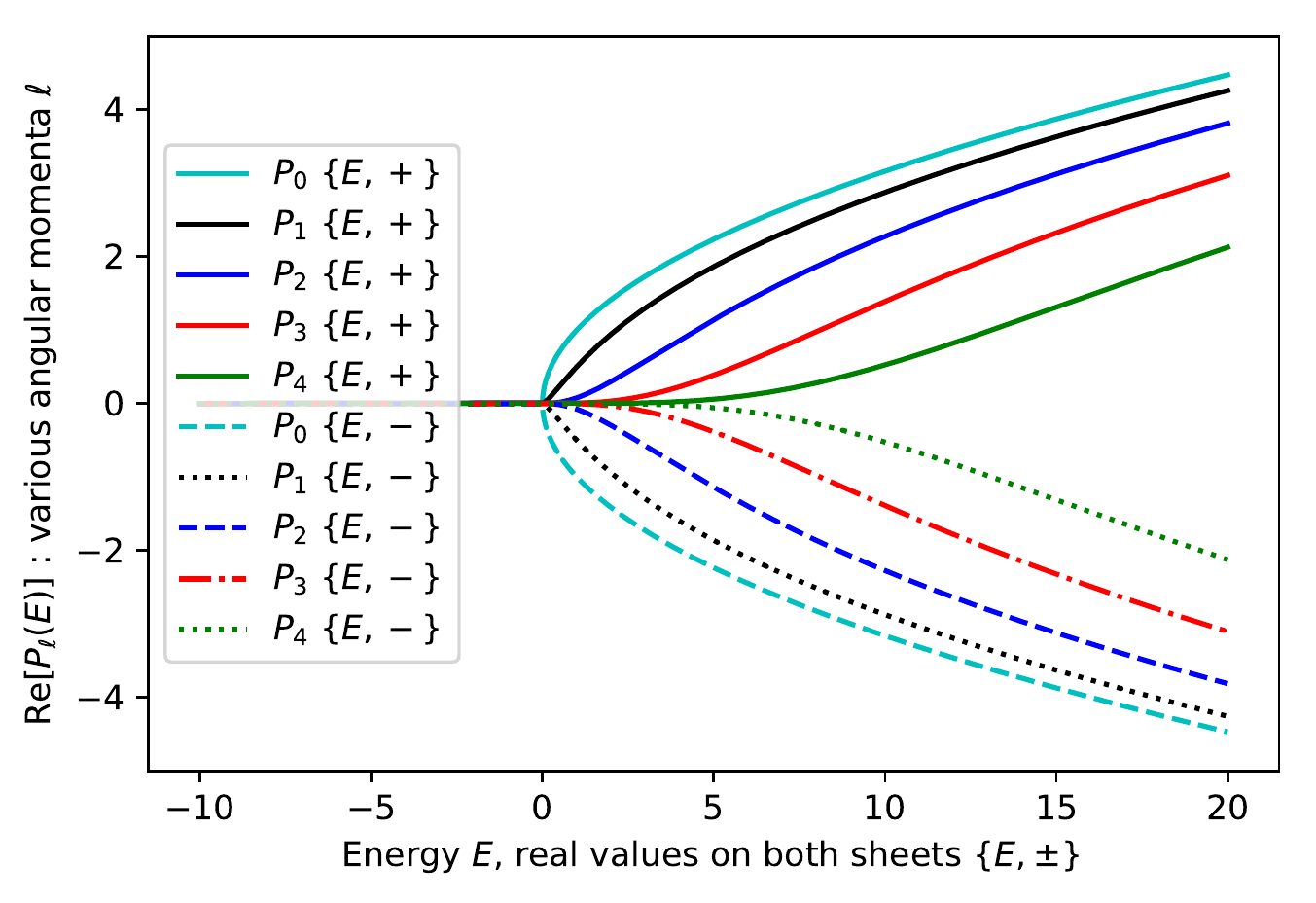}} 
\subfigure[\ Analytic penetration function imaginary part $\Im{\left[ P_\ell(E) \right] }$.]{\includegraphics[width=0.478\textwidth]{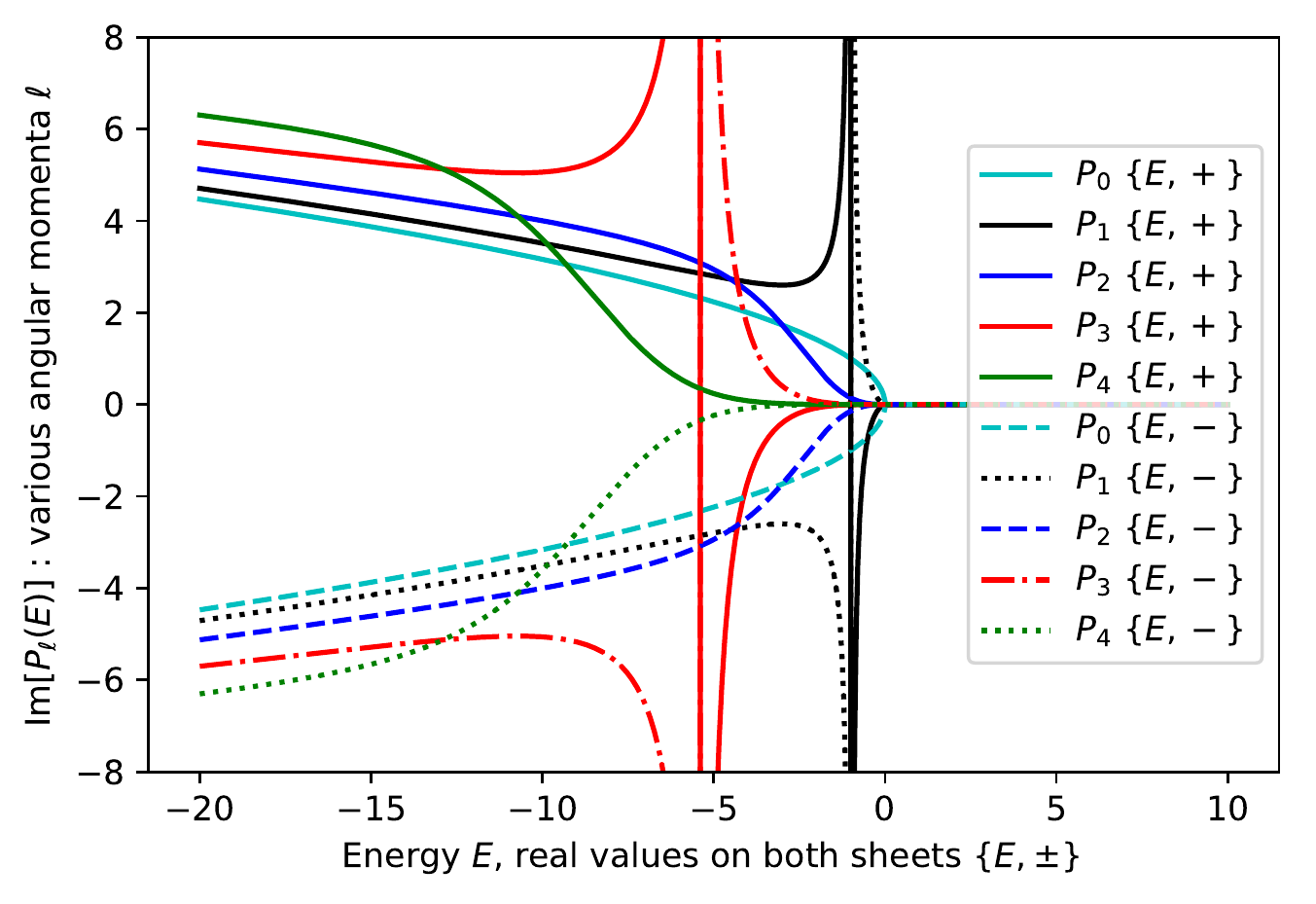}}
\caption{\small{Analytic shift $S_\ell(E)$ and penetration $P_\ell(E)$ functions, definition (\ref{eq:: Def S and P analytic continuation from L}), non-relativistic neutral particles (table \ref{tab::S_and_P_expressions_neutral}) $k_c(E)$ mapping (\ref{eq:rho_c massive}), angular momenta $\ell \in \llbracket 0, 4 \rrbracket$, zero threshold $E_{T_c} = 0$. The shift function $S_c(E)$ has no $\big\{ E, \pm \big\}$ branches (proof in lemma \ref{lem:: analytic S_c and P_c lemma}), has discontinuities (odd $\ell$) and non-monotonic behavior (even $\ell$) below threshold. $P_\ell(E)$ has $\big\{ E, \pm \big\}$ branches, is purely real above threshold, and purely imaginary below. Units such that $\rho_0 = 1$.}} 
\label{fig:Analytic_continuation_S_and_P}
\end{figure}


An example to illustrate the difference between definitions (\ref{eq:: Def S = Re[L], P = Im[L]}) and (\ref{eq:: Def S and P analytic continuation from L}) is depicted in figures \ref{fig:Lane_and_Thomas_shift_and_penetration_factors_with_branch_points} and \ref{fig:Analytic_continuation_S_and_P}.
Consider the elemental case of a neutron channel with angular momentum $\ell_c = 1$, and let $\rho_0$ be the proportionality constant so that (\ref{eq:rho_c massive}) is written $\rho(E) = \pm \rho_0 \sqrt{E - E_{T_c}}$. Let us also set a zero threshold $E_{T_c} = 0$, for simplicity.

In this case, the legacy Lane \& Thomas definition (\ref{eq:: Def S = Re[L], P = Im[L]}) corresponds to taking $S(E) \triangleq S_c(\rho_c(E)) = -\frac{1}{1+\rho_c^2}$ for above-threshold energies $E\geq E_{T_c}$, and switch to $S(E) \triangleq L_c(\rho_c(E)) = \frac{- 1 + \mathrm{i}\rho_c + \rho_c^2 }{1- \mathrm{i}\rho_c}$ for sub-threshold energies $E < E_{T_c}$.
Since the (\ref{eq:rho_c massive}) mapping $\rho(E) = \pm \rho_0 \sqrt{E - E_{T_c}}$ has two sheets, this means definition (\ref{eq:: Def S = Re[L], P = Im[L]}) entails: $S(E) \triangleq S_c(E) = -\frac{1}{1+\rho_0^2 E}$ for $E\geq E_{T_c}$, and $S(E) \triangleq L_c(E) = \frac{- 1 \pm \mathrm{i}\rho_0 \sqrt{E} + \rho_0^2 E }{1  \mp \mathrm{i}\rho_0 \sqrt{E}}$ for $E < E_{T_c}$, which is a real quantity. Definition (\ref{eq:: Def S = Re[L], P = Im[L]}) thus introduces the ramifications reported in figure \ref{fig:Lane_and_Thomas_shift_and_penetration_factors_with_branch_points}. In particular, the full cyan line ($L_0 \left\{E,+\right\}$ case) of our $\Re\big[L_\ell(E)\big]$ plot (figure \ref{fig:Lane_and_Thomas_shift_and_penetration_factors_with_branch_points}) corresponds to the uncharged case for angular momentum $\ell = 0$ reported as a full black curve in FIG.1, p.6 of \cite{Brune_Mark_monotonic_properties_of_shift_2018}. 
Notice that all the $\big\{ E,+ \big\}$ sheet curves are continuous and monotonically increasing ($\frac{\partial S_c}{\partial E} \geq 0$), which is in accordance to the monotonic properties established in  \cite{Brune_Mark_monotonic_properties_of_shift_2018}. 
However, on the $\big\{E,-\big\}$ sheet below threshold, $\Re\big[L_c(E)\big]$ is no longer monotonic for even angular momenta ($\frac{\partial \Re\big[L_c(E)\big]}{\partial E} \geq 0$ does not hold), and is discontinuous in the case of odd angular momenta.

In contrast, for our same elemental case, the analytic continuation definition (\ref{eq:: Def S and P analytic continuation from L}) simply defines $S(E) \triangleq S_c(\rho_c(E)) = -\frac{1}{1+\rho_c^2}$ for all real or complex energies $E \in \mathbb{C}$, that is $S(E) \triangleq -\frac{1}{1+\rho_0^2 E}$. The later happens to have a real pole, which introduces a discontinuity, at $E_{\mathrm{dis.}} = -\frac{1}{\rho_0^2}$, as can be seen in figure \ref{fig:Analytic_continuation_S_and_P}.
One can observe that all odd angular momenta are monotonous but have a real sub-threshold pole.
For even angular momenta, $S_\ell(E)$ is continuous, monotonically increasing above-threshold, but $\frac{\partial S}{\partial E}(E) \geq 0$ does not hold below-threshold.
For the penetration function $P_c(E)$, each ramification is monotonous, but in opposite, mirror direction. 
In figure \ref{fig:Analytic_continuation_S_and_P}, the shift function $S_c(E)$ does not present branch points, as proved in lemma \ref{lem:: analytic S_c and P_c lemma}: it is a function of $\rho^2$ so no $\pm \sqrt{\cdot}$ choice is necessary in $\rho_c(E)$ mapping (\ref{eq:rho_c EDA}).

\subsection{\label{subsubsec::Number of Brune poles}Number of alternative poles: \\ existence of shadow poles}

Definitions (\ref{eq:: Def S = Re[L], P = Im[L]}) and (\ref{eq:: Def S analytic}) have a major impact on the alternative parameters (\ref{eq:Brune parameters}): they command that the number $N_S$ of alternative poles $\left\{\widetilde{E_i}\right\}$, solutions to Brune's generalized eigenproblem (\ref{eq:Brune eigenproblem}), is greater than the $N_\lambda$ previously found in \cite{Brune_2002}: i.e. $N_S \geq N_\lambda$.
And this is regardless of whether definition (\ref{eq:: Def S = Re[L], P = Im[L]}) or (\ref{eq:: Def S analytic}) is chosen for the shift factor $S_c(E)$ when searching for these solutions.

The fundamental reason for this is that Brune's three-step monotony argument, which elegantly proved in \cite{Brune_2002} that there are exactly $N_\lambda$ solutions to (\ref{eq:Brune eigenproblem}) and which we here recall in the last paragraph of section \ref{sec:R_S def}, rests on two hypotheses on the shift function $S_c(E)$: 1) it is continuous (i.e. has no real poles), and; 2) it is monotonously increasing, i.e. $\frac{ \partial S_c}{\partial E} \geq 0$.
In \cite{Brune_Mark_monotonic_properties_of_shift_2018}, these two hypotheses have just been proved to hold true for energies above threshold $E \geq E_{T_c}$, i.e. for real wavenumbers $k_c \in \mathbb{R}$.
Yet, we just established in lemmas \ref{lem:: Lane and Thomas S_c ramification properties} and \ref{lem:: analytic S_c and P_c lemma} that proper accounting of the multi-sheeted nature of the Riemann surface created by mapping (\ref{eq:rho_c(E) mapping}) shows these two hypotheses do not hold for sub-threshold energies $E < E_{T_c}$, where the wavenumber is purely imaginary from mapping (\ref{eq:rho_c massive}). 
This engenders additional solutions to Brune's generalized eigenproblem (\ref{eq:Brune eigenproblem}), so that the number $N_S$ of alternative poles $\left\{\widetilde{E_i}\right\}$ is in fact greater than the number of channels: $N_S \geq N_\lambda$.
So how many $N_S$ solutions are there? 
This depends on the R-matrix parameters and on the definition chosen for the shift function $S_c(E)$, as we now show in theorems \ref{theo::branch_brune_poles} and \ref{theo::analytic_Brune_poles}, for definitions (\ref{eq:: Def S = Re[L], P = Im[L]}) and (\ref{eq:: Def S analytic}), respectively.

\begin{theorem}\label{theo::branch_brune_poles} \textsc{Alternative Branch Poles}. \\
Let the alternative branch poles $\left\{\widetilde{E_i}\right\}$ be the solutions of the Brune generalized eigenproblem (\ref{eq:Brune eigenproblem}), using the legacy Lane \& Thomas definition (\ref{eq:: Def S = Re[L], P = Im[L]}) for the shift $S_c(E)$, and let $N_S$ be the number of such solutions, then:
\begin{itemize}
    \item all the alternative branch poles are real, and reside on the $2^{N_c}$ sheets of the Riemann surface from (\ref{eq:rho_c(E) mapping}) mapping: $\left\{\widetilde{E_i}, \underbrace{\pm , \hdots , \pm  }_{N_c}\right\} \in \mathbb{R}^{N_S}$,
\item exactly $N_\lambda$ alternative branch poles are present on the $\big\{ E, \underbrace{+ , \hdots , + }_{N_c} \big\}$ sheet of mapping (\ref{eq:rho_c(E) mapping}): these are the principal (or resonant) poles,
    \item additional alternative branch shadow poles can be found below threshold, $E<E_{T_c}$, on the $\big\{ E, - \big\}$ sheets of mapping (\ref{eq:rho_c(E) mapping}), depending on the values of the resonance parameters $\big\{E_\lambda, \gamma_{\lambda c}, B_c, E_{T_c}, a_c \big\}$ -- though in a way that is invariant under change of boundary-condition $B_c$,
    \item each neutral particle, odd angular momentum $\ell_c \equiv 1 \; (\mathrm{mod}\; 2)$, channel adds at least one alternative branch shadow pole below threshold on its $\big\{E, - \big\}$ sheet,
\end{itemize}
so that the total number $N_S^{\pm}$ of alternative branch poles on all sheets of mapping (\ref{eq:rho_c(E) mapping}) is greater or equal to the number $N_\lambda$ of levels: $N_S^{\pm} \geq N_\lambda$.
\end{theorem}

\begin{proof}
Let us go about solving the Brune generalized eigenproblem (\ref{eq:Brune eigenproblem}), following the three-step argument of Brune (c.f. last paragraph of section \ref{sec:R_S def}).
We consider the left-hand side of (\ref{eq:Brune eigenproblem}).
According to definition (\ref{eq:: Def S = Re[L], P = Im[L]}), the shift function is always real, even for complex wavenumbers $k_c \in \mathbb{C}$. 
Since by construction the Wigner-Eisenbud R-matrix parameters $\big\{E_\lambda, \gamma_{\lambda c}, B_c, E_{T_c}, a_c \big\}$ are also all real, this implies the right-hand side must be real to solve (\ref{eq:Brune eigenproblem}).
Thus, all the alternative branch poles from definition (\ref{eq:: Def S = Re[L], P = Im[L]}) are real.
To find them, we follow Brune's approach: for any energy $E$, on any of the $2^{N_c}$ sheets of mapping (\ref{eq:rho_c(E) mapping}), the left-hand side is a real symmetric matrix, and its eigenvalue decomposition will thus yield $N_\lambda$ real eigenvalues: $\big\{\widetilde{E_i}(E)\big\} \in \mathbb{R}$. 
We then have to vary the $E$ value until these real eigenvalues cross the $E=E$ identity line in the right-hand side. 
In general, the full accounting of all the Riemann sheets from mapping (\ref{eq:rho_c(E) mapping}) will entail solutions of the generalized Brune eigenproblem (\ref{eq:Brune eigenproblem}) on all sheets.
These alternative branch poles should thus be reported with the choice of sheet from the mapping (\ref{eq:rho_c(E) mapping}) for each channel: $\left\{\widetilde{E_i}, +, -, \hdots, + \right\} $.

We state in lemma \ref{lem:: Lane and Thomas S_c ramification properties} than on the $\big\{ E, + \big\}$ sheet, $S_c(E)$ is indeed continuous and monotonously increasing. 
We can thus apply Brune's three-step argument: the $N_\lambda$ eigenvalues of the left-hand side of (\ref{eq:Brune eigenproblem}) will satisfy $\frac{ \partial \widetilde{E_i}}{\partial E} (E) \leq 0$, and thus each and every one of them will eventually cross the $E=E$ identity line exactly once as $E$ varies continuously. 
On the $\big\{ E, + \big\}$ sheet for all channels, there are thus exactly $N_\lambda $ alternative poles: $\left\{\widetilde{E_i}, \underbrace{+ , \hdots , +  }_{N_c}\right\} \in \mathbb{R}^{N_\lambda}$

However, we showed in lemma \ref{lem:: Lane and Thomas S_c ramification properties} that $S_c(E)$ is not monotonous and can be discontinuous for sub-threshold energies $E < E_{T_c}$ on the $\big\{ E, - \big\}$ sheet.
So how many alternative poles are there on all sheets? 
Unfortunately, the number of solutions to Brune's generalized eigenproblem (\ref{eq:Brune eigenproblem}) will depend on the values of the resonance parameters $\big\{E_\lambda, \gamma_{\lambda c}, B_c, E_{T_c}, a_c \big\}$ -- though in a way that is invariant under change of boundary-condition $B_c$, as made evident in (\ref{eq:R_S by Brune det search}) when considering invariance (\ref{eq:: R_B invariance for B'}).
That the number of solutions to (\ref{eq:Brune eigenproblem}) depends on the parameters can be observed in figure \ref{fig:Brune_argument_even_angular_momenta}. 

For neutral particles odd momenta $\ell_c \equiv 1 \; (\mathrm{mod} \; 2)$ channels, lemma \ref{lem:: Lane and Thomas S_c ramification properties} also showed there exist exactly one sub-threshold pole to $S_c(E)$ on the $\big\{ E, - \big\}$ sheet of mapping (\ref{eq:rho_c(E) mapping}).
This pole will automatically cross the $E=E$ identity line of Brune's three-step argument twice, once below and once above threshold, adding an additional alternative shadow pole to the $N_\lambda$ ones Brune found in \cite{Brune_2002}.
This proves that there exists alternative shadow poles, just as shadow poles in the Siegert-Humblet parameters were revealed by G.Hale in \cite{Hale_1987, anti-Hale_1987}. 
This behavior is illustrated in figure \ref{fig:Brune_argument_odd_angular_momenta}.
\end{proof}

\begin{figure}[ht!!] 
  \centering
  \subfigure[\ Positive resonance energy $E_\lambda = 2$.]{\includegraphics[width=0.49\textwidth]{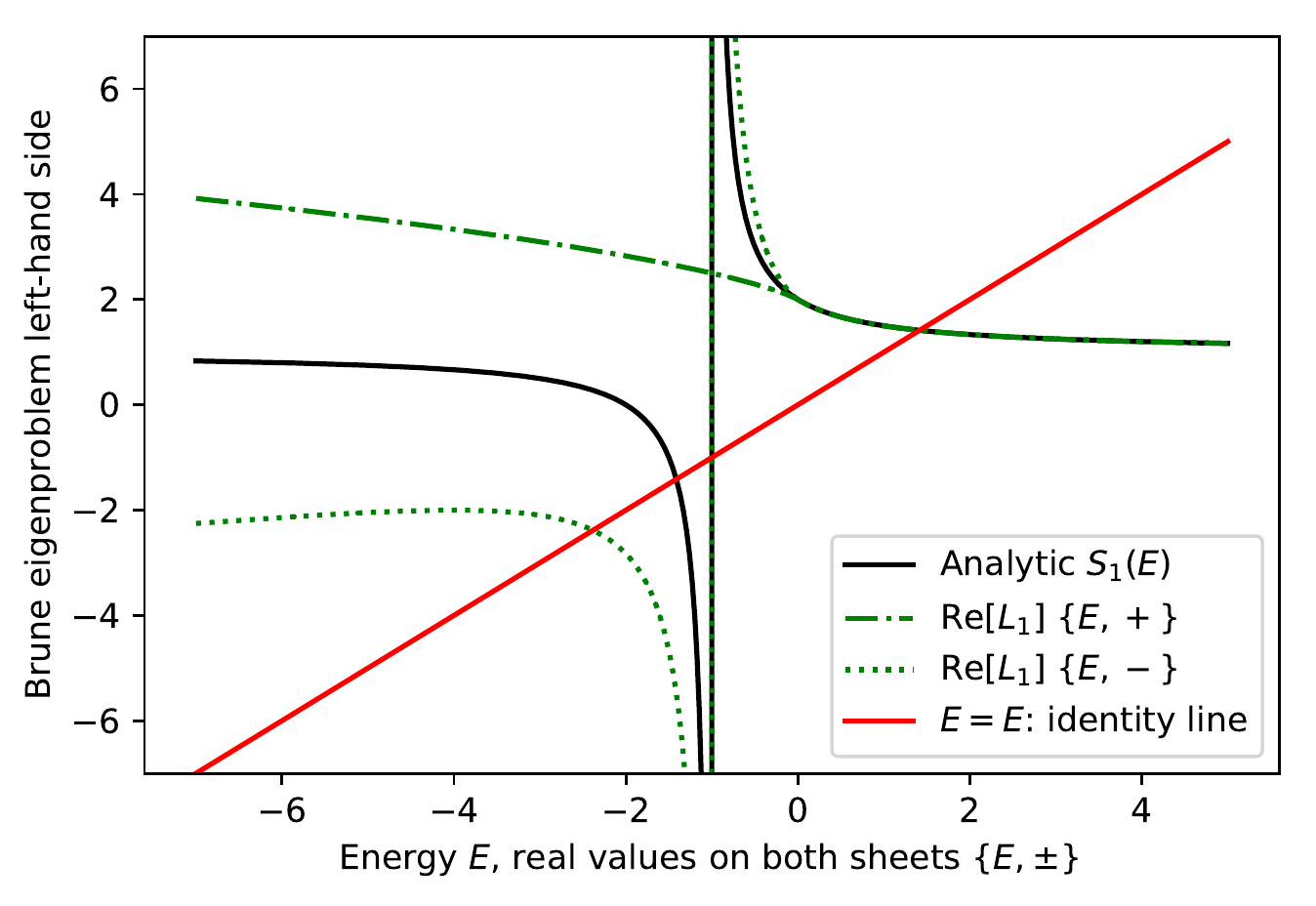}}
  \subfigure[\ Negative resonance energy $E_\lambda = -2$.]{\includegraphics[width=0.49\textwidth]{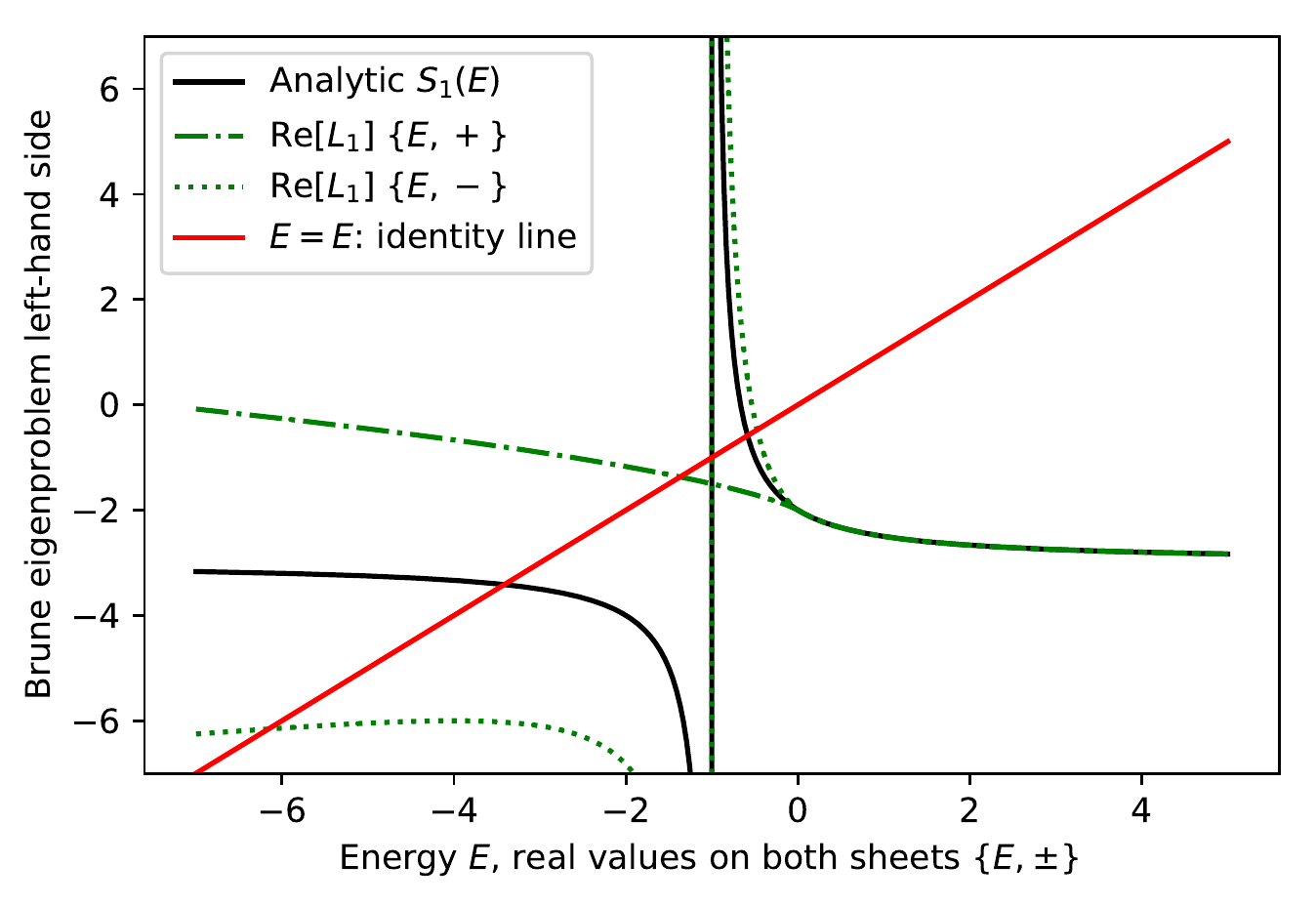}}
  \caption{\small{Brune eigenproblem (\ref{eq:: elemental Brune problem}) for 1-level 1-channel p-wave: \\ comparison of real solutions from definitions (\ref{eq:: Def S = Re[L], P = Im[L]}) versus (\ref{eq:: Def S and P analytic continuation from L}), for angular momentum $\ell_c = 1$, neutral particles energy-wavenumber mapping (\ref{eq:rho_c massive}), resonance width $\gamma_{\lambda,c} = 1$, using $B_c = - \ell_c$ convention and zero threshold $E_{T_c}= 0$. Units such that $\rho_0 = 1$. Since both have a real sub-threshold poles, both will yield two solutions (crossing the E=E diagonal), one above and one below the discontinuity. If at threshold energy $E_{T_c}$ the left hand side of (\ref{eq:: elemental Brune problem}) is above the E=E diagonal, then the above-threshold solutions from both definitions coincide. In any case, the sub-threshold solutions differ. Behavior is analogous for all odd angular momenta $\ell_c \equiv 1 (\mathrm{mod} 2)$.}}
  \label{fig:Brune_argument_odd_angular_momenta}
\end{figure}

Theorem \ref{theo::branch_brune_poles} establishes the existence of sub-threshold alternative shadow poles when the legacy Lane \& Thomas definition (\ref{eq:: Def S = Re[L], P = Im[L]}) is chosen for the shift function $S_c(E)$. 
If instead the analytic continuation definition (\ref{eq:: Def S analytic}) is chosen, we now show in theorem \ref{theo::analytic_Brune_poles} that this unfolds the Riemann surface for the shift function $S_c(E)$ so that no branch points are required to define the alternative analytic parameters.
We argue in a follow-up article that the analytic continuation approach (\ref{eq:: Def S and P analytic continuation from L}) is the physically correct one \cite{Ducru_Scattering_Matrix_of_Complex_Wavenumbers_2019}, as it conserves the meromorphic properties of the Kapur-Peierls operator, which preserves general unitarity, cancels non-physical poles out of the scattering matrix $\boldsymbol{U}(E)$ otherwise spuriously introduced by the Lane \& Thomas approach (\ref{eq:: Def S = Re[L], P = Im[L]}), allows for parameters transform under change of channel radius, and still should close cross sections below channel thresholds. 
Though there is no absolute consensus yet amongst the community as to which approach ought to be valid, both yield identical results for real energies above threshold (real wavenumbers $k_c \in \mathbb{R}$) in the case of exact R-matrix formalism (but not the Reich-Moore approximation as we show in section \ref{sec:Generalized Brune parameters for Reich-Moore approximation}).

\begin{theorem}\label{theo::analytic_Brune_poles} \textsc{Alternative Analytic Poles}. \\
Let the alternative analytic poles $\left\{\widetilde{E_i}\right\}$ be the solutions of the Brune generalized eigenproblem (\ref{eq:Brune eigenproblem}), using the analytic continuation definition (\ref{eq:: Def S and P analytic continuation from L}) for the shift $S_c(E)$, and let $N_S$ be the number of such solutions, then:
\begin{itemize}
    \item the alternative analytic poles are in general complex, and live on the single sheet of the unfolded Riemann surface from (\ref{eq:rho_c(E) mapping}) mapping: $\left\{\widetilde{E_i} \right\} \in \mathbb{C}^{N_S}$, 
    \item in the neutral particle case, there are exactly $N_S$ complex alternative analytic poles with:
\begin{equation}
N_S = N_\lambda + \sum_{c=1}^{N_c}\ell_c
\label{eq:N_S Brune poles}
\end{equation}
    \item in the charged particles case, there is a countable infinity of complex alternative analytic poles: $N_S = \infty$,
    \item for each level $\lambda$, there exists a real principal (or resonant) alternative analytic pole. These $N_\lambda$ principal poles are the same as the principal alternative branch poles of theorem \ref{theo::branch_brune_poles},
    \item the number $N_S^{\mathbb{R}}$ of real alternative analytic poles, $\left\{\widetilde{E_i} \right\} \in \mathbb{R}^{N_S^{\mathbb{R}}}$, is greater than the number of levels, $ N_S^{\mathbb{R}} \geq N_\lambda$, and depends on the values of the resonance parameters $\big\{E_\lambda, \gamma_{\lambda c}, B_c, E_{T_c}, a_c \big\}$ -- though in a way that is invariant under change of boundary-condition $B_c$,
    \item each neutral particle, odd angular momentum $\ell_c \equiv 1 \; (\mathrm{mod}\; 2)$, channel adds at least one real alternative analytic shadow pole below threshold,
\end{itemize}
so that the number $N_S$ of complex and $N_S^{\mathbb{R}}$ of real alternative analytic poles is greater than the number $N_\lambda$ of levels: $ N_S \geq N_S^{\mathbb{R}} \geq N_\lambda$.
\end{theorem}

\begin{proof}

The proof follows the one of theorem \ref{theo::branch_brune_poles}. 
However, when considering the left-hand side of (\ref{eq:Brune eigenproblem}), the shift function is now defined from analytic continuation definition (\ref{eq:: Def S and P analytic continuation from L}), which in general entails $S_c(E)$ is a complex number. 
This entails the left-hand side of (\ref{eq:Brune eigenproblem}) is now a complex symmetric matrix. 
In general, a complex symmetric matrix is not diagonalizable, has no special properties on its spectrum, and we refer to reference literature on its Jordan canonical form and other properties 
\cite{Craven_complex_symmetric_1969, Nondefective_complex_symmetric_matrices_1985, complex_symmetric_matrix_SVD_1988, Scott_complex_symmetric_1993, fast_diag_of_complex_symmetric_matrices_for_quantum_applications_1997, Complex_symmetric_operators_2005, Complex_symmetric_operators_II_2007}.
Nonetheless, we know the left-hand side of (\ref{eq:Brune eigenproblem}) will be real-symmetric, thus diagonalizable, for real energies above threshold, which hints (but does not prove) it is probably a good assumption to assume the complex symmetric matrix to be non-defective in general. 
Regardless of the eigenvectors, we can search for the alternative poles $\left\{\widetilde{E_i} \right\} $ by solving problem (\ref{eq:R_S by Brune det search}) directly (c.f. discussion around equation (51) in \cite{Brune_2002}).
Here, the analytic properties of definition (\ref{eq:: Def S and P analytic continuation from L}), established in lemma \ref{lem:: analytic S_c and P_c lemma}, entail the determinant in (\ref{eq:R_S by Brune det search}) is a meromorphic operator of $\rho^2$, which unfolds mapping (\ref{eq:rho_c(E) mapping}) so that all the solutions of (\ref{eq:R_S by Brune det search}) live on one single sheet. 

In the case of $N_c$ massive neutral channels, the shift factor $S_c(\rho)$ is a rational fraction in $\rho^2$ with a degree of $\ell_c$ (in $E$ space) in the denominator, where $\ell_c$ is the angular momentum of the channel (c.f. table \ref{tab::S_and_P_expressions_neutral} and lemma \ref{lem:: analytic S_c and P_c lemma} with table \ref{tab::roots of the outgoing wave functions}).
The search for the poles of the $\boldsymbol{R}_S$ operator (\ref{eq:R_S by Brune det search}) will then yield $N_S$ complex alternative poles $\left\{\widetilde{E_i}\right\} \in \mathbb{C}$ with $N_S = N_\lambda + \sum_{c=1}^{N_c}\ell_c$, as stated in (\ref{eq:N_S Brune poles}).
The intuition behind this number $N_S$ is that both the R-matrix (\ref{eq:R expression}) and the diagonal matrix of shift functions, $\boldsymbol{S}(E) \triangleq \mathrm{\boldsymbol{diag}}\left( S_c(E)\right)$, will each contribute their number of poles, $N_\lambda$ and $\sum_c \ell_c$ respectively, adding them up to yield $N_S = N_\lambda + \sum_{c=1}^{N_c}\ell_c$ solutions (\ref{eq:N_S Brune poles}) to the determinant problem (\ref{eq:R_S by Brune det search}).
We achieved a formal proof of result (\ref{eq:N_S Brune poles}), though it is somewhat technical. 
It rests on the diagonal divisibility and capped multiplicities lemma \ref{lem::diagonal divisibility and capped multiplicities}, which we apply to the developed rational fraction $\mathrm{det}\left( \boldsymbol{R}_S^{-1}(E)\right)$ in (\ref{eq:R_S by Brune det search}), or directly to (\ref{eq:Brune eigenproblem}), depending on whether $N_\lambda \geq N_c$ or $N_c \geq N_\lambda$. 
In the (most common) case of $N_\lambda \geq N_c$, we develop $\mathrm{det}\left( \boldsymbol{R}_S^{-1}\right)(E) = \mathrm{det}\left( \boldsymbol{R}^{-1} - \boldsymbol{S^0}\right)(E)$ by n-linearity:
  $ \mathrm{det}\left( \boldsymbol{R}^{-1} - \boldsymbol{S^0}\right)= \mathrm{det}\left( \boldsymbol{R}^{-1}\right) \mathrm{det}\left( \Id{} - \boldsymbol{R}\boldsymbol{S^0}\right)$ with $\mathrm{det}\left( \Id{} - \boldsymbol{R}\boldsymbol{S^0}\right) = 1 - \mathrm{Tr}\left( \boldsymbol{R}\boldsymbol{S^0}\right) + \hdots + \mathrm{Tr}\left( \boldsymbol{\mathrm{Adj}}\left(-\boldsymbol{R}\boldsymbol{S^0}\right)\right) + \mathrm{det}\left(-\boldsymbol{R}\boldsymbol{S^0}\right) $, so that: 
  $\mathrm{det}\left( \boldsymbol{R}_S^{-1}\right) = \mathrm{det}\left( \boldsymbol{R}^{-1}\right)  - \mathrm{Tr}\left( \boldsymbol{\mathrm{Adj}}\left(\boldsymbol{R}^{-1}\right)\boldsymbol{S^0}\right) + \hdots - \mathrm{Tr}\left(\boldsymbol{R}^{-1}\boldsymbol{\mathrm{Adj}}\left(\boldsymbol{S^0}\right)\right) +  (-1)^{N_c} \mathrm{det}\left(\boldsymbol{S^0}\right)  $.
  In the latter expression, $\boldsymbol{R}^{-1}(E) = \boldsymbol{\gamma}^+ \left(\boldsymbol{e} - E \Id{} \right) {\boldsymbol{\gamma}^\mathsf{T}}^+$ has no poles, so its determinant is a polynomial $\mathrm{det}\left( \boldsymbol{R}^{-1}\right)(E) \in \mathbb{C}[X]$. 
 The rational fraction with greatest degree in the denominator is $\mathrm{det}\left(\boldsymbol{S^0}\right) (E) \in \mathbb{C}(X)$.
 For neutral particles $S_c^0(E) = \frac{s^0_c(E)}{d_c(E)} $, where the denominator is of degree $\ell_c = \mathrm{deg}\left(d_c(E) \right)$ in $E$ space (c.f. table \ref{tab::S_and_P_expressions_neutral}), so that to rationalize the rational fraction $\mathrm{det}\left( \boldsymbol{R}_S^{-1}\right)(E) \in \mathbb{C}(X)$, we must multiply it by the denominator of $\mathrm{det}\left(\boldsymbol{S^0}\right) (E)$, which is $\prod_{c=1}^{N_c} d_c(E)$, a polynomial of degree $\sum_c \ell_c$. 
 That is $\left(\prod_{c=1}^{N_c} d_c(E) \right) \times \mathrm{det}\left( \boldsymbol{R}_S^{-1}\right)(E) = \left(\prod_{c=1}^{N_c} d_c(E) \right) \times \mathrm{det}\left( \boldsymbol{R}^{-1}\right)(E)  + \hdots + (-1)^{N_c}\prod_{c=1}^{N_c} s^0_c(E)  \in \mathbb{C}[X]$.
 The dominant degree polynomial in this expression is $\left(\prod_{c=1}^{N_c} d_c(E) \right) \times \mathrm{det}\left( \boldsymbol{R}^{-1}\right)(E) $. In this expression, the total degree of the polynomial is the sum of the degrees of the product terms.
 We readily have $\mathrm{deg}\left(\prod_{c=1}^{N_c} d_c(E) \right) = \sum_c \ell_c$. 
For the degree of the determinant term $\mathrm{det}\left( \boldsymbol{R}^{-1}\right)(E)$, the application of diagonal divisibility and capped multiplicities lemma \ref{lem::diagonal divisibility and capped multiplicities} stipulates that if $E_{\lambda_1} = E_{\lambda_2} = \hdots = E_{\lambda_{m_\lambda}} $, this multiplicity ${m_\lambda}$ of the resonance energy value $E_\lambda $ will be capped by $N_c$. In practice, this does not happen because the Wigner-Eisenbud resonance parameters $E_\lambda$ are defined as different from each other $E_\lambda \neq E_{\mu\neq \lambda}$.
 This is no longer true in the case $N_c \geq N_\lambda$, where developing the determinant of (\ref{eq:Brune eigenproblem}) directly will similarly yield by n-linearity, and denoting $\boldsymbol{\Delta} \triangleq \boldsymbol{e} - E \Id{}$ for clarity of scripture: 
 $\mathrm{det}\left( \boldsymbol{\Delta} - \boldsymbol{\gamma} \boldsymbol{S^0} \boldsymbol{ \gamma}^\mathsf{T}\right) = \mathrm{det}\left( \boldsymbol{\Delta}\right)  - \mathrm{Tr}\left( \boldsymbol{\mathrm{Adj}}\left(\boldsymbol{\Delta}\right)\boldsymbol{\gamma} \boldsymbol{S^0} \boldsymbol{ \gamma}^\mathsf{T}\right) + \hdots - \mathrm{Tr}\left(\boldsymbol{\Delta} \; \boldsymbol{\mathrm{Adj}}\left(\boldsymbol{\gamma} \boldsymbol{S^0} \boldsymbol{ \gamma}^\mathsf{T}\right)\right) + (-1)^{N_\lambda} \mathrm{det}\left(\boldsymbol{\gamma} \boldsymbol{S^0} \boldsymbol{ \gamma}^\mathsf{T}\right)  $.
Again, in the latter expression the rational fraction with the highest-degree denominator is $\mathrm{det}\left(\boldsymbol{\gamma} \boldsymbol{S^0} \boldsymbol{ \gamma}^\mathsf{T}\right)(E) \in \mathbb{C}(X)$. 
Applying the diagonal divisibility and capped multiplicities lemma \ref{lem::diagonal divisibility and capped multiplicities} to it commands that if there are various channels with the same $S_c(E)$, for instance with the same $\ell_c$ and ${\rho_0}_c$, their multiplicity of occurrence is capped by $N_\lambda$ when rationalizing the fraction  $\mathrm{det}\left(\boldsymbol{\gamma} \boldsymbol{S^0} \boldsymbol{ \gamma}^\mathsf{T}\right)(E) \in \mathbb{C}(X)$, so that $Q(E) \times \mathrm{det}\left(\boldsymbol{\gamma} \boldsymbol{S^0} \boldsymbol{ \gamma}^\mathsf{T}\right)(E) \in \mathbb{C}[X]$ is a polynomial, with $Q(E) \triangleq \left(\underset{d_c \neq d_{c \neq c'}}{\prod_{c = 1} 
^{N_c}}d_c(E) \right) \times  \left( \underset{ d_c = d_{c \neq c'}}{\prod_{c = 1}^{\mathrm{min}\left\{N_c,N_\lambda\right\}} }d_c(E) \right)$.
In the developed expression of the polynomial $ Q(E) \times \mathrm{det}\left( \boldsymbol{\Delta} - \boldsymbol{\gamma} \boldsymbol{S^0} \boldsymbol{ \gamma}^\mathsf{T}\right)$, the dominant degree term is now: $Q(E) \times \mathrm{det}\left( \boldsymbol{\Delta}\right) $, the degree of which is the sum of the degree of each term. The degree of $ \mathrm{det}\left( \boldsymbol{\Delta}\right)$ is $N_\lambda$, whereas the degree of $Q(E)$ is $\mathrm{deg}\left(Q(E)\right) = \sum_{c=1 | \ell_c \neq \ell_{c'}}^{N_c} \ell_c + \sum_{c=1 | \ell_c = \ell_{c'}}^{\mathrm{min}\left\{N_\lambda, N_c\right\}} \ell_c $.
Hence, we find back the expression (\ref{eq:N_S Brune poles}) to be proved: $N_S = N_\lambda + \sum_{c=1}^{N_c}\ell_c$, but with the additional subtlety that the multiplicities (repeating occurrences) are capped, both for $\sum_{\begin{array}{c}
     \tiny{E_\lambda \mathrm{\, multiplicity }}  \\
    \tiny{\mathrm{ capped \, at} \, N_c}
 \end{array}} \mathrm{deg}\left(E_\lambda - \rho^2(E)\right)$ and for $\sum_{\begin{array}{c}
   \tiny{S_c \mathrm{\, multiplicity }}  \\
    \tiny{\mathrm{ capped \, at} \, N_\lambda}
 \end{array}} \mathrm{deg}\left(d_c(\rho(E))\right)$, so that the final, exact number of complex eigenvalues to Brune's generalized eigenproblem (\ref{eq:Brune eigenproblem}) in the neutral channels case is: 
\begin{equation}
\begin{IEEEeqnarraybox}[][c]{rcl}
 N_S =  N_\lambda + \sum_{\begin{array}{c}
   \tiny{S_c \mathrm{\, multiplicity }}  \\
    \tiny{\mathrm{ capped \, at} \, N_\lambda}
 \end{array}} \ell_c 
\IEEEstrut\end{IEEEeqnarraybox}
\label{eq: capped multiplicities N_S}
\end{equation}
This means that if many channels, say $m_c$, have the same shift function $S_c = S_{c'}$, the resulting $\ell_c = \ell_{c'} $ will only be added $\mathrm{min}\left\{m_c, N_\lambda\right\}$ times in the sum (\ref{eq: capped multiplicities N_S}). \\
A final technical note to state that this number $N_S$ of poles (\ref{eq: capped multiplicities N_S}) is true in $E$ space, as we have showed in lemma \ref{lem:: analytic S_c and P_c lemma} that definition (\ref{eq:: Def S and P analytic continuation from L}) unfolds the Riemann sheet of (\ref{eq:rho_c(E) mapping}). 
If we were performing this in $\rho$ space, we would thus simply multiply the degrees by 2.
This is not true if we were searching for the poles of the Kapur-Peierls operator $\boldsymbol{R}_L$, as the mapping of $\rho(E)$ is not one-to-one anymore. 
From table \ref{tab::L_values_neutral}, we would be able to perform the same analysis that yielded (\ref{eq: capped multiplicities N_S}), but it would have to be in $\rho$ space. 

In the charged particles case, $S_c(E)$ has an infinity of poles (c.f. lemma \ref{lem:: analytic S_c and P_c lemma}). Extending our proof of (\ref{eq: capped multiplicities N_S}) from the neutral particles to the charged particles ones would thus yield a countable infinity of complex alternative analytic poles. 

A key question is: how many of the $N_S$ complex alternative poles are real?
To address it, we come back to the three-step Brune argument and look for real eigenvalues from the left-hand-side of (\ref{eq:Brune eigenproblem}) that will cross the right-hand side identity line $E=E$ for real values. 
Here again, Brune's three-step argument will guarantee at least $N_\lambda$ real solutions.
There are in general more solutions however, and as for the alternative shadow poles of theorem \ref{theo::branch_brune_poles}, the number of real alternative analytic poles, solutions to (\ref{eq:Brune eigenproblem}), will depend on the R-matrix parameters $\big\{E_\lambda, \gamma_{\lambda c}, B_c, E_{T_c}, a_c \big\}$, in a way that is invariant under change of boundary-condition $B_c$ (plug-in invariance (\ref{eq:: R_B invariance for B'}) into (\ref{eq:R_S by Brune det search})).
We illustrate various such cases in figure \ref{fig:Brune_argument_even_angular_momenta}. 
However, each neutral particle channel with odd angular momentum $\ell_c \equiv 1 \; (\mathrm{mod} \; 2 )$ will add at least one real sub-threshold solution to the $N_\lambda$ ones, due to the real sub-threshold pole of $S_c(E)$ discovered in lemma \ref{lem:: analytic S_c and P_c lemma}.
This behavior is depicted in figure \ref{fig:Brune_argument_odd_angular_momenta}.
\end{proof}

\begin{lem}\label{lem::diagonal divisibility and capped multiplicities}
\textsc{Diagonal divisibility and capped multiplicities}.\\
Let $\boldsymbol{M} \in \mathbb{C}^{m\times n}$ be a complex matrix and $\boldsymbol{D}(z) \in \boldsymbol{\mathrm{Diag}}_{n}\left(\mathbb{C}\left(X\right)\right)$ be a diagonal matrix of complex rational functions with simple poles, that is $D_{ij}(z) = \delta_{ij} \frac{R_i(z) \in \mathbb{C}\left[ X \right] }{P_i(z) \in \mathbb{C}\left[ X \right]}$, with $\mathbb{C}\left[ X \right]$ designating the set of polynomials and $\mathbb{C}\left(X\right)$ the set of rational expressions, and we assume $P_i(z)$ has simple roots.\\
Let $Q(z) \in \mathbb{C}\left[ X \right]$ be the denominator of $\mathrm{det}\left( \boldsymbol{D} \right)(z)$, but with all multiplicities capped by $m$, i.e.
\begin{equation}
Q(z) \triangleq \prod_{\begin{array}{c}
     j = 1  \\
     P_j \neq P_{i \neq j}
\end{array}}^n P_j(z)\prod_{\begin{array}{c}
     i = 1  \\
     P_i = P_{i \neq j}
\end{array}}^{\mathrm{mim}\left\{n,m\right\}} P_i(z)
\end{equation}
then $Q(z)$ is the denominator of $\mathrm{det}\left( \boldsymbol{M} \boldsymbol{D}(z) \boldsymbol{M}^\mathsf{T} \right) $, so that:
\begin{equation}
Q(z)\cdot \mathrm{det}\left( \boldsymbol{M} \boldsymbol{D}(z) \boldsymbol{M}^\mathsf{T} \right)  \in \mathbb{C}\left[ X \right]
\end{equation}
\end{lem}

\begin{proof}
Leibniz's determinant formula yields:
\begin{equation*}
\begin{IEEEeqnarraybox}[][c]{rcl}
\mathrm{det}\left( \boldsymbol{M} \boldsymbol{D}(z) \boldsymbol{M}^\mathsf{T} \right)  & = & \sum_{\sigma \in S_m} \epsilon(\sigma) \prod_{i=1}^{m}\sum_{j=1}^{n}M_{ij}M_{\sigma{i} j} \frac{R_j(z)}{P_j(z)}  \\
\IEEEstrut\end{IEEEeqnarraybox}
\label{eq:: Leibniz formula}
\end{equation*}
Let us now develop the product using the formula:
\begin{equation*}
\begin{IEEEeqnarraybox}[][c]{rcl}
\prod_{i=1}^{m}\sum_{j=1}^{n}x_{i,j} = \sum_{j_1, \hdots , j_m \in \llbracket 1, n \rrbracket ^m } \prod_{i=1}^{m}x_{i,j_i}
\IEEEstrut \end{IEEEeqnarraybox}
\label{eq:: Yoann's touch formula}
\end{equation*}
which leads to:
\begin{equation}
\begin{IEEEeqnarraybox}[][c]{rcl}
\mathrm{det}\left( \boldsymbol{M} \boldsymbol{D} \boldsymbol{M}^\mathsf{T} \right)  & = & \sum_{\sigma \in S_m} \epsilon(\sigma)\sum_{\tiny{\begin{array}{c}
     j_1, \hdots , j_m  \\
      \in \llbracket 1, n \rrbracket ^m 
\end{array}}}  \prod_{i=1}^{m} M_{ij_i}M_{\sigma{i} j_i} \frac{R_{j_i}(z)}{P_{j_i}(z)}  \\
\IEEEstrut\end{IEEEeqnarraybox}
\label{eq:: Developed determinant for divisibility Lemma}
\end{equation}
We here have a sum of products of $m$ terms; thus, the $ \frac{R_{j}(z)}{P_{j}(z)}$ never appear more than $m$ times in each product -- nor more than their multiplicity in $\mathrm{det}\left( \boldsymbol{D} \right)(z)$.
It thus suffices to account for each $P_j(z)$ a number of times that is the maximum between its multiplicity and $m$ in order to rationalize the $\mathrm{det}\left( \boldsymbol{M} \boldsymbol{D}(z) \boldsymbol{M}^\mathsf{T} \right) \in \mathbb{C}(X)$ fraction.
\end{proof}

\begin{figure}[ht!!] 
  \centering
  \subfigure[\ Positive resonance energy $E_\lambda = 1$.]{\includegraphics[width=0.49\textwidth]{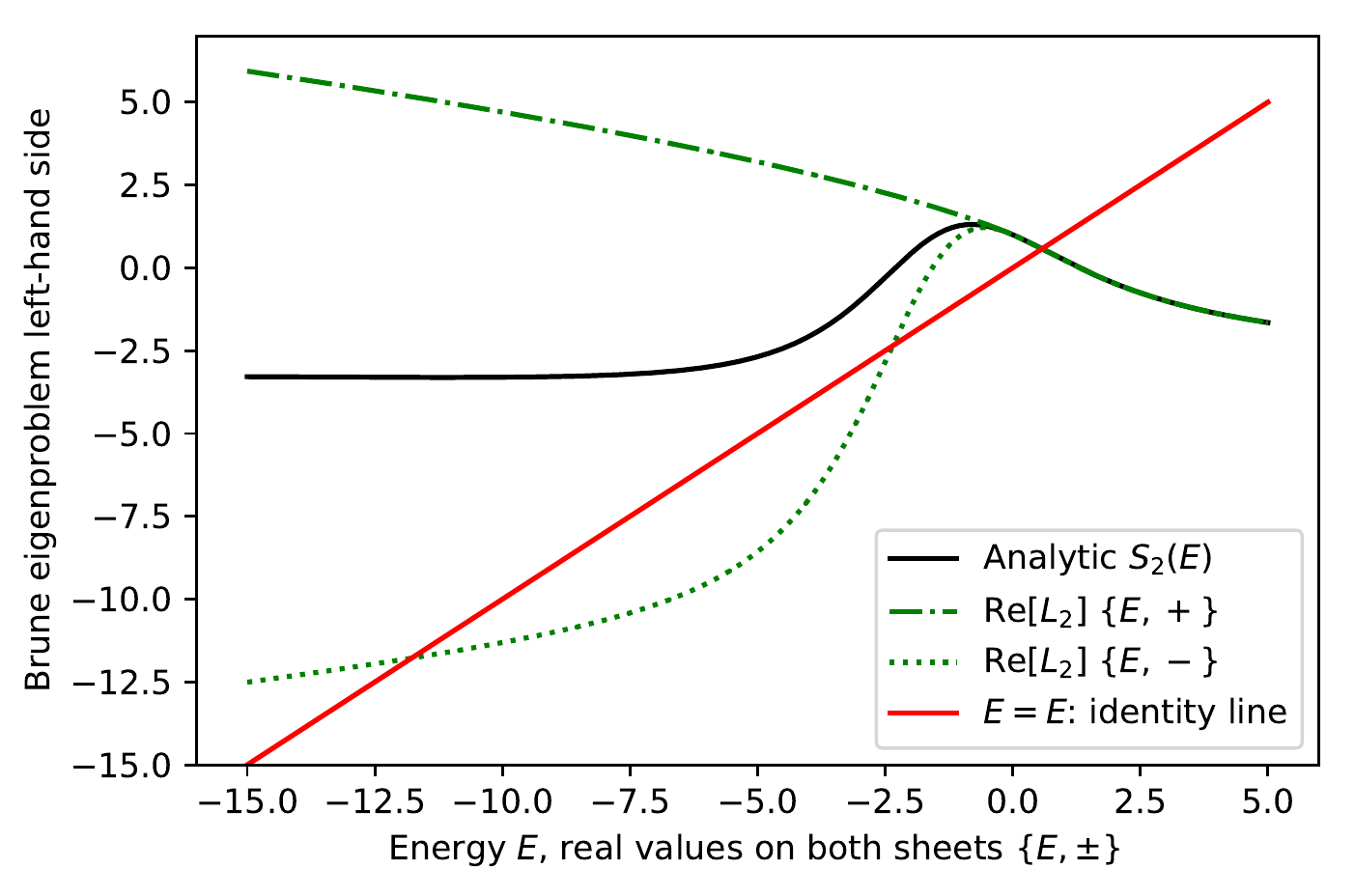}}
  \subfigure[\ Negative resonance energy $E_\lambda = -2$.]{\includegraphics[width=0.49\textwidth]{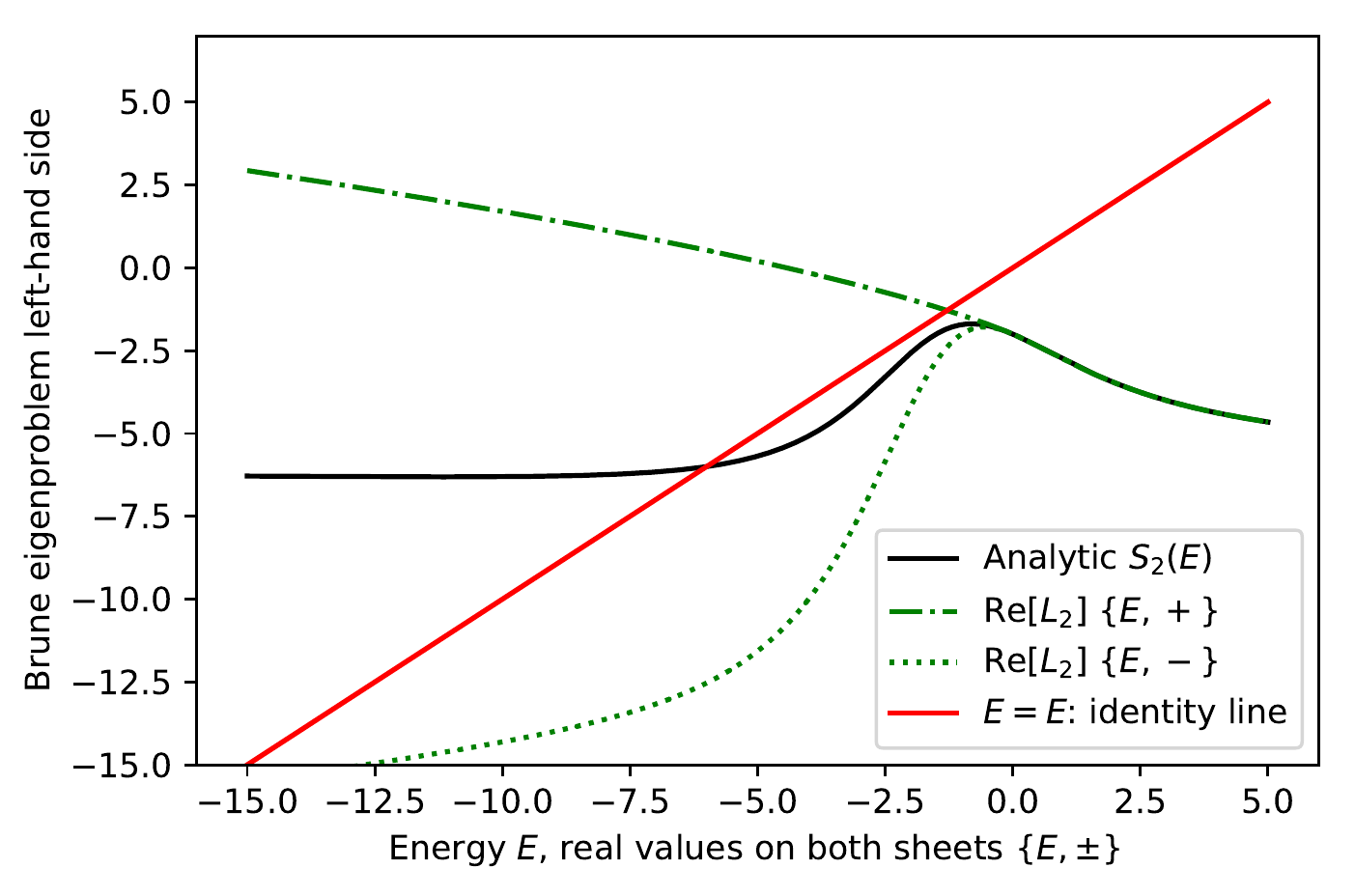}}
  \caption{\small{Brune eigenproblem (\ref{eq:: elemental Brune problem}) for 1-level 1-channel d-wave: \\ comparison of real solutions from definitions (\ref{eq:: Def S = Re[L], P = Im[L]}) versus (\ref{eq:: Def S and P analytic continuation from L}), for angular momentum $\ell_c = 2$, neutral particles energy-wavenumber mapping (\ref{eq:rho_c massive}), resonance width $\gamma_{\lambda,c} = \sqrt{2}$, using $B_c = - \ell_c$ convention and zero threshold $E_{T_c}= 0$. Units such that $\rho_0 = 1$. 
  Since there are no real sub-threshold poles, both can yield one, two, or three solutions (crossing the E=E diagonal), depending on the values of the resonance parameters. If at threshold energy $E_{T_c}$ the left hand side of (\ref{eq:: elemental Brune problem}) is above the E=E diagonal, then the above-threshold solutions from both definitions coincide. In any case, the sub-threshold solutions differ. Behavior is analogous for all even angular momenta $\ell_c \equiv 0 (\mathrm{mod} 2)$.}}
  \label{fig:Brune_argument_even_angular_momenta}
\end{figure}

Importantly, since both shift function $S_c(E)$ definitions (\ref{eq:: Def S = Re[L], P = Im[L]}) and (\ref{eq:: Def S analytic}) coincide above threshold, the solutions to (\ref{eq:Brune eigenproblem}) will be the same above thresholds. 
The discrepancy in the values of the alternative parameters, solutions to (\ref{eq:Brune eigenproblem}), will only differ when certain channels have to be considered below threshold: $S_c(E)$ with $E<E_{T_c}$.

To illustrate these differences, let us consider the simple example of a one-level, one-channel neutral particle interaction, with a zero-threshold $E_{T_c} = 0$, and set about solving the Brune generalized eigenproblem (\ref{eq:Brune eigenproblem}), which here takes the simple scalar form:
\begin{equation}
    E_\lambda - \gamma_{\lambda,c}\big( S_c(E) - B_c \big) \gamma_{\lambda,c} = E
    \label{eq:: elemental Brune problem}
\end{equation}
In figures \ref{fig:Brune_argument_odd_angular_momenta} and \ref{fig:Brune_argument_even_angular_momenta}, we plotted the left and right hand side of this elemental Brune eigenproblem (\ref{eq:: elemental Brune problem}), for both definitions (\ref{eq:: Def S = Re[L], P = Im[L]}) and (\ref{eq:: Def S and P analytic continuation from L}) of the shift function $S_c(E)$, for various values of resonance parameters $\left\{E_\lambda, \gamma_{\lambda,c}\right\}$ and the convention $B_c = -\ell_c$, for different angular momenta $\ell_c$.

In the case of $\ell_c = 1$, depicted in figure \ref{fig:Brune_argument_odd_angular_momenta}, one can observe that the real sub-threshold pole engendered by odd angular momenta (c.f. section \ref{subsubsec::Ambiguity in shift and penetration}) introduces a sub-threshold alternative parameter, where the left-hand side of (\ref{eq:: elemental Brune problem}) crosses the $E=E$ identity line.
In the case of the Lane \& Thomas legacy definition (\ref{eq:: Def S = Re[L], P = Im[L]}), this sub-threshold alternative shadow pole is on the $\big\{ E, - \big\}$ sheet of mapping (\ref{eq:rho_c massive}),  whereas for analytic continuation definition (\ref{eq:: Def S and P analytic continuation from L}) it is on the same, unique sheet. The same behavior will be observable for all odd angular momenta $\ell_c \equiv 1 \; (\mathrm{mod} \; 2)$.

In the case $\ell_c=2$, depicted in figure \ref{fig:Brune_argument_even_angular_momenta}, the non-purely-imaginary poles $\left\{ \omega_n , \omega_n^* \right\} \not\in \mathrm{i}\mathbb{R}$ (c.f. lemma \ref{lem:: analytic S_c and P_c lemma} and table \ref{tab::roots of the outgoing wave functions}) will impact the shift function $S_c(\rho_c)$ in ways that may or may not produce additional real solutions $\left\{\widetilde{E_i}\right\} \in \mathbb{R}$ to the generalized eigenproblem (\ref{eq:Brune eigenproblem}).
This behavior is reported in figure \ref{fig:Brune_argument_even_angular_momenta}, where one can observe that, depending on the R-matrix parameter values $\big\{ E_\lambda, \gamma_{\lambda,c}, B_c \big\}$, there are either one, two (tangential for the analytic continuation definition), or three solutions to the Brune generalized eigenproblem (\ref{eq:: elemental Brune problem}).
For instance, one can see that definition (\ref{eq:: Def S = Re[L], P = Im[L]}) can yield situations with two sub-threshold alternative branch poles -- one on the $\big\{E,+\big\}$ branch and one shadow pole (i.e. on the $\big\{E,-\big\}$ branch) -- or with two sub-threshold alternative shadow poles -- both sub-threshold on the $\big\{E,-\big\}$ branch -- or situations where only one, above-threshold solution is produced. 
On the other hand, analytic continuation definition (\ref{eq:: Def S and P analytic continuation from L}) can also yield one, two (tangentially) or three solutions, depending on the sub-threshold behavior and the resonant parameters eigenvalues $\big\{ E_\lambda, \gamma_{\lambda,c}, B_c \big\}$.
The number of real solutions $\left\{\widetilde{E_i}\right\} \in \mathbb{R}$ to the Brune generalized eigenproblem (\ref{eq:Brune eigenproblem}) will thus depend on the R-matrix parameters, and is in general comprised between $N_\lambda$ and $N_S$.

To verify the number of complex alternative analytic poles (\ref{eq:N_S Brune poles}), a trivial example is considering (\ref{eq:: elemental Brune problem}) in the $\ell_c = 1$ case, where the analytic shift function takes the wavenumber dependence, $S(\rho) = - \frac{1}{1+\rho^2}$, and thus the poles of the $\boldsymbol{R}_S$ operator are nothing but the solutions to $\frac{E_\lambda - E}{\gamma_{\lambda,c}^2} + B + \frac{1}{1+\rho_0^2(E-E_{T_c})} = 0$.
The fundamental theorem of algebra then guarantees this problem has $N_S = 2$ complex solutions, not $N_\lambda = 1$.
The surprising part is that both are real poles: one above and one below threshold, which again stems from the fact the number of roots $\left\{ \omega_n \right\}$ is odd and that their symmetries thus require one pole to be exactly imaginary (in wavenumber space), as explained in section \ref{subsubsec::Ambiguity in shift and penetration}. 
For $\ell_c = 2$, we would have $S_2^0(E) = \frac{3E + 2 E^2}{\frac{9}{\rho_0^2} + 3E + E^2}$, so that the fundamental theorem of algebra commands (\ref{eq:: elemental Brune problem}) will have $N_S = 3$ solutions, verifying the $N_S = N_\lambda + \sum_{c=1}^{N_c}\ell_c$ complex poles we establish in (\ref{eq:N_S Brune poles}).
In the general charged-particles case, the shift factor $S_c(\rho)$ is no longer a rational fraction in $\rho^2$ but is a meromorphic operator in $\rho^2$ with an infinity of poles (c.f. lemma \ref{lem:: analytic S_c and P_c lemma}). This illustrates how, in general, there exist $N_\lambda \leq N_S \leq \infty$ complex poles of the $\boldsymbol{R}_S$ operator, and that at least $N_\lambda$ of them are real.

When the left-hand side of (\ref{eq:: elemental Brune problem}) crosses the $E=E$ identity line above threshold, the alternative branch poles coincide with the alternative analytic poles, as can be observed in figures \ref{fig:Brune_argument_odd_angular_momenta} and \ref{fig:Brune_argument_even_angular_momenta}. 
Since the shift function $S_c(E)$ is continuous and monotonically increasing above threshold, the question is whether the eigenvalues of the left-hand side of (\ref{eq:Brune eigenproblem}) are above the $E=E$ identity line at the threshold value: $E = E_{T_c}$.
If yes, then it would mean that past the last threshold there will be exactly $N_\lambda$ solutions to (\ref{eq:Brune eigenproblem}).
However, nothing guarantees \textit{a priori} that all the eigenvalues of the left hand side of (\ref{eq:Brune eigenproblem}) are above the $E=E$ at the last threshold.
From solving the elemental Brune problem (\ref{eq:: elemental Brune problem}), we observed that it seems to require negative resonance levels $E_\lambda < 0$ to induce the left-hand side of (\ref{eq:Brune eigenproblem}) to be below the $E=E$ identity line at the threshold value, as illustrated in figures \ref{fig:Brune_argument_odd_angular_momenta} and \ref{fig:Brune_argument_even_angular_momenta}.
When this happens, the alternative poles will be sub-threshold, and thus depend on the (\ref{eq:: Def S = Re[L], P = Im[L]}) or (\ref{eq:: Def S and P analytic continuation from L}) definition for the shift function $S_c(E)$.
However, the fact that different channels will have different threshold levels $E_{T_c} \neq E_{T_{c'}}$, and that nothing stops R-matrix parameters from displaying negative resonance levels $E_\lambda < 0$, mean no definitive conclusion can be reached as to the number of real alternative parameters (other than there exists at least $N_\lambda$ of them).

\subsection{\label{subsubsec::choice of Brune poles}Choice of alternative poles}

Brune defined the alternative parameters in (\ref{eq:Brune parameters}) and (\ref{eq:: Brune invA}) by building the square matrix $\boldsymbol{a}$, and then inverting it to guarantee (\ref{eq:: Brune A = aAa}) (c.f. section \ref{sec:R_S def}).
We just demonstrated in theorems \ref{theo::branch_brune_poles} and \ref{theo::analytic_Brune_poles} that there are in general more alternative poles $N_S$ -- either alternative branch poles or alternative analytic poles -- than the number $N_\lambda$ of resonance levels: $N_S \geq N_\lambda$.
Yet the fact that there are more than $N_\lambda$ solutions to (\ref{eq:Brune eigenproblem}) implies the $\boldsymbol{a} \triangleq \left[\boldsymbol{a_1}, \hdots, \boldsymbol{a_i} , \hdots , \boldsymbol{a_{N_S}} \right]$ matrix, composed of the $N_S$ solutions to Brune's eigenproblem (\ref{eq:Brune eigenproblem}), is in general not square, and could even be infinite if $N_S = \infty$ (Coulomb channels).
This brings two critical questions: 1) do these additional alternative poles impede us from well defining the alternative parameters? 2) can we still uniquely define the alternative poles?

We here demonstrate in theorem \ref{theo::Choice of Brune poles} the striking property that choosing any finite set of at least $N_\lambda$ different solutions from the $N_\lambda \leq N_S \leq \infty$ solutions of Brune's eigenproblem (\ref{eq:Brune eigenproblem}), suffices, under our new extended definition (\ref{eq:: pseudo invA Brune}), to properly describe the R-matrix scattering model.

\begin{theorem}\label{theo::Choice of Brune poles} \textsc{Choice of alternative poles} \\
If we generalize definition (\ref{eq:: Brune invA}) of the alternative level matrix $\boldsymbol{\widetilde{A}} $ by defining it as the following (Moore-Penrose) pseudo-inverse:
\begin{equation}
\boldsymbol{\widetilde{A}} \ \triangleq \ \left[\boldsymbol{a}^\mathsf{T}\boldsymbol{A}^{-1}\boldsymbol{a}\right]^{+}
\label{eq:: pseudo invA Brune}
\end{equation}
then the choice of any number $N_S$ of alternative poles, solutions to the Brune generalized eigenproblem (\ref{eq:Brune eigenproblem}), will reconstruct the scattering matrix $\boldsymbol{U}(E)$, if, and only if, we choose at least $N_\lambda$ solutions: $N_S \geq N_\lambda$.
\end{theorem}

\begin{proof}
The proof rests on the pseudo-inverse property for independent columns and rows, and applies it to the $\boldsymbol{a} \triangleq \left[\boldsymbol{a_1}, \hdots, \boldsymbol{a_i} , \hdots , \boldsymbol{a_{N_S}} \right]$ matrix, constructed by choosing $N_S$ solutions of the generalized eigenproblem (\ref{eq:Brune eigenproblem}).
If $N_S \geq N_\lambda$, then $\boldsymbol{a}$ has independent rows so that its pseudo-inverse will yield: $ \boldsymbol{\widetilde{A}}  = \boldsymbol{a}^{+}\boldsymbol{A}{\boldsymbol{a}^\mathsf{T}}^{+}$.
This property in turn entails (\ref{eq:: Brune A = aAa}) is satisfied, and thus (\ref{eq:R_L unchanged by Brune}) stands, leaving unchanged the Kapur-Peierls operator $\boldsymbol{R}_L$, and hence fully representing the scattering matrix $\boldsymbol{U}(E)$.
\end{proof}

Critically, $N_\lambda$ real solutions to (\ref{eq:Brune eigenproblem}) can always be found -- as shown in theorems \ref{theo::branch_brune_poles} and \ref{theo::analytic_Brune_poles} -- meaning the alternative parametrization is always capable of fully reconstructing the scattering matrix energy behavior with real parameters through generalized pseudo-inverse definition (\ref{eq:: pseudo invA Brune}). It is well defined.

Yet, if any choice of $N_\lambda$ alternative poles will yield the same scattering matrix $\boldsymbol{U}(E)$ through definition (\ref{eq:: pseudo invA Brune}), this choice is \textit{a priori} not unique. 
Can we define some conventions on the choice of alternative parameters to make them unique?
Under the legacy Lane \& Thomas definition (\ref{eq:: Def S = Re[L], P = Im[L]}), this can readily be achieved by neglecting the shadow poles and restraining the search to the principal sheet $\big\{ E, + , \hdots , +\big\}$, for all $N_c$ channels, where we have shown in theorem \ref{theo::branch_brune_poles} that one will find exactly $N_\lambda$ poles. 
Under the analytic continuation definition (\ref{eq:: Def S and P analytic continuation from L}), one can still uniquely define the $N_\lambda$ ``first" solutions in the following algorithmic way: one starts the search by diagonalizing, at the last threshold energy (greatest $E_{T_c}$ value), the left-hand side of (\ref{eq:Brune eigenproblem}). If all the eigenvalues are above the $E=E$ line, then increase the energy until the eigenvalues cross the $E=E$ diagonal, and we will have $N_\lambda $ uniquely defined real alternative analytic poles. If at the first threshold some eigenvalues are below the $E=E$ identity line (as we saw could happen if some resonance energies are negative $E_\lambda < 0$), then we can decrease the energy values until those cross the $E=E$ identity line for the first time, and stop the search there, thus again uniquely defining $N_\lambda$ alternative analytic poles. 
This foray into the algorithmic procedure for solving (\ref{eq:: Def S and P analytic continuation from L}) gives us the occasion to point to the vast literature on methods to solve non-linear eigenvalue problems, in particular \cite{Handbook_of_linear_algebra}. 

In the end, though we argue that the physically correct definition for the shift function $S_c(E)$ ought to be through analytic continuation (\ref{eq:: Def S and P analytic continuation from L}), both approaches enable to set conventions that will uniquely determine $N_\lambda$ real alternative poles.

\section{\label{sec:Generalized Brune parameters for Reich-Moore approximation} Generalized alternative parameters for Reich-Moore formalism}

In this section, we study how the community could convert present nuclear data libraries -- featuring both Wigner-Eisenbud parameters and Reich-Moore parameters -- to alternative parameters, in order to eliminate the dependence on the arbitrary boundary condition parameters $B_c$.
We generalize the alternative parametrization to encompass the widely used Reich-Moore formalism, with which many evaluations are conducted, and we show that it is necessary for the community to decide on a convention to continue R-matrix operators to complex wavenumbers -- that is we must choose between branch-points definition (\ref{eq:: Def S = Re[L], P = Im[L]}) and analytic continuation (\ref{eq:: Def S analytic}).

\subsection{\label{subsubsec::Generalization to Reich-Moore approximation and Teichmann-Wigner eliminated channels}Generalization to Reich-Moore formalism and Teichmann-Wigner eliminated channels}

In practice we are only interested in certain outcomes of a nuclear reaction (such as neutron fission, scattering, etc.) and we are sometimes unable to track the vast number of all possible channels (such as every single individual photon interaction) -- this is specially true of heavy nuclides for which the interaction becomes a large many-body problem. 
For these cases, the community has traditionally resorted to Teichmann and Wigner's channel elimination method (c.f. \cite{Teichmann_and_Wigner_1952} or section X, p.299 of \cite{Lane_and_Thomas_1958}) to not explicitly treat all the channels we are not interested in, but still capture their effects on channels of interest.
This yields the Reich-Moore approximation of R-matrix theory \cite{Reich_Moore_1958}, which models the effects of all the eliminated channels (usually $\gamma$ ``gamma capture'' photon channels) on every level by adding to every level's resonance energy $E_\lambda$ a partial eliminated capture width $\Gamma_{\lambda,\gamma}$ that shifts the effective resonance energy into the complex plane:
\begin{equation}
\begin{IEEEeqnarraybox}[][c]{rcl}
\boldsymbol{e}_{\mathrm{R.M.}} & \ \triangleq \ & \boldsymbol{\mathrm{diag}}\left( E_\lambda - \mathrm{i}\frac{\Gamma_{\lambda,\gamma}}{2}\right)
\IEEEstrut\end{IEEEeqnarraybox}
\label{eq:e diagonal matrix Reich-Moore}
\end{equation}
From this, the Reich-Moore formalism R-matrix (\ref{eq:R expression}), where all the eliminated capture channels have been collapsed into one $\gamma$ channel, is now:
\begin{equation}
\begin{IEEEeqnarraybox}[][c]{rcl}
R_{c,c'\not\in\gamma_{\mathrm{elim.}} } & \; \triangleq \; &  \sum_{\lambda=1}^{N_\lambda}\frac{\gamma_{\lambda,c}\gamma_{\lambda,c'}}{E_\lambda - \mathrm{i}\frac{\Gamma_{\lambda,\gamma}}{2} - E}  \\
 \mathrm{i.e.}  \quad \boldsymbol{R}_{\mathrm{R.M.}} & \; = \; & \boldsymbol{\gamma}^\mathsf{T} \left(\boldsymbol{e}_{\mathrm{R.M.}} - E\Id{}\right)^{-1}\boldsymbol{\gamma}
\IEEEstrut\end{IEEEeqnarraybox}
\label{eq:R expression Reich Moore}
\end{equation}
and, equivalently, the Reich-Moore inverse level matrix (\ref{eq:inv_A expression}) thereby becomes:
\begin{equation}
\begin{IEEEeqnarraybox}[][c]{rcl}
\boldsymbol{A^{-1}}_{\mathrm{R.M.}} & \ \triangleq \ & \boldsymbol{e}_{\mathrm{R.M.}} - E\, \Id{} - \boldsymbol{\gamma}\left( \boldsymbol{L} - \boldsymbol{B} \right)\boldsymbol{\gamma}^\mathsf{T}
\IEEEstrut\end{IEEEeqnarraybox}
\label{eq:inv_A expression Reich-Moore}
\end{equation}
All the other R-matrix expressions linking these operators to the scattering matrix (\ref{eq:U expression}), and thereby the cross section (\ref{eq:partial sigma_cc'from scattering matrix}) remain unchanged in the Reich-Moore formalism.
In practice, the only effect of this channel elimination is that the Reich-Moore formalism allows for complex resonance energies (\ref{eq:e diagonal matrix Reich-Moore}) in the parametrizations of R-matrix theory. In this sense, though initially emerging from the channel elimination approximation, the Reich-Moore formalism can be seen like an actual extension of the exact R-matrix formalism.

This has consequential effects on the alternative parametrization. If one wants to convert the Reich-Moore parameters into alternative parameters, Brune's equations of section \ref{sec:R_S Brune transform} are not valid, since they assume the left hand side of (\ref{eq:Brune eigenproblem}) is a real symmetric matrix (and thus diagonalizable with real eigenvalues). 
We here generalize the alternative parametrization of R-matrix theory to encompass the Reich-Moore formalism -- which is of great practical importance -- and the additional shadow poles previously discovered (in theorems \ref{theo::branch_brune_poles} and \ref{theo::analytic_Brune_poles}).
First, we notice that in the Reich-Moore formalism, Brune's generalized eigenproblem (\ref{eq:Brune eigenproblem}) becomes:
\begin{equation}
\begin{IEEEeqnarraybox}[][c]{rCl}
\left[\boldsymbol{e}_{\mathrm{R.M.}}  - \boldsymbol{\gamma} \left( \boldsymbol{S}(\widetilde{E_i}) - \boldsymbol{B} \right)\boldsymbol{ \gamma}^\mathsf{T} \right]\boldsymbol{a_i} = \widetilde{E_i}\boldsymbol{a_i}
\IEEEstrut\end{IEEEeqnarraybox}
\label{eq:Brune eigenproblem Reich-Moore}
\end{equation}
The fact that the left hand side of generalized eigenproblem (\ref{eq:Brune eigenproblem Reich-Moore}) is now a complex symmetric matrix (and not a real symmetric nor a Hermitian matrix) entails the solutions $\widetilde{E_i}$ are no longer real, but complex (we now have complex alternative poles $\widetilde{E_i} \in \mathbb{C}$ and eigenvectors $\boldsymbol{a_i} \in \mathbb{C}^{N_\lambda}$).
In order to conserve an euclidean norm on the space of eigenvectors, the normalization condition must now be generalized to vectors by means of the Hermitian conjugate: 
\begin{equation}
\begin{IEEEeqnarraybox}[][c]{rCl}
\boldsymbol{a_i}^\dagger \boldsymbol{a_i} =1
\IEEEstrut\end{IEEEeqnarraybox}
\label{eq:Brune eigenvectors normalized Reich-Moore}
\end{equation}
We then define the alternative parameters with Hermitian conjugate transformation:
\begin{equation}
\boldsymbol{\widetilde{e}_{\mathrm{R.M.}}} \triangleq  \boldsymbol{\mathrm{diag}}(\widetilde{E_i}) \quad \quad , \quad \quad \boldsymbol{\widetilde{\gamma}} \triangleq \boldsymbol{a}^\dagger \boldsymbol{\gamma}
\label{eq:Brune parameters Reich Moore}
\end{equation}
where $\boldsymbol{a}$ is the matrix composed of the column eigenvectors: $\boldsymbol{a} \triangleq \left[\boldsymbol{a_1}, \hdots, \boldsymbol{a_i} , \hdots \right]$.
We then define the generalized alternative level matrix by means of theorem \ref{theo::Choice of Brune poles}, generalized to complex eigenvectors and for an arbitrary number $N_S \geq N_\lambda$ (at least as many generalized alternative poles as the number of levels) of solutions (now complex) to the generalized eigenproblem (\ref{eq:Brune eigenproblem Reich-Moore}):
\begin{equation}
\boldsymbol{\widetilde{A}}_{\mathrm{R.M.}} \ \triangleq \ \left[\boldsymbol{a}^\dagger\boldsymbol{A}_{\mathrm{R.M.}}^{-1}\boldsymbol{a}\right]^{+}
\label{eq:: pseudo invA Brune generalized to Reich-Moore}
\end{equation}
This generalized definition will guarantee that the Kapur-Peierls operator (\ref{eq:Kapur-Peierls Operator and Channel-Level equivalence}) will we conserved through the following generalization of Brune's relation (\ref{eq:R_L unchanged by Brune}):
\begin{equation}
{\boldsymbol{R}_L}_{\mathrm{R.M.}} = \boldsymbol{\gamma^\mathsf{T} A_{\mathrm{R.M.}} \gamma} = \boldsymbol{\widetilde{\gamma}^\dagger \widetilde{A}_{\mathrm{R.M.}} \widetilde{\gamma}}
\label{eq: Kapur-Peierls operator with generalized Brune parameters}
\end{equation}
thus preserving the scattering matrix (\ref{eq:U expression}) and ultimately the cross section (\ref{eq:partial sigma_cc'from scattering matrix}), as long as we choose more (or equal) solutions to (\ref{eq:Brune eigenproblem Reich-Moore}) than there are levels: $N_S \geq N_\lambda$.

Note that our generalization (\ref{eq: Kapur-Peierls operator with generalized Brune parameters}) does not make the Kapur-Peierls operator Hermitian, since the generalized alternative level $\boldsymbol{\widetilde{A}}_{\mathrm{R.M.}}$ matrix (\ref{eq:: pseudo invA Brune generalized to Reich-Moore}) is still not Hermitian, but complex symmetric. 

\subsection{\label{subsubsec::Necessary choice: how to continue R-matrix operators into the complex plane?}Necessary choice: how to continue R-matrix operators into the complex plane?}

The fact that $\boldsymbol{\widetilde{e}_{\mathrm{R.M.}}} $ is now complex -- complex alternative poles $\widetilde{E_i} \in \mathbb{C}$ and eigenvectors $\boldsymbol{a_i} \in \mathbb{C}^{N_\lambda}$ solve (\ref{eq:Brune eigenproblem Reich-Moore}) -- has profound consequences on the Reich-Moore alternative parameters (\ref{eq:Brune parameters Reich Moore}), because it breaks Brune's three-step monotony argument (last paragraph of section \ref{sec:R_S def}) to prove that here are exactly $N_\lambda$ real solutions on the physical sheet above threshold (we showed there are shadow poles below threshold or in the complex plane in both theorem \ref{theo::branch_brune_poles} and \ref{theo::analytic_Brune_poles}). Indeed, nothing guarantees the three-step monotony argument still stands in the complex plane, when calling the shift operator $\boldsymbol{S}(\widetilde{E_i})$ at complex values $\widetilde{E_i}\in\mathbb{C}$. 
Actually, the choice of convention to continue the R-matrix operators into the complex plane -- that is branch-point definition (\ref{eq:: Def S = Re[L], P = Im[L]}) or analytic continuation (\ref{eq:: Def S and P analytic continuation from L}) -- is now of critical importance, since it will change $\boldsymbol{S}(\widetilde{E_i})$ (for $\widetilde{E_i} \in \mathbb{C}$) and thus the values of the all the Reich-Moore alternative parameters (\ref{eq:Brune parameters Reich Moore}), including the principal poles. 
If we choose analytic continuation definition (\ref{eq:: Def S and P analytic continuation from L}), then theorem \ref{theo::analytic_Brune_poles} still stands and there are $N_S \geq N_\lambda$ (complex) alternative poles as in (\ref{eq:N_S Brune poles}).
However, if we choose branch-point definition (\ref{eq:: Def S = Re[L], P = Im[L]}), then Brune's three-step monotony argument does not stand and we have no guarantee on the number of alternative poles anymore, nor on which sheet of mapping (\ref{eq:rho_c(E) mapping}) the alternative poles reside. 

The only workaround to this is to use the \textit{Generalized Reich-Moore} framework to convert the Reich-Moore parameters into real R-matrix parameters as described in \cite{Generalized_Reich_Moore_2017}; but this would incur a great computational and memory cost as we will have to expand a few eliminated channels R-matrix ($N_c \times N_c$ with $c \not \in \gamma_{\mathrm{elim.}}$) into a square R-matrix of the size of the levels ($N_\lambda \times N_\lambda$), when for large nuclides we often have $N_\lambda \gg N_c$.
And even in the case of Generalized Reich-Moore (which is equivalent to exact R-matrix in that it yields real resonance parameters), the values of the alternative parameters will still depend on the choice of continuation in the complex plane -- branch-point definition (\ref{eq:: Def S = Re[L], P = Im[L]}) versus analytic continuation definition (\ref{eq:: Def S analytic}) -- when there are many different thresholds for different channels, and the $\boldsymbol{S}(E)$ operator must be called below threshold for certain channels when solving (\ref{eq:Brune eigenproblem Reich-Moore}). In fact, the only case where the choice of continuation -- definition (\ref{eq:: Def S = Re[L], P = Im[L]}) versus definition (\ref{eq:: Def S analytic}) -- has no consequence on the values of the principal alternative poles (the Shadow poles always differ) is when we are using the exact R-matrix equations (or the generalized Reich-Moore ones \cite{Generalized_Reich_Moore_2017}) and all the alternative poles are above the thresholds of all the channels. 

In practice, this means that the choice of continuation matters because it changes the values of all the alternative parameters: for Reich-Moore alternative parameters (\ref{eq:Brune parameters Reich Moore}) we need to call the external R-matrix operators ($\boldsymbol{O}, \boldsymbol{I}, \boldsymbol{P}, \boldsymbol{S}$) into the complex plane; and for real R-matrix alternative parameters (\ref{eq:Brune parameters}) the many thresholds will mix up in the sub-threshold (shadow) values of the $S_c(E)$ operator (unless we are only solving past the last threshold). 
Thus, in order to convert nuclear data libraries from Wigner-Eisenbud to alternative parameters, the nuclear scientists community must convene on a convention -- either branch-point definition (\ref{eq:: Def S = Re[L], P = Im[L]}) or analytic continuation definition (\ref{eq:: Def S analytic}) -- to compute R-matrix operators for complex wavenumbers. 
The authors are publishing a follow-up article arguing in favor of analytic continuation \cite{Ducru_Scattering_Matrix_of_Complex_Wavenumbers_2019}.

\section{\label{sec:Evidence of shadow Brune poles in xenon 134} Evidence of shadow alternative poles in xenon $^{\mathrm{134}}\mathrm{Xe}$}

We here report the first evidence of the existence of shadow poles in the alternative parametrization of R-matrix theory, observed in isotope xenon $^{\mathrm{134}}\mathrm{Xe}$ for neutron reactions: $\mathrm{n} + {^{\mathrm{134}}\mathrm{Xe}}$.
In doing so, we also demonstrate that all alternative parameters depend on the convention used for continuation into the complex plane of R-matrix operators. 

We chose xenon $^{\mathrm{134}}\mathrm{Xe}$ because it has only a few resonances per spin group, this makes it a clear case that is simple to solve numerically. 
Xenon $^{\mathrm{134}}\mathrm{Xe}$ is stable and the fourth most abundant isotope of xenon (10.436\% of natural content, most abundant is ${^{\mathrm{132}}\mathrm{Xe}}$ with 26.909\%). 
The isotope spin is $0^{(+)}$, and the neutron's $1/2^{(+)}$. There are three spin groups: $J^\pi = 1/2^{(+)}$ with 3 s-wave resonances; $J^\pi = 1/2^{(-)}$ with 2 p-wave resonances; and $J^\pi = 3/2^{(-)}$ with 1 p-wave resonance.
The R-matrix parameters of xenon $^{\mathrm{134}}\mathrm{Xe}$, here reported in table \ref{tab:resonance parameters}, were taken from ENDF/B-VIII.0 nuclear data library \cite{ENDFBVIII8th2018brown}, where we observe the two p-waves in the $J^\pi = 1/2^{(-)}$ spin group.
\begin{table}[b]
\caption{\label{tab:resonance parameters}
Xenon $^{\mathrm{134}}\mathrm{Xe}$ resonance parameters for the two p-waves of spin group $J^\pi = 1/2^{(-)}$, from ENDF/B-VIII.0 evaluation}
\begin{ruledtabular}
\begin{tabular}{l}
$ z = \sqrt{E}$ with $E$ in (eV) \\
$A = 132.760$ \\
$a_c = 5.80$ : channel radius (Fermis) \\
$\rho(z) = \frac{A a_c \sqrt{\frac{2m_n}{h}} }{A + 1} z$ \\
with $\sqrt{\frac{2m_n}{h}} = 0.002196807122623 $ in units ($1/(10^{-14}\text{m} \sqrt{\text{eV}})$) \\
$E_1 = 2186.0$ : first resonance energy (eV) \\
$\Gamma_{1,n} = 0.2600$ : neutron width of first resonance \\ (not reduced width), i.e. $\Gamma_{\lambda,c} = 2P_c(E_\lambda) \gamma_{\lambda,c}^2$ \\
$\Gamma_{1,\gamma} = 0.0780$ : eliminated capture width (eV) \\
$E_2 = 6315.0$ : second resonance energy (eV) \\
$\Gamma_{2,n} = 0.4000$ (eV) \\
$\Gamma_{2,\gamma} = 0.0780$ (eV) \\
$g_{J^\pi} = 1/3$ : spin statistical factor \\
$B_c = - 1 $
\end{tabular}
\end{ruledtabular}
\end{table}
The xenon $^{\mathrm{134}}\mathrm{Xe}$ ENDF/B-VIII.0 evaluation is listed as a MLBW (Multi-Level Breit-Wigner) with B=S approximation, which means that the exact R-matrix equations are not used (neither the Reich-Moore ones), but instead the physically incorrect approximation that $S_c(E) = B_c$ is constant is made (i.e. the shift function is forced onto the boundary parameters, to simplify the evaluation process).
All ENDF/B-VIII.0 evaluations with only a few resonances are carried out with this B=S approximation.
Though this has no incidence on s-waves ($S_{\ell=0} = 0$) for neutral channels, and in general the B=S approximation has only small effects in practice on the evaluation, these equations cannot rigorously match the R-matrix-equivalent formalisms we here derive.

\begin{table}[b]
\caption{\label{tab:Brune parameters R-matrix}
Xenon $^{\mathrm{134}}\mathrm{Xe}$ alternative parameters for spin-parity group $J^\pi = 1/2^{(-)}$. For a verisimilar fictitious isotope, all capture widths (including eliminated channels) are set to zero, and the p-waves are converted using ENDF/B-VIII.0 resonance parameters into R-matrix equations, and solving the Brune eigenproblem (\ref{eq:Brune eigenproblem}) as detailed in section \ref{sec:R_S def}, for both conventions to continue the shift function to complex wavenumbers: Lane \& Thomas (\ref{eq:: Def S = Re[L], P = Im[L]}) versus analytic continuation (\ref{eq:: Def S analytic}). (Results to two significant digits or discrepancy)}
\begin{ruledtabular}
\begin{tabular}{l}
\textsc{R-matrix alternative parameters} (\ref{eq:Brune eigenproblem})\\
 \textsc{(verisimilar fictitious isotope)}\\
\hline
\hline
Lane \& Thomas definition (\ref{eq:: Def S = Re[L], P = Im[L]})\\
\hline
Alternative branch poles (eV) and their sheet of mapping (\ref{fig:mapping rho - E}):\\
$\begin{array}{lcl}
    \left\{\widetilde{E_1}, - \right\} = -626,938 &  \quad ; \quad & \widetilde{z_1} = - 791.794 \, \mathrm{i} \\
     \left\{\widetilde{E_2}, + \right\} = 2,183 &  \quad ; \quad & \widetilde{z_2} = 46.73 \\
     \left\{\widetilde{E_3}, + \right\} = 6,313 &  \quad ; \quad  & \widetilde{z_3} = 79.454
\end{array}$ \\
Corresponding eigenvectors:\\
$\boldsymbol{a_1} = \left[0.8732734 \; \; , \; \; 0.4872304\right]^\mathsf{T}$\\
$\boldsymbol{a_2} = \left[1.0 \; \; , \; \; 2.9864\times 10^{-4}\right]^\mathsf{T}$\\
$\boldsymbol{a_3} = \left[ - 8.5708\times 10^{-4} \; \; , \; \; 1.0\right]^\mathsf{T}$\\
\hline 
Analytic continuation definition (\ref{eq:: Def S analytic})\\
\hline
Alternative analytic poles (eV) and their sheet of mapping (\ref{fig:mapping rho - E}):\\
$\begin{array}{lcl}
    \left\{\widetilde{E_1}, \pm \right\} = -626,111 &  \quad ; \quad & \widetilde{z_1} = 791.271 \, \mathrm{i} \\
     \left\{\widetilde{E_2}, \pm \right\} = 2,183 &  \quad ; \quad & \widetilde{z_2} = 46.73 \\
     \left\{\widetilde{E_3}, \pm \right\} = 6,313 &  \quad ; \quad  & \widetilde{z_3} = 79.454
\end{array}$ \\
Corresponding eigenvectors:\\
$\boldsymbol{a_1} = \left[0.8732752 \; \; , \; \; 0.4872272\right]^\mathsf{T}$\\
$\boldsymbol{a_2} = \left[1.0 \; \; , \; \; 2.9864\times 10^{-4}\right]^\mathsf{T}$\\
$\boldsymbol{a_3} = \left[ -8.5708\times 10^{-4} \; \; , \; \; 1.0\right]^\mathsf{T}$\\
\end{tabular}
\end{ruledtabular}
\end{table}
To validate theorems \ref{theo::branch_brune_poles} and \ref{theo::analytic_Brune_poles}, we first create a verisimilar fictitious single-channel xenon $^{\mathrm{134}}\mathrm{Xe}$ isotope in R-matrix formalism (instead of MLBW), by setting all the capture widths (explicit $\gamma$ or eliminated capture) to zero, and treating the resulting purely scattering system with R-matrix equations -- i.e. (\ref{eq:U expression}) and (\ref{eq:R expression}).
We then convert these Wigner-Eisenbud R-matrix parameters into alternative parameters by solving the generalized eigenvalue system (\ref{eq:Brune eigenproblem}), and report the results in table \ref{tab:Brune parameters R-matrix}. 
The alternative poles reported in table \ref{tab:Brune parameters R-matrix} exhibit all the behaviors proved in theorems \ref{theo::branch_brune_poles} and \ref{theo::analytic_Brune_poles}. 
As in theorem \ref{theo::branch_brune_poles}, the alternative branch poles -- i.e. found using the Lane \& Thomas definition (\ref{eq:: Def S = Re[L], P = Im[L]}) -- are all real and count $N_\lambda = 2$ principal poles on the $\left\{E,+\right\}$ sheet of mapping (\ref{fig:mapping rho - E}), near the resonances, as well as one shadow alternative branch pole on the $\left\{E,-\right\}$ sheet bellow threshold. 
Meanwhile, as proved in theorem \ref{theo::analytic_Brune_poles}, there are three (from (\ref{eq:N_S Brune poles}) we have $N_S = 2 + 1 $) alternative analytic poles -- i.e. using the analytic continuation definition (\ref{eq:: Def S analytic}). Two (the `principal' ones) are real (because $N_\lambda = 2$), and the last one (the `shadow alternative analytic pole') is sub-threshold and also happens to be real because $\ell_c = 1$ is an odd number (c.f. theorem \ref{theo::analytic_Brune_poles}). Again, since definition (\ref{eq:: Def S and P analytic continuation from L}) unfolds mapping (\ref{eq:rho_c(E) mapping}), the alternative analytic poles have no multi-sheeted structure (which we made explicit by stating both $\left\{E,\pm\right\}$ sheets).

To validate our generalization to Reich-Moore, established in section \ref{subsubsec::Generalization to Reich-Moore approximation and Teichmann-Wigner eliminated channels}, we proceed just as we did with the fictitious R-matrix single-channel xenon $^{\mathrm{134}}\mathrm{Xe}$ isotope, and convert the ENDF/B-VIII.0 resonance parameters into alternative parameters by solving the Brune-generalized-to-Reich-Moore eigenproblem (\ref{eq:Brune eigenproblem Reich-Moore}).
The results are reported in table \ref{tab:Brune parameters Reich-Moore}. 
\begin{table}[b]
\caption{\label{tab:Brune parameters Reich-Moore}
Xenon $^{\mathrm{134}}\mathrm{Xe}$ alternative parameters for spin-parity group $J^\pi = 1/2^{(-)}$. The p-waves are converted using ENDF/B-VIII.0 resonance parameters into Reich-Moore equations, and solving the generalized eigenproblem (\ref{eq:Brune eigenproblem Reich-Moore}) as detailed in section \ref{subsubsec::Generalization to Reich-Moore approximation and Teichmann-Wigner eliminated channels}, for both conventions to continue the shift function to complex wavenumbers: Lane \& Thomas (\ref{eq:: Def S = Re[L], P = Im[L]}) versus analytic continuation (\ref{eq:: Def S analytic}).
(Results to two significant digits or discrepancy)}
\begin{ruledtabular}
\begin{tabular}{l}
\textsc{Reich-Moore alternative parameters} (\ref{eq:Brune eigenproblem Reich-Moore})\\
\hline
\hline
Lane \& Thomas definition (\ref{eq:: Def S = Re[L], P = Im[L]})\\
\hline
Alternative branch poles (eV) and their sheet of mapping (\ref{fig:mapping rho - E}):\\
$\begin{array}{rcl}
    \left\{\widetilde{E_1}, - \right\} & = & - 626938 -   0.039 \, \mathrm{i} \\
    \widetilde{z_1} & = & 2.462\times10^{-5} - 791.794  \, \mathrm{i} \\
     \left\{\widetilde{E_2}, - \right\} & = & 2183.8031735 - 0.039 \, \mathrm{i} \\
     \widetilde{z_2} & = & 46.73117988 - 4.172\times10^{-4} \, \mathrm{i} \\
     \left\{\widetilde{E_3}, - \right\} & = & 6313.013519 - 0.039 \, \mathrm{i} \\ \widetilde{z_3} & = & 79.454474511 - 2.4542\times10^{-4} \, \mathrm{i} 
\end{array}$ \\
Corresponding eigenvectors:\\
$\boldsymbol{a_1} = \left[0.8732734 \; \; , \; \; 0.487230\right]^\mathsf{T}$\\
$\boldsymbol{a_2} = \left[1.0 \; \; , \; \; 2.9863794\times 10^{-4}\right]^\mathsf{T}$\\
$\boldsymbol{a_3} = \left[ - 8.570833\times 10^{-4} \; \; , \; \; 1.0\right]^\mathsf{T}$\\
\hline 
Analytic continuation definition (\ref{eq:: Def S analytic})\\
\hline
Alternative analytic poles (eV) and their sheet of mapping (\ref{fig:mapping rho - E}):\\
$\begin{array}{rcl}
    \left\{\widetilde{E_1}, \pm \right\} & = & - 626111 - 5.119\times 10^{-5} \, \mathrm{i} \\ 
    \widetilde{z_1} & = & 3.234\times10^{-8} - 791.271  \, \mathrm{i} \\
     \left\{\widetilde{E_2}, \pm \right\} & = &  2183.8031770 - 0.03896 \, \mathrm{i} \\ \widetilde{z_2} & = & 46.73117992 - 4.168\times10^{-4} \, \mathrm{i} \\
     \left\{\widetilde{E_3}, \pm \right\} & = & 6313.013521 - 0.03898 \, \mathrm{i} \\ \widetilde{z_3} & = & 79.454474522 - 2.4534\times10^{-4} \, \mathrm{i} 
\end{array}$ \\
Corresponding eigenvectors:\\
$\boldsymbol{a_1} = \left[0.8732752 + 3.5344\times 10^{-10} \, \mathrm{i} \; \; , \; \; 0.487227 \right]^\mathsf{T}$\\
$\boldsymbol{a_2} = \left[1.0 + 1.7772\times 10^{-5} \, \mathrm{i} \; \; , \; \; 2.9863747\times 10^{-4}  \right]^\mathsf{T}$\\
$\boldsymbol{a_3} = \left[ -8.570825\times 10^{-4} + 5.2348\times 10^{-9} \, \mathrm{i} \; \; , \; \; 1.0\right]^\mathsf{T}$\\
\end{tabular}
\end{ruledtabular}
\end{table}
The Reich-Moore generalized alternative parameters in table \ref{tab:Brune parameters Reich-Moore} also inherit most of the results from theorems \ref{theo::branch_brune_poles} and \ref{theo::analytic_Brune_poles}. There are some notable differences however. Generalizing to Reich-Moore entails all the alternative poles are now complex, regardless of which definition (\ref{eq:: Def S = Re[L], P = Im[L]}) or (\ref{eq:: Def S analytic}) is chosen to continue the shift function $S_c(k_c)$ to complex wavenumbers $k\in\mathbb{C}$.\\
This has major consequences when choosing the Lane \& Thomas definition (\ref{eq:: Def S = Re[L], P = Im[L]}): unlike in the R-matrix case, in Reich-Moore the $N_\lambda$ principal poles are no longer on the physical sheet $\left\{E,+\right\}$. Indeed, we observe in table \ref{tab:Brune parameters Reich-Moore} that in our case all the alternative branch poles are now on the non-physical sheet $\left\{E,-\right\}$. In the general case, the alternative branch poles could be on the physical or the non-physical sheet (we have no proof for either), and one could thus say that all alternative poles are shadow poles in the Reich-Moore formalism.
This lack of knowledge of on what sheet to find the alternative branch poles comes atop the fact, discussed in section \ref{subsubsec::Necessary choice: how to continue R-matrix operators into the complex plane?}, that Brune's three-step monotony argument (which proved the existence of exactly $N_\lambda$ real alternative poles above threshold) is only valid for real-symmetric matrices. When generalized to Reich-Moore, eigenproblem (\ref{eq:Brune eigenproblem Reich-Moore}) counts a complex-symmetric matrix, entailing Brune's three-step monotony argument at the core of theorem \ref{theo::branch_brune_poles} is no longer valid, and we actually do not have proof of the number of alternative branch poles (theorem \ref{theo::branch_brune_poles} proof only stands for R-matrix, or generalized Reich-Moore). \\
This is not the case for the alternative analytic poles, which generalize quite naturally to Reich-Moore formalism. In fact, the only difference to theorem \ref{theo::analytic_Brune_poles} is that the three-step monotony argument can no longer be used to prove that $N_\lambda$ of the $N_S = N_\lambda + \sum_c \ell_c$ alternative analytic poles (\ref{eq:N_S Brune poles}) are real -- and indeed they are not, as shown in table \ref{tab:Brune parameters Reich-Moore}. Apart from that, theorem \ref{theo::analytic_Brune_poles} remains intact for our generalization to Reich-Moore established in section \ref{subsubsec::Generalization to Reich-Moore approximation and Teichmann-Wigner eliminated channels}: there are still $N_S = N_\lambda + \sum_c \ell_c$ alternative analytic poles (\ref{eq:N_S Brune poles}), and there is no need to specify on which $\left\{E,\pm \right\}$ sheet of mapping (\ref{eq:rho_c(E) mapping}) they reside since the analytic continuation of the shift function $S_c(\rho_c)$ unfolds the mapping (c.f. lemma \ref{lem:: analytic S_c and P_c lemma} and theorem \ref{theo::analytic_Brune_poles}). \\

Take a closer observation at the results in tables \ref{tab:Brune parameters R-matrix} and \ref{tab:Brune parameters Reich-Moore}. 
Notice that the imaginary part of the alternative branch poles are all equal to -0.039, which is exactly the opposite of half the eliminated channel width: convenient. 
This can readily be explained by splitting the generalized-to-Reich-Moore Brune-eigenproblem~(\ref{eq:Brune eigenproblem Reich-Moore}) into real and imaginary parts, and noticing that if all the eliminated channel widths are the same, then the eigenvalue's (alternative pole) imaginary part is exactly opposite to the eliminated channel width divided by two, i.e. if $\forall \lambda,\lambda' \in J^\pi , \; \Gamma_{\lambda, \gamma} = \Gamma_{\lambda' ,\gamma}$ then $\forall j \, ,  \; \Im\left[\widetilde{E_j}\right] = - \frac{\Gamma_{\lambda,\gamma}}{2}$, from (\ref{eq:e diagonal matrix Reich-Moore}).
It so happens that in our particular case of xenon $^{\mathrm{134}}\mathrm{Xe}$ this is indeed true, all eliminated capture widths are equal to 0.078 (c.f. table \ref{tab:resonance parameters}). Looking at the ENDF/B-VIII.0 library it is surprisingly common to have the same eliminated capture widths within the same spin group (both xenon-132 and xenon $^{\mathrm{134}}\mathrm{Xe}$ are such examples).
But this is of course not true in general, and a quick look at uranium-238 will show that different levels have different eliminated capture widths (i.e. $\exists \lambda,\lambda' \in J^\pi , \; \; \Gamma_{\lambda, \gamma} \neq \Gamma_{\lambda' ,\gamma}$).
So when the Lane \& Thomas convention (\ref{eq:: Def S = Re[L], P = Im[L]}) is chosen, the alternative branch poles imaginary part will in general not coincide with the eliminated capture widths: $\Im\left[\widetilde{E_j}\right] \neq - \frac{\Gamma_{\lambda,\gamma}}{2}$.
Similarly, neither will the alternative branch eigenvectors be real in general.\\

Note that the eigenvectors in tables \ref{tab:Brune parameters R-matrix} and \ref{tab:Brune parameters Reich-Moore} are close, but differ when going to higher digits precision, and those small differences have strong impact on the cross section calculation. This leads us to discuss the numerical methods employed to solve the generalized eigenproblem (\ref{eq:Brune eigenproblem Reich-Moore}), which need to be solved in wavenumber space $k_c$ (we here use the variable $z = \sqrt{E}$) to properly describe the multi-sheeted nature of mapping (\ref{eq:rho_c(E) mapping}).
For the fictitious R-matrix problem (dealing with only one channel and real values), we coded the analytic continuation of the $S_c(\rho_c)$ shift function (c.f. table \ref{tab::S_and_P_expressions_neutral}), and used the built-in MATLAB polynomial rootfinder to solve (\ref{eq:Brune eigenproblem Reich-Moore}), verifying that the results were indeed roots.
For the branch-point definition (\ref{eq:: Def S = Re[L], P = Im[L]}), we used a built-in MATLAB numerical solver for equations of the type $f(x) = 0$ on the determinant of the left-hand side of (\ref{eq:Brune eigenproblem}), and solved the roots one-by-one.
For the generalized Reich-Moore Brune eigenproblem (\ref{eq:Brune eigenproblem Reich-Moore}), the alternative analytic poles are readily found in the complex plane with the same polynomial rootfinder (we discuss methods to solve for all the roots of a polynomial simultaneously in \cite{Ducru_PHYSOR_conversion_2016} and \cite{Analytic_Benchmark_1_2019}).
Finding the alternative branch poles is much more complicated: the built-in $f(x) = 0$ MATLAB solver finds the two principal poles (on the non-physical sheet $\left\{E,-\right\}$ this time), but to find the shadow pole we had to devise a procedure manually: from the solution when the eliminated capture width is zero (R-matrix case), we zoom-in in the region around that solution and build a convex bowl around it and then slowly increase the eliminated capture width from zero. For each capture width value we did a minimization on the norm of the determinant to find the updated alternative pole for an in-between value of the eliminated capture width. We then iteratively re-solve, re-do a new complex bowl, augment the eliminated capture width, until we converge on the shadow alternative branch pole.
This cumbersome procedure points to the mathematical advantages of analytic continuation definition (\ref{eq:: Def S and P analytic continuation from L}) as it conserves smooth analytic properties of the Kapur-Peierls operator into the complex plane, which greatly simplifies the conversion to alternative poles for Reich-Moore evaluations. \\

Finally, to validate theorem \ref{theo::Choice of Brune poles}, as well as the entire generalization to Reich-Moore formalism we establish in section \ref{subsubsec::Generalization to Reich-Moore approximation and Teichmann-Wigner eliminated channels}, we construct the corresponding cross sections using the xenon $^{\mathrm{134}}\mathrm{Xe}$ resonance parameters from ENDF/B-VIII.0 with the exact R-matrix and Reich-Moore equations -- i.e. (\ref{eq:U expression}) and (\ref{eq:R expression}) to compute (\ref{eq:partial sigma_cc'from scattering matrix}).
The resulting point-wise cross section values are plotted in figure \ref{fig:xenon-134 J=1/2(-) cross section}, and their peak resonance values are provided for reference in table \ref{tab:Cross_sections_peak_values}. 
These cross sections do not exactly coincide with the point-wise evaluation values from ENDF/B-VIII.0, since ENDF uses the (coarser) MLBW equations instead of the Reich-Moore (or R-matrix) ones.
\begin{figure}[ht!!] 
  \centering
  \subfigure[\ First p-wave resonance.]{\includegraphics[width=0.49\textwidth]{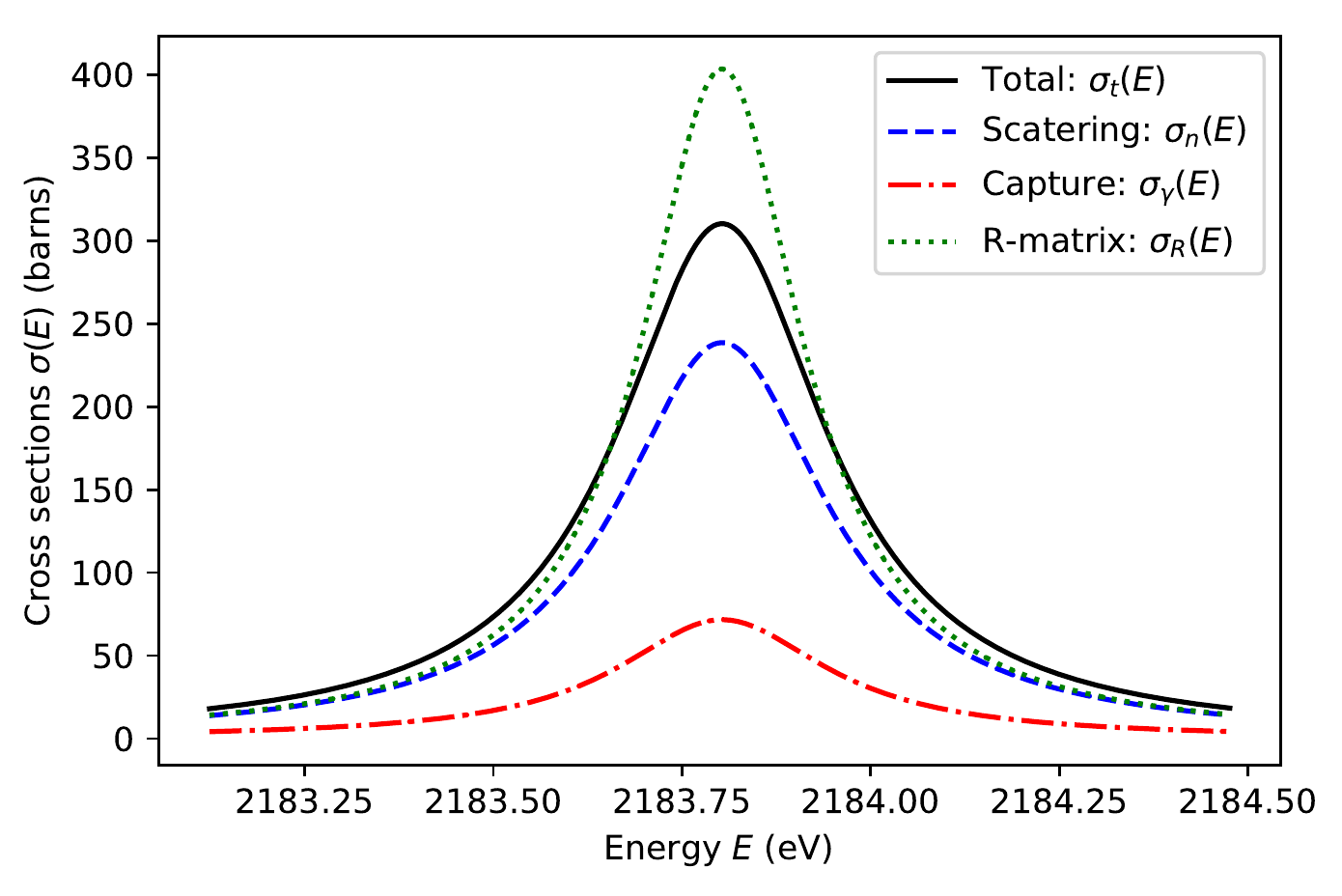}}
  \subfigure[\ Second p-wave resonance.]{\includegraphics[width=0.49\textwidth]{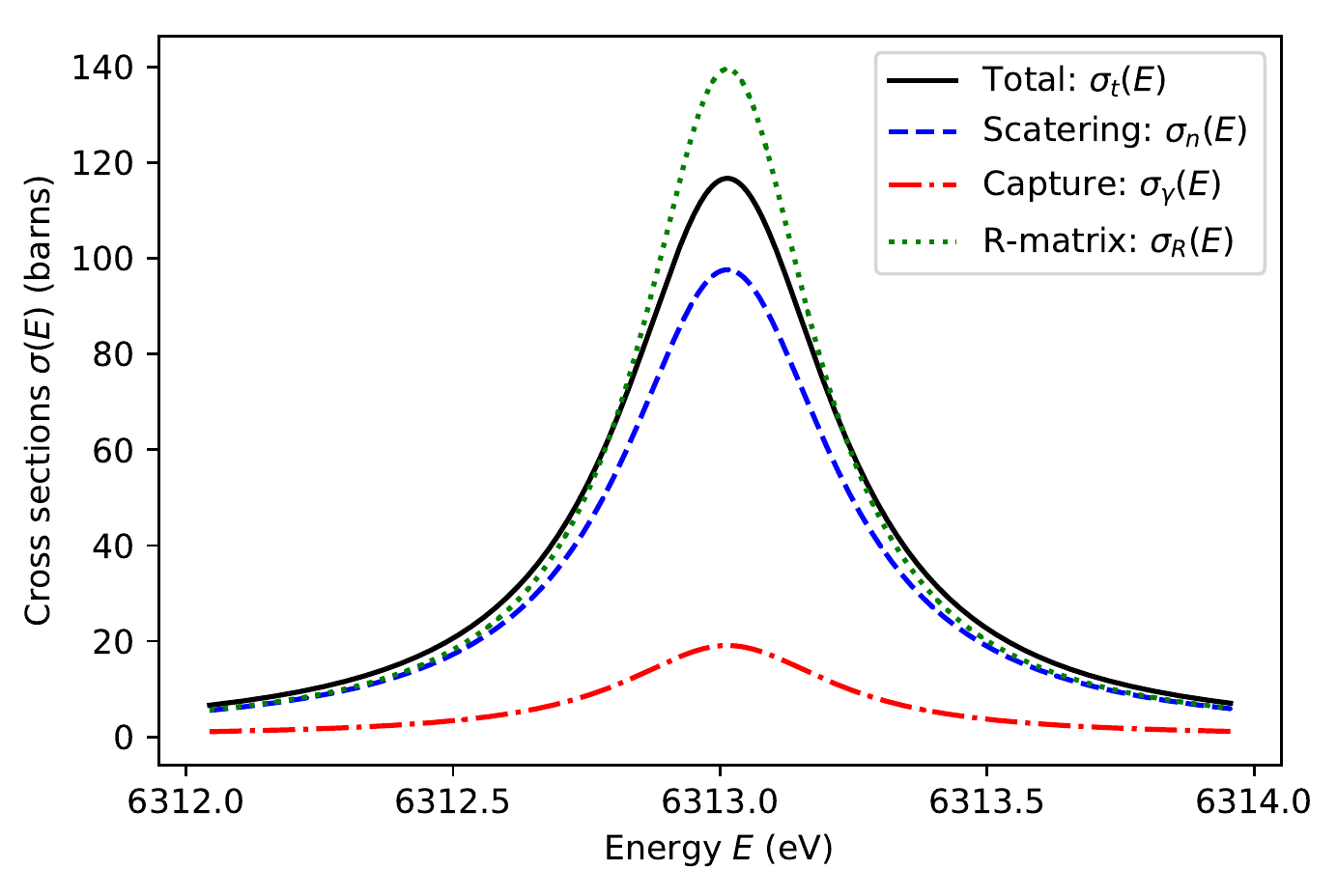}}
  \caption{\small{Xenon $^{\mathrm{134}}\mathrm{Xe}$ Reich-Moore cross sections for spin-parity group $J^\pi = 1/2^{(-)}$ p-wave resonances: the cross sections are generated using the ENDF/B-VIII.0 resonance parameters (a MLBW evaluation) into the Reich-Moore formalism equations.
  Similarly, the R-matrix cross section is generated by setting all capture (including eliminated) widths to zero. For reference, the exact cross section values at the resonance peaks are reported in table \ref{tab:Cross_sections_peak_values}.
  }}
  \label{fig:xenon-134 J=1/2(-) cross section}
\end{figure}

\begin{table}[b]
\caption{\label{tab:Cross_sections_peak_values}
\textsc{Resonance peaks}: Xenon $^{\mathrm{134}}\mathrm{Xe}$ spin-parity group $J^\pi = 1/2^{(-)}$ two p-waves cross section values at the peak of the resonances (truncated to 4 digits accuracy). The cross sections are generated using the ENDF/B-VIII.0 resonance parameters (a MLBW evaluation) into the Reich-Moore formalism equations. Similarly, the R-matrix cross section is generated by setting all capture (including eliminated) widths to zero.}
\begin{ruledtabular}
\begin{tabular}{l|l|l}
Energy (eV)  &  2183.8030 & 6313.0138 \\ \hline
Total cross-section (barns) &  310.2761 & 116.7259 \\ \hline
Scattering cross-section (barns) &  238.6095 & 97.6224 \\ \hline
Capture cross-section (barns) &  71.6666 & 19.1035 \\ \hline
R-matrix cross-section (barns) & 403.4677 & 139.5677 
\end{tabular}
\end{ruledtabular}
\end{table}
We then compute the cross section using the alternative parameters reported in tables \ref{tab:Brune parameters R-matrix} and \ref{tab:Brune parameters Reich-Moore} and following the procedure established in section \ref{subsubsec::Generalization to Reich-Moore approximation and Teichmann-Wigner eliminated channels} to reconstruct the Kapur-Peierls operator (\ref{eq: Kapur-Peierls operator with generalized Brune parameters}), necessary for computing the scattering matrix (\ref{eq:U expression}), and ultimately the cross section (\ref{eq:partial sigma_cc'from scattering matrix}).
We can now observe the alternative parametrization yield the exact same cross section as the R-matrix (or the Reich-Moore) parametrizaton, for both the Lane \& Thomas (\ref{eq:: Def S = Re[L], P = Im[L]}) or the analytic continuation (\ref{eq:: Def S analytic}) conventions, and by choosing any subset of at least $N_\lambda$ alternative poles: including discarding the principal poles and instead using the shadow poles. This result validates theorem \ref{theo::Choice of Brune poles} and its generalization to Reich-Moore (\ref{eq:: pseudo invA Brune generalized to Reich-Moore}), and will come as quite striking to some evaluators: for p-waves (or higher angular momentum) one can choose to discard the principal alternative pole directly close to the resonance and instead use the shadow pole, which is far below the threshold and into the complex plane, or even use both principal and shadow alternative poles (using the generalized inverse definition (\ref{eq:: pseudo invA Brune generalized to Reich-Moore}) and the procedure detailed in section \ref{subsubsec::Generalization to Reich-Moore approximation and Teichmann-Wigner eliminated channels}), to produce the exact same cross section resonance behavior.

\section{\label{sec:Conclusion}Conclusion}

This article establishes the existence of shadow poles in the alternative parametrization of R-matrix theory.
This parametrization is being considered as an alternative to the traditional Wigner-Eisenbud resonance parameters to document nuclear cross section values in the nuclear data libraries.

The Wigner-Eisenbud parameters are the poles $\big\{E_\lambda\big\}$ and residue widths $\big\{\gamma_{\lambda,c}\big\}$ of the $\boldsymbol{R}$ matrix (\ref{eq:R expression}). They are $N_\lambda \in \mathbb{N}$ real poles, which are independent from one another (meaning any choice of real parameters are physically acceptable), and de-entangle the energy dependence of the $\boldsymbol{R}$ matrix from the branch-points the thresholds $\left\{E_{T_c}\right\}$ introduce in the multi-sheeted Riemann surface of mapping (\ref{eq:rho_c(E) mapping}).
Both $\big\{E_\lambda\big\}$ and $\big\{\gamma_{\lambda,c}\big\}$ are dependent on both the channel radii $\big\{a_c \big\}$ and the boundary conditions $\big\{ B_c\big\}$.
The set of Wigner-Eisenbud parameters $\Big\{ E_{T_c}, a_c, B_c, E_{\lambda}, \gamma_{\lambda,c} \Big\}$ is sufficient to fully describe the scattering matrix $\boldsymbol{U}(E)$ energy dependence (\ref{eq:U expression}).

The alternative parameters are the poles $\left\{\widetilde{E_i}\right\}$ of the $\boldsymbol{R}_S$ matrix (\ref{eq:R_S by Brune det search}) and the widths $\left\{\widetilde{\gamma_{i,c}}\right\}$, transformed by (\ref{eq:Brune parameters}) from the residue widths of the alternative level matrix $\boldsymbol{\widetilde{A}}$ in (\ref{eq::Brune physical level matrix}) and (\ref{eq:Brune eigenproblem}). They are $N_S^{\pm} \geq N_\lambda$ poles, and are intimately interdependent in that not any set of real parameters is physically acceptable (they must be solutions of  (\ref{eq:Brune eigenproblem})).
If the legacy Lane \& Thomas definition (\ref{eq:: Def S = Re[L], P = Im[L]}) is chosen for the shift function $\boldsymbol{S}$, the alternative branch poles live on the multi-sheeted Riemann surface of mapping (\ref{eq:rho_c(E) mapping}): they have branch shadow poles $\left\{\widetilde{E_i}\right\}$ on the unphysical sheets $\left\{ E, -\right\}$ below threshold $ E < E_{T_c}$, though there are only $N_\lambda$ real poles on the physical sheet (theorem \ref{theo::branch_brune_poles}).
If analytic continuation definition (\ref{eq:: Def S analytic}) is chosen, then the shift factor $\boldsymbol{S}$ is a function of $\rho_c^2$, which unfolds the sheets in mapping (\ref{eq:rho_c(E) mapping}): there are then $N_S^{\mathbb{C}} \geq N_\lambda$ analytic poles $\left\{\widetilde{E_i}\right\}$, in general complex (though for R-matrix at least $N_\lambda$ of them are real), all living on the same sheet with no branch points (theorem \ref{theo::analytic_Brune_poles}).
Both $\left\{\widetilde{E_i}\right\}$ and $\left\{\widetilde{\gamma_{i,c}}\right\}$ are invariant to change in boundary conditions $\big\{B_c\big\}$, though both depend on the channel radii $\big\{a_c\big\}$.
Any subset of $N_\lambda$ or more alternative parameters $\Big\{ E_{T_c}, a_c, \widetilde{E_{i}}, \widetilde{\gamma_{i,c}} \Big\}$ is sufficient to entirely determine the energy behavior of the scattering matrix $\boldsymbol{U}$ through (\ref{eq:R_L unchanged by Brune}) and (\ref{eq:U expression}) (theorem \ref{theo::Choice of Brune poles}).

The first shadow alternative poles are observed in xenon isotope $\mathrm{^{134}_{\; \, 54} Xe}$ spin-parity group $J^\pi = 1/2^{(-)}$, which has two p-wave resonance.
We show how the shadow alternative poles can be chosen instead of the traditional principal alternative poles to compute the cross sections.
We also demonstrate that any subset of $N_\lambda$ alternative poles will also reconstruct the cross section. Since there are $N_\lambda$ principal (resonant) alternative poles, this means that the shadow poles can be discarded from future nuclear data libraries without compromising their capability to fully reconstruct R-matrix cross sections (i.e. entirely describe their energy dependence).

In order to convert the xenon resonance parameters, we generalize the alternative parameters to deal with the Reich-Moore approximation and the additional shadow poles. 
The Reich-Moore approximation -- widely used in nuclear data libraries -- introduces complex Reich Moore alternative parameters (\ref{eq:Brune parameters Reich Moore}), and their values depend on which convention -- analytic continuation definition (\ref{eq:: Def S analytic}) versus branch-point definition (\ref{eq:: Def S = Re[L], P = Im[L]}) -- is chosen to continue the R-matrix operators to complex wavenumbers. 
Deciding on this convention is thus a necessary prerequisite to converting nuclear data libraries to alternative parameters. For mathematical and physical reasons, we argue in favor of analytic continuation in a follow-up article \cite{Ducru_Scattering_Matrix_of_Complex_Wavenumbers_2019}.

\begin{acknowledgments}
This work was partly funded by the Los Alamos National Laboratory (summer 2017 research position in T-2 division with G. Hale and M. Paris), as well as by the Consortium for Advanced Simulation of Light Water Reactors (CASL), an Energy Innovation Hub for Modeling and Simulation of Nuclear Reactors under U.S. Department of Energy Contract No. DE-AC05-00OR22725. In addition to a Research Assistantship from the Massachusetts Institute of Technology, U.S.A., the first author was also supported as AXA Fellow of the Schwarzman Scholars Program, Tsinghua University, Beijing, China.

We would like to thank and acknowledge Yoann Desmouceaux, author of the proof of diagonal divisibility and capped multiplicities lemma \ref{lem::diagonal divisibility and capped multiplicities}, who's help and inputs were key on technical algebraic points.
Our genuine gratitude to Dr. Andrew Holcomb, from Oak-Ridge National Laboratory, who helped us numerically test the veracity of Mittag-Leffler expansion (\ref{eq::Explicit Mittag-Leffler expansion of L_c}).
We would also like to thank Prof. Javier Sesma, from Universidad de Zaragoza, for his invaluable guidance on the properties of the Hankel functions, as well as Haile Owusu for his insight into Hamiltonian degeneracy.
Finally, we are grateful towards the organizers of the R-matrix workshop of summer 2016 in Santa Fe, New Mexico \cite{website:2016RMatrixWorkshop}, which was genuinely catalytic to these findings.

\end{acknowledgments}

\bibliography{Shadow_poles_in_the_alternative_parametrization_of_R_matrix_theory}
\end{document}